\emergencystretch=\maxdimen
\documentclass[12pt,a4paper]{article} 
\usepackage{color}
\usepackage[utf8]{inputenc} 
\usepackage[english]{babel}
\usepackage{mathtools,amsmath,amsfonts,amsthm,amssymb} 
\numberwithin{equation}{section} 
\usepackage{authblk} 
\usepackage{indentfirst} 
\usepackage[margin=2cm]{geometry} 
\usepackage{tocloft}
\usepackage{hyperref}
\hypersetup{
    colorlinks,
    citecolor=black,
    filecolor=black,
    linkcolor=black,
    urlcolor=black
}

\newtheorem{theorem}{Theorem}
\newtheorem{lemma}[theorem]{Lemma}
\newtheorem{proposition}[theorem]{Proposition}
\newtheorem{corollary}[theorem]{Corollary}
\newtheorem{remark}[theorem]{Remark}
\newtheorem{question}{Question}

\newcommand{\naw}[1]{\left( {#1} \right)}
\newcommand{\pp}{\mathrm{p}}
\newcommand{\pre}{\mathrm{pre}}
\newcommand{\Dom}{\mathrm{Dom}}
\newcommand{\Ker}{\mathrm{Ker}}
\newcommand{\Ran}{\mathrm{Ran}}
\newcommand{\Num}{\mathrm{Num}}
\newcommand{\sgn}{\mathrm{sgn}}
\newcommand{\supp}{\mathrm{supp}}
\newcommand{\RE}{\mathrm{Re}}
\newcommand{\IM}{\mathrm{Im}}
\newcommand{\R}{\mathbb{R}}
\newcommand{\C}{\mathbb{C}}
\newcommand{\Z}{\mathbb{Z}}
\newcommand{\N}{\mathbb{N}}
\newcommand{\D}{\mathrm{d}}
\renewcommand{\i}{\mathrm{i}}
\newcommand{\e}{\mathrm{e}}
\newcommand{\T}{\mathrm{T}}

\newcommand{\cI}{\mathcal{I}}
\newcommand{\cK}{\mathcal{K}}
\newcommand{\cV}{\mathcal{V}}
\newcommand{\cU}{\mathcal{U}}
\newcommand{\cW}{\mathcal{W}}
\newcommand{\cH}{\mathcal{H}}
\newcommand{\cJ}{\mathcal{J}}
\newcommand{\cM}{\mathcal{M}}
\newcommand{\cE}{\mathcal{E}}
\newcommand{\cZ}{\mathcal{Z}}
\newcommand{\Wr}{\mathcal{W}}
\newcommand{\CP}{\mathbb{CP}}

\newcommand{\ess}{\mathrm{ess}}
\newcommand{\tperp}{\mathrm{perp}}
\newcommand{\loc}{\mathrm{loc}}
\newcommand{\Cl}{\mathrm{Cl}}

\newcommand{\SO}{\mathrm{SO}}
\newcommand{\Spin}{\mathrm{Spin}}
\renewcommand{\S}{\mathbb S}
\newcommand{\vol}{\mathrm{vol}}
\newcommand{\hol}{\mathrm{hol}}
\newcommand{\mix}{\mathrm{mix}}

\def\bbbone{{\mathchoice {\rm 1\mskip-4mu l} {\rm 1\mskip-4mu l}
{\rm 1\mskip-4.5mu l} {\rm 1\mskip-5mu l}}}
\def\one{\bbbone}

\begin{document}

\title{
Holomorphic family of Dirac-Coulomb \\
Hamiltonians  in arbitrary dimension}

\author[1]{Jan Derezi\'{n}ski}
\author[2]{B\l{}a\.{z}ej Ruba}

\affil[1]{Department of Mathematical Methods in Physics, Faculty of Physics, \protect\\
University of Warsaw, Pasteura 5, 02-093 Warszawa, Poland, \protect\\ 
email: jan.derezinski@fuw.edu.pl}
\affil[2]{Institute of Theoretical Physics, Jagiellonian University, \protect\\
prof. Łojasiewicza 11, 30-348 Kraków, Poland, \protect\\
email: blazej.ruba@doctoral.uj.edu.pl}
\date{\today}
\maketitle

\begin{abstract}
We study  massless 1-dimensional Dirac-Coulomb  Hamiltonians, that
    is, operators on the half-line of the form $D_{\omega,\lambda}:=\begin{bmatrix}
-\frac{\lambda+\omega}{x} & - \partial_x \\
\partial_x & -\frac{\lambda-\omega}{x} 
\end{bmatrix}$. 
We describe  their closed realizations 
 in the sense of the Hilbert space $L^2(\R_+,\C^2)$, allowing
for complex values 
of the 
parameters $\lambda,\omega$. 
In physical situations, $\lambda$ is proportional to the electric 
charge and $\omega$  is related to the angular momentum.

We focus on realizations of $D_{\omega,\lambda}$
 homogeneous of degree $-1$. They
can be organized in a single holomorphic family of closed operators
parametrized
by a certain 2-dimensional complex manifold.
We describe the spectrum and the numerical
range of these realizations. We give an explicit formula for the
integral kernel of their
resolvent in terms of Whittaker functions. We also describe their stationary scattering
theory, providing formulas for a~natural pair of
diagonalizing operators and for the scattering
operator. We describe the point spectrum of their nonhomogeneous realizations.

It is well-known that $D_{\omega,\lambda}$
arise after separation of variables 
of the Dirac-Coulomb operator in  dimension 3. We give a simple argument
why this is still true in any dimension. Furthermore, we
  explain the relationship
  of spherically symmetric Dirac operators with the Dirac operator on
  the sphere and its eigenproblem.

Our work is mainly motivated by a large literature devoted to
distinguished self-adjoint realizations of Dirac-Coulomb
Hamiltonians. We show that these realizations arise naturally if the
 holomorphy is taken as the guiding principle.
Furthermore, they are infrared attractive fixed points of the scaling
action. Beside applications in relativistic quantum mechanics,
Dirac-Coulomb Hamiltonians are argued to provide a natural setting for
the study of Whittaker (or, equivalently, confluent hypergeometric) functions.
\end{abstract}

\begin{flushright}
  Dedicated to the memory of Krzysztof Gaw\c{e}dzki
  \end{flushright}

\newpage

\tableofcontents

\section{Introduction}

The main topic of this paper is 
the 1-dimensional massless Dirac Hamiltonian
with a two-parameter perturbation proportional to the Coulomb potential
\begin{equation}\label{oper}
  D_{\omega, \lambda} = 
\begin{bmatrix}
-\frac{\lambda+\omega}{x} & - \partial_x \\
\partial_x & -\frac{\lambda-\omega}{x} 
\end{bmatrix}.
\end{equation}
We allow the parameters $\omega,\lambda$ to be complex. We will
describe  realizations of \eqref{oper} as  closed operators on
$L^2(\R_+,\C^2).$ We will call \eqref{oper} the {\em one-dimensional Dirac-Coulomb
Hamiltonian} or {\em operator} (omitting usually the adjective one-dimensional, or
shortening it to 1d).

The formal operator 
$  D_{\omega, \lambda} $ is homogeneous of degree $-1$. Among its various
 closed realizations we will be especially interested in  homogeneous ones, i.e.\ those whose domain is invariant with respect to scaling transformations.

Our main motivation to study $D_{\omega, \lambda}$ comes from the 3d Dirac-Coulomb
Hamiltonian
\begin{equation}\label{dracoul}
  \sum_{j=1}^3 \alpha_j p_j +\beta m-\frac{\lambda}{r}\end{equation}
acting on four component spinor functions on $\R^3$. Here $m \in \R$ is the mass parameter, $\lambda \in \R$ is related to the charge of nucleus and $p_j:=-\i\partial_{x^j}$. As is
well known, after separation of variables in
\eqref{dracoul} with $m=0$ one obtains \eqref{oper}.
Possible
values of $\omega$ are $\pm 1, \pm 2,\dots$. They are related to the
angular momentum. Similar separation  is possible also in other
dimensions, albeit leading to different values of $\omega$. We remark
that the mass term is  bounded  and hence does
  not change the domain. Therefore, the analysis
of the
$m=0$ case yields the description of  closed
  realizations of the massive Dirac-Coulomb operator.

The second source of interest in $D_{\omega, \lambda}$ is the
expectation that models with scaling symmetry describe the behaviour of much more complicated systems in certain limiting cases.   

There exists another important motivation for the study of  Dirac-Coulomb Hamiltonians.
Objects related to  \eqref{oper}, such as its eigenfuntions and
Green's kernels can be expressed in terms of Whittaker
functions (or, equivalently, confluent functions). Whittaker functions are eigenfunctions of the {\em Whittaker
    operator}
\begin{equation}\label{whit1.}
L_{\beta,\alpha } :=-\partial_x^2+\Big(\alpha -\frac14\Big)\frac{1}{x^2}-\frac{\beta}{x}.
\end{equation}
The Dirac-Coulomb Hamiltonian may be viewed as a good way to organize
our knowledge about Whittaker functions, one of the most important
families of special functions in mathematics. Curiously, it
seem more suitable for this goal than
the Whittaker operator itself. Indeed, the  homogeneity of  the
Dirac-Coulomb operator
leads to
several identities
which have no
counterparts in the case of the Whittaker operator (e.g.  the
scattering theory described in Section \ref{Diagonalization} with
\cite{DeRi18_01} and \cite{DeFaNgRi20_01}).

Let us briefly describe the content of our paper.
The most obvious closed realizations of $  D_{\omega, \lambda} $ are 
the minimal and maximal realizations, denoted  $  D_{\omega, \lambda}^{\min}
$ and  $  D_{\omega, \lambda}^{\max}
$. Both are 
homogeneous of degree $-1$. They depend holomorphically on parameters
$\omega,\lambda$, except for $|\RE\sqrt{\omega^2-\lambda^2}|=\frac12$, where a kind of a ``phase 
 transition'' occurs. 
One of the signs of this phase transition is the following: For
$|\RE\sqrt{\omega^2-\lambda^2}| \geq \frac12$, we have $  D_{\omega, \lambda}^{\min}
=  D_{\omega, \lambda}^{\max}$, so that in this parameter range there is only one
closed realization of $  D_{\omega, \lambda} $.
However, for $|\RE\sqrt{\omega^2-\lambda^2}|<\frac12$,
the domain of $  D_{\omega, \lambda}^{\min} $ has codimension 2 as a 
subspace of the domain of $ D_{\omega, \lambda}^{\max}$. 
This means that for fixed $(\omega, \lambda)$ in this region there exists a one-parameter family of closed realizations of $D_{\omega, \lambda} $ strictly 
between the minimal and maximal realization.

In  operator theory (and other domains of mathematics) it is
useful to organize objects in holomorphic families \cite{Kato,DeWr}. Therefore we ask whether 
$  D_{\omega, \lambda}^{\min}= D_{\omega, \lambda}^{\max}$ can be
analytically continued beyond the region
$|\RE\sqrt{\omega^2-\lambda^2}|>\frac12$.
The answer is positive, but the domain of this continuation is a complex manifold which is not simply an open subset of the ``$(\omega,\lambda)$-plane''~$\C^2$. To define this manifold we start with the following subset of $\C^3$:
\begin{equation}\label{qua}
\Big\{(\omega,\lambda,\mu)\ |\   \mu^2=\omega^2-\lambda^2,\quad\mu>-\frac12\Big\}.
\end{equation}
Then we ``blow up'' the singularity $(\omega,\lambda,\mu)=(0,0,0)$. The resulting complex 2-dimensional manifold is denoted
$\cM_{-\frac12}$. There exists a natural projection $\cM_{- \frac12}
\to \C^2$. The preimage of $(\omega, \lambda) \in \C^2$ has one
element if $|\RE \sqrt{\omega^2 - \lambda^2}| \geq \frac12$, two
elements if $|\RE \sqrt{\omega^2 - \lambda^2}| < \frac12$, and
$(\omega, \lambda) \neq (0,0)$ and infinitely many elements if $\omega
= \lambda =0$. This last preimage, called the {\em zero fiber}, is isomorphic to the Riemann sphere $\CP^1$, for which we use homogeneous coordinates $[a{:}b]$. Away from the zero fiber, points of $\cM_{- \frac{1}{2}}$ may be labeled by triples $(\omega, \lambda, \mu)$.

The main result of our paper is the construction of a holomorphic family of closed
operators
$\cM_{-\frac12}\ni p\mapsto D_p$ consisting of homogeneous
Dirac-Coulomb Hamiltonians. If 
$p\in\cM_{-\frac12}$ 
lies over $(\omega,\lambda)$, then we have inclusions
\begin{equation}
  \Dom(  D_{\omega, \lambda}^{\min})\subset
  \Dom(  D_p)\subset
  \Dom(  D_{\omega, \lambda}^{\max}).\label{piy}
 \end{equation}
If $|\RE \sqrt{\omega^2 - \lambda^2}| \geq\frac12$, both inclusions in \eqref{piy} are equalities. On the other hand, for
$|\RE \sqrt{\omega^2 - \lambda^2}| < \frac12$ both inclusions are
proper and elements of the domain of $  \Dom(D_p)$ are distinguished
by the following behavior near zero:
\begin{equation}
\label{sol11}
\sim \frac{x^{\mu}}{\omega+ \lambda} \begin{bmatrix}-\mu\\\omega+\lambda\end{bmatrix},\qquad \sim \frac{x^{\mu}}{\omega - \lambda}\begin{bmatrix}\omega-\lambda\\-\mu\end{bmatrix}.
\end{equation}
Note that the two functions in \eqref{sol11}, when both well defined, are
proportional to one another.

We describe various properties of $D_p$: we find its point
spectrum, essential spectrum, numerical range, discuss conditions for
(maximal) dissipativity. We construct explicitly the
resolvent. Some spectral properties, including their point spectra, of nonhomogeneous realizations of $D_{\omega,\lambda}$ are also discussed.


Whenever $D_p$ is self-adjoint, its spectrum is absolutely
  continuous, simple and coincides with~$\R$. In non-self-adjoint
  cases, the essential spectrum is still $\R$, but on  certain exceptional subsets of the parameter space there is also point spectrum $\{\IM(k) > 0\}$ or $\{\IM(k) < 0\}$. Away from exceptional sets $D_p$ possesses non-square-integrable eigenfunctions,
  which can be called {\em distorted waves}.
They can be normalized in two ways: as {\em incoming} and {\em
  outgoing} distorted waves.
They  define the integral kernels of a pair  of operators $\cU^\pm$
that, at least formally, diagonalize
 $D_p$. More precisely, on a dense domain  $\cU^\pm$  intertwine
$D_p$ with the operator of the multiplication by the
independent variable $k\in\R$. Up to a trivial factor, $\cU^\pm$ can be interpreted
as  the {\em wave (M{\o}ller) operators}. The operators $\cU^+$ and $\cU^-$
are related to one another by 
the identity $S \cU^-:= \cU^+$, which defines the {\em scattering operator} $S$.
Thus we are able to describe rather
completely the {\em stationary scattering theory} of homogeneous
  Dirac-Coulomb Hamiltonians.

For self-adjoint $D_p$, the operators $\cU^\pm$ are
unitary. If $\lambda$ is real, they are still bounded and
invertible, even if $D_p$ are not self-adjoint.
We show that $\cU^\pm$ can be written (up to a trivial factor) as $\Xi^\pm(\sgn(k),A)$, where
$A$ is the
 dilation generator and $\sgn(k)$ is the sign of the spectral
 parameter.
We express  $\Xi^\pm$ in terms of the hypergeometric function. We prove that they behave as
$s^{|\IM(\lambda)|}$ for $s\to\infty$. In particular, this shows that
$\cU^\pm$ are bounded only for real $\lambda$.

The Coulomb potential is long-range. Therefore we cannot use
the standard formalism of
scattering theory.
In our paper we restrict ourselves
to the stationary formalism, where the long-range character of the
perturbation
is taken into account by using appropriately modified plane waves.

Operators $D_p$ with $p$ in the zero fiber can be
fully analyzed by elementary means. All operators
strictly between $D_{0,0}^{\min}$ and $D_{0,0}^{\max}$ are homogeneous
and are specified by boundary conditions at zero of the form $f(0) \in
\C \begin{bmatrix} a \\ b \end{bmatrix}$ for $[a{:}b]\in
\CP^1$. Operator corresponding to boundary condition $[a{:}b]$ will be
denoted $D_{[a:b]}$. Other cases in which operators $D_p$ are
particularly simple are discussed in  Appendix \ref{sec:dir1d}.

\begin{figure}[ht]
  \centering
  \includegraphics[width=0.7\textwidth]{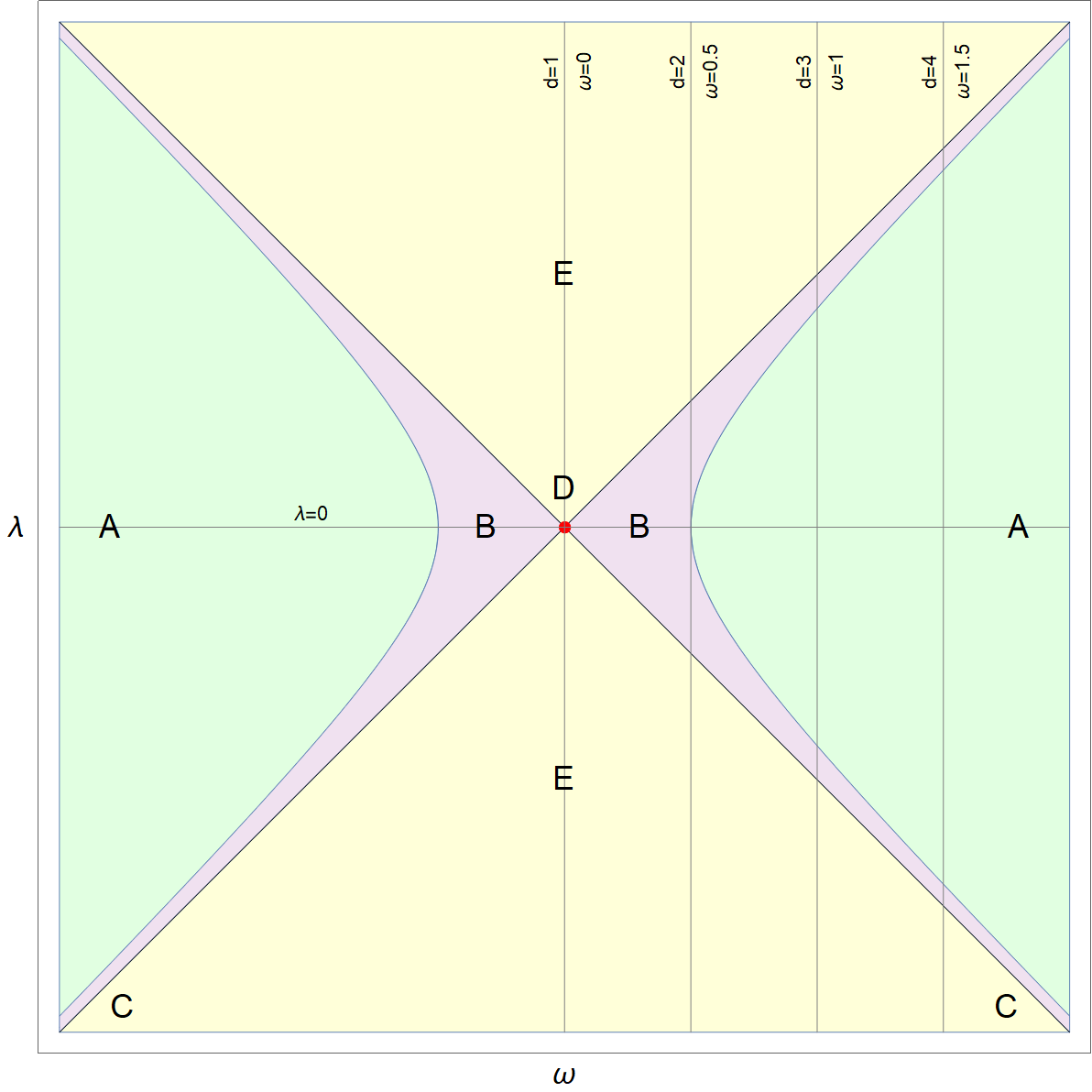}   
  \caption{Phase diagram of operators
      $D_{\omega,
      \lambda}$ for $(\omega,\lambda)\in\R^2$. It
    is partitioned into five subsets corresponding to five possible
    behaviours, see Propositions \ref{self-adj1} and \ref{self-adj2}
    and Figure \ref{fig:phases}. We label regions as follows. Color
    green and letter $A$ refer to $\omega^2 - \lambda^2 \geq \frac 14$
    (we do not give a separate name to the boundary of this region,
    although it is also somewhat special). Color blue and letter $B$
    refer to the subset $0 < \omega^2 - \lambda^2 < \frac14$. Black
    lines and letter $C$ refer to the lines $\omega = \pm \lambda$,
    except for the special point $(\omega, \lambda)=(0,0)$, which is
    marked with a fat red dot and letter $D$. Yellow color and letter
    $E$ are used for the region $\omega^2 - \lambda^2 < 0$. In
    addition we present lines corresponding to the lowest angular momentum values for dimensions $d=0$, $d=1$, $d=2$ and $d=3$. Here we disregard the possible sign of $\omega$, which is irrelevant due to symmetry $\omega \mapsto - \omega$.}
\label{fig:plane}
\end{figure}

\begin{figure}[ht]
  \centering
  \includegraphics[width=0.9\textwidth]{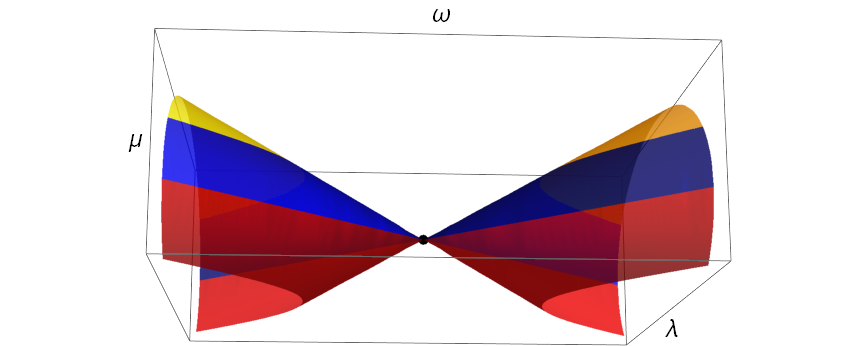}   
  \caption{Parameter space of homogeneous self-adjoint Dirac-Coulomb Hamiltonians projected onto axes $\omega, \lambda, \mu$. Regions colored yellow, blue and red are described by inequalities $\mu > \frac12$, $0 < \mu < \frac12$ and $- \frac12 < \mu < 0$, respectively. The fat dot at the origin represents a circle contained in the zero fiber $\cZ$, so the whole parameter space is topologically a cylinder.}
\label{fig:cone}
\end{figure}


The operator $D_{\omega, \lambda}^{\min}$ is Hermitian (symmetric with respect to
the scalar product $(\cdot|\cdot)$) if and only if $\omega,
\lambda \in \mathbb R$. Below we state our main results about self-adjoint realizations
 of  $  D_{\omega, \lambda} $
 in the form of two propositions.
 They are immediate consequences of the  results of Sections \ref{sec:minmax}, \ref{sec:homog}. We present also the phase diagram of operators $D_{\omega, \lambda}$ on Figure \ref{fig:plane} and the parameter space of homogeneous self-adjoint Dirac-Coulomb Hamiltonians on Figure \ref{fig:cone}. 

Let $H^1(\R_+)$ be the first Sobolev space on $\R_+$ and
$H_0^1(\R_+)$ be the closure of $C_\mathrm{c}^{\infty}$ in $H^1(\R_+)$.

\begin{proposition} \label{self-adj1}
 Let  $\omega, \lambda \in \mathbb R$.
The Hermitian operator $D_{\omega, \lambda}^{\min}$ has the following properties.
\begin{enumerate}
\item If $\frac{1}{4}\leq\omega^2 - \lambda^2$, it is self-adjoint
  and $D_{\omega, \lambda}^{\min}=D_{\omega, \lambda,\sqrt{\omega^2-\lambda^2}}$
    \item If $\omega^2 - \lambda^2<\frac14$, it
has deficiency indices $(1,1)$. Hence there exists a circle of
self-adjoint extensions.
\end{enumerate}
\end{proposition}

\begin{proposition} \label{self-adj2}
\begin{enumerate}
    \item[1a.] If $\frac{1}{4}<\omega^2 - \lambda^2$, we have $\Dom(D_{\omega, \lambda}^{\min}) = H_0^1(\R_+,\C^2)$.
    \item[1b.] If  $\frac{1}{4}=\omega^2 - \lambda^2$, we have
$H_0^1(\R_+,\C^2) \subsetneq \Dom(D_{\omega, \lambda}^{\min}) 
      $.
    \item[2a.] If $0<\omega^2 - \lambda^2<\frac14$, exactly two
      self-adjoint extensions
      of  $D_{\omega, \lambda}^{\min}$ are homogeneous,
namely $D_{\omega, \lambda, \sqrt{\omega^2 - \lambda^2}}$ and
 $D_{\omega, \lambda, -\sqrt{\omega^2 - \lambda^2}}$. The former is
 distinguished among all self-adjoint extensions by \begin{equation}
   \int_{0}^{\infty} \left( \frac{|f(x)|^2}{x} + |f(x)|
  |f'(x)| \right) \D x < \infty\text{ for }f \in \Dom(D_{\omega, \lambda,
  \sqrt{\omega^2 - \lambda^2}}),\end{equation}
 i.e.\ elements of its domain have finite expectation values of kinetic and potential energy.
    \item[2b.] If $| \lambda | = |\omega| \neq 0$, exactly one self-adjoint extension
      of  $D_{\omega, \lambda}^{\min}$ is homogeneous,
namely $D_{\omega, \lambda, 0}$. It has the property $H_0^1(\R_+,
\C^2) \subsetneq
\Dom(D_{\omega, \lambda, 0}) \subsetneq H^{1}(\R_+,\C^2)$.
\item[2c.] If $\lambda = \omega =0$, all
self-adjoint extensions
      of  $D_{\omega, \lambda}^{\min}$
are homogeneous. They have the form $D_{[a:b]}$ with $[a:b] \in \mathbb{RP}^1$.
\item[2d.] If $|\lambda| > |\omega|$, none of
self-adjoint extensions
      of  $D_{\omega, \lambda}^{\min}$ is homogeneous.
\end{enumerate}
\end{proposition}

Now let $\omega, \lambda$ be real and suppose that $\omega^2 -
\lambda^2 < \frac14$.  Let $\tau\mapsto U_\tau$ denote the scaling
transformation. The parameter space of self-adjoint extensions is a circle. It admits an action of the scaling group given by
\begin{equation}
D \mapsto U_{\tau} D U_{\tau}^{-1}.
\label{eq:scaling_action}
\end{equation}
The fixed points of this action are the homogeneous self-adjoint extensions. Main properties of this action are illustrated by Figure \ref{fig:phases}.

\begin{figure}
  \centering
  \includegraphics[width=0.95\textwidth]{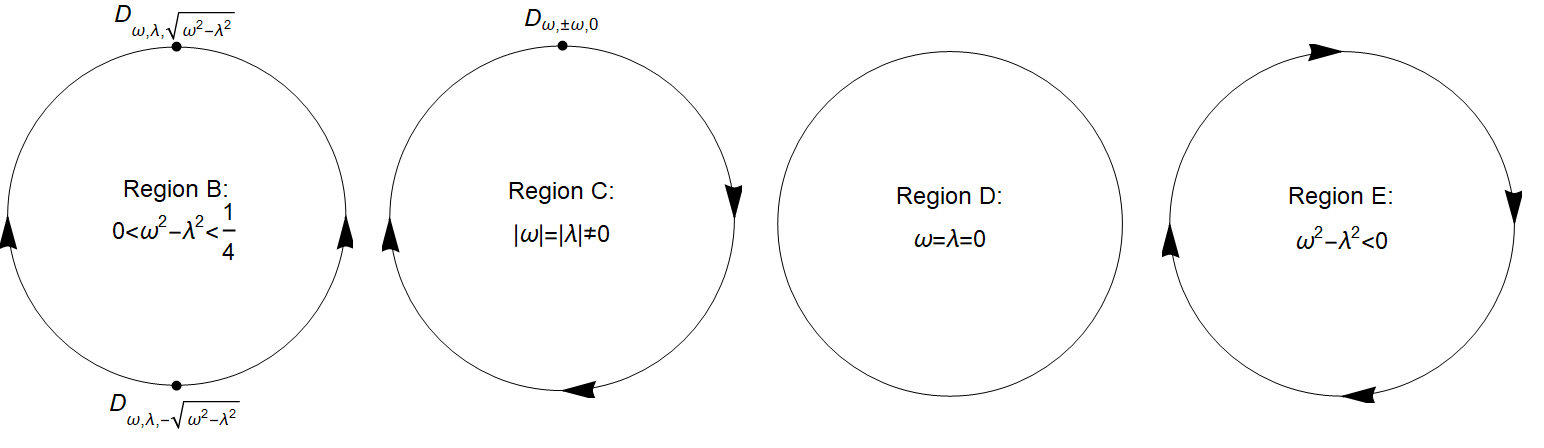}
  \caption{Visualization of the action of the scaling group on
    self-adjoint extensions of $D_{\omega, \lambda}^{\min}$ in four
    regions covering the set of $\omega, \lambda$ satisfying $\omega^2
    - \lambda^2 < \frac14$. Fat dots are the fixed points, while
    arrowheads indicate the direction of the flow as $\tau$ increases,
    see \eqref{eq:scaling_action}. In the first region there are two
    fixed points, attractive and repulsive, corresponding to a positive
    and negative $\mu$, respectively. As $\omega^2 - \lambda^2$
    decreases to zero, the two fixed points merge to one degenerate
    fixed point, except for the point $\omega = \lambda =0$ at which
    the scaling action becomes trivial. As $\omega^2 - \lambda^2$
    decreases below zero, the scaling action becomes periodic with period $\frac{\pi}{\sqrt{\lambda^2 - \omega^2}}$.}
\label{fig:phases}
\end{figure}

As we present in  Appendix
\ref{ddimensions},  $d$-dimensional Dirac-Coulomb Hamiltonians can
be reduced to  the
radial operator \eqref{oper}. Combined with the analysis presented above, one
obtains rather complete information about self-adjointness and
homogeneity properties of these operators. Here we point out only a
few facts concerning these extensions on the lowest angular momentum sector.
\begin{itemize}
\item \textbf{dimension $1$:} There exist no homogeneous self-adjoint realizations for any $\lambda \neq 0$.
\item \textbf{dimension $2$:} The operator defined on smooth spinor-valued
  functions with compact support not containing zero is not
  essentially self-adjoint for any $\lambda \neq 0$. For
  $|\lambda|< 1$ there exist  homogeneous self-adjoint
  extensions of $D_{\omega, \lambda}^{\min}$. These
  homogeneous extensions can be organized into two continuous
  families. The (more physical)  family is defined on
  $[-1,1]$. At the endpoints $\lambda=\pm1$ it meets the other family, which
  is defined on $[-1,0[\cup]0,1]$.
\item \textbf{dimension $d \geq 3$:} The operator defined as above is
  essentially self-adjoint if $\lambda^2 \leq
  \frac{d(d-2)}{4}$. If~$\frac{d(d-2)}{4} < \lambda^2$
  it is not essentially self-adjoint. However, for $\lambda^2\leq
  \frac{(d-1)^2}{4}$
  there exists  homogeneous self-adjoint extensions of $D_{\omega,
    \lambda}^{\min}$. They can be organized into two  families
  depending
  continuously on $\lambda$. The more physical family is defined
  on $[-  \frac{(d-1)^2}{4},   \frac{(d-1)^2}{4}]$. The second family
  meets the first at the endpoints and is defined on
  $[-  \frac{(d-1)^2}{4}, -\frac{d(d-2)}{4}[\,\cup\,
]\frac{d(d-2)}{4},  \frac{(d-1)^2}{4}]. $
\end{itemize}
In all cases in which there exist no homogeneous self-adjoint extensions, the defect indices are nevertheless equal and hence there exist nonhomogeneous self-adjoint extensions.

Analysis of self-adjoint
realizations
  of the 3-dimensional Dirac-Coulomb Hamiltonian has a~long and rich history
in the mathematical literature. There even exists  a recent review
paper devoted to this subject \cite{Gal}. Let us explain the main
points of this history, refering the reader to \cite{Gal} for more details.

A direct application of the Kato-Rellich theorem yields the essential
self-adjointness of the (massive, 3d) Dirac-Coulomb Hamiltonian only for 
$|\lambda|<\frac12$. This proof is due to Kato
\cite{Kato1,Kato}. The essential self-adjointness up to the boundary
of the ``regular region''
$|\lambda|<\frac{\sqrt{3}}{2}$
 was proven independently by Gustaffson-Rejt\"o  \cite{GuRe,Re} and Schmincke 
 \cite{Schm1}. They needed to use slightly more refined
arguments going beyond to the basic Kato-Rellich theorem. The ``distinguished self-adjoint extension'' in the region
$\frac{\sqrt{3}}{2}<|\lambda|<1$ was described in several equivalent ways,
mostly involving the characterization of the domain, by Schmincke,
W\"ust, Klaus, Nenciu and others
\cite{Schm2,Wust1,Wust2,Nen,Gal2, Cassano1,Cassano2}. 
  The characterization of distinguished self-adjoint extensions based
  on holomorphic families of operators was first proposed by Kato in \cite{KatoHol}. Esteban and Loss \cite{EL}
characterized the distinguished self-adjoint realization at the boundary
of the ``transitory region'', that is for $|\lambda|=1$, by using
the so-called Hardy-Dirac inequalities. Self-adjoint realizations in the
``supercritical region'' $|\lambda|>1$ were first studied by Hogreve in
  \cite{Ho}, and then (with some corrections) in \cite{Gal2}.
The authors of \cite{Gal2} analyze also the second distinguished branch
of self-adjoint extensions in the critical region, which they call
``mirror distinguished''. \cite{Cassano1,Cassano2} include in their
analysis a term proportional to $\frac{1}{r^2}\beta\alpha_ix_i$,
which they call ``anomalous magnetic''.

Our treatment of Dirac-Coulomb Hamiltonians is quite different from the
above references. We use  exact solvability to describe rather
completely their resolvent, domain and (stationary)
  scattering theory.
We do not add the mass term,
which helps with exact solvability and makes possible to use the
homogeneity as a good criterion for distinguished realizations. Another concept which we use
  is that of a holomorphic family of operators, which we view as an
important  criterion for distinguishing a realization.
The mass term is bounded, so it does not affect the basic picture of
distinguished realizations. Our analysis includes
realizations which are not necessarily self-adjoint, but
  turn out to be self-transposed with respect to a  natural complex
  bilinear form.
  Our description of various closed realizations of Dirac-Coulomb
  Hamiltonians
  is quite straightforward and  involves only elementary functions.  
 We use neither the von Neumann
  nor the Krein-Vishik theory
  theory of self-adjoint extensions, which lead to a rather
complicated description of the domains of closed
  description
  involving Whittaker functions, see \cite{Gal2,Gal3}.

Our analysis of Dirac-Coulomb Hamiltonians can be viewed as a
continuation of a series of papers  about holomorphic families of certain
1-dimensional Hamiltonians:  Bessel operators
\cite{BuDeGe11_01,DeRi17_01}
and Whittaker
operators \cite{DeRi18_01, DeFaNgRi20_01}.

\newpage

Let us mention some more papers, where Dirac-Coulomb operators play
an important role.

First, there exists a number of papers \cite{DV,TE,GY,Daude} devoted
to 
the time dependent approach to scattering theory for self-adjoint Dirac Hamiltonians on $\R^3$ with long range potentials.

There also  exists a large and 
  interesting literature devoted to eigenvalues inside a spectral gap 
  of a self-adjoint operator, with massive Dirac-Coulomb Hamiltonians as prime
  examples \cite{EL,ELS,SST,Gal3}. Massless Dirac-Coulomb Hamiltonians
  do not have a gap, and eigenvalues are possible only
    in non-self-adjoint nonhomogeneous cases. Nevertheless, we 
  believe that methods of our paper are relevant for the eigenvalue 
  problem in the massive self-adjoint case. 
  
For a study of one-dimensional Dirac operators with locally integrable
complex potentials, see \cite{Malamud}.

Finally, let us mention another interesting related topic, where the question
of distinguished self-adjoint realizations arises:  2-body Dirac-Coulomb 
Hamiltonians. Their mathematical study was undertaken in \cite{DO}. Even
though the
physical significance of these Hamiltonians is not very clear, they are
widely used in quantum chemistry.

Let us briefly describe the organization of our paper.
  Its main part, that is Sections
  \ref{Blown-up quadric}--\ref{mixed-bc}
  describes  realizations of 1d Dirac-Coulomb Hamiltonians on $L^2(\R_+,\C^2)$
  focusing on the homogeneous ones. Besides, our paper contains four
  appendices, which can be read independently.

  Appendix
  \ref{sec:dir1d} first discusses some general concepts related to 1d
  Dirac operators. Then two special classes of 1d Dirac-Coulomb
  Hamiltonians are analyzed in detail.

    Essentially all  papers that we mentioned in our bibliographical
    sketch treat the 3-dimensional case.
    It was pointed out in \cite{XIAOetal}  that a general $d$-dimensional spherically symmetric Dirac 
Hamiltonian can be reduced to 
a 1-dimensional one. We describe this reduction in detail in
Appendix \ref{ddimensions}. We also analyze its various group-theoretical
and differential-geometric aspects, including the relation to Dirac operators on spheres and the famous Lichnerowicz formula. Spectra of the latter are computed in two independent ways and a construction of eigenvectors is presented.

The short Appendix \ref{app:operators_Rplus} is devoted to the Mellin
transformation.

 Finally, in Appendix \ref{whittaker} we collect properties of various
 special functions, mostly, Whittaker functions, which are used in our paper.
 We mostly follow the conventions of
 \cite{DeRi18_01, DeFaNgRi20_01}.

\subsection{Remarks about notation}
\label{Remarks about notation}
Symbol $( \cdot | \cdot )$ is used for standard scalar products on $L^2$ spaces, linear in the second argument, while $\langle \cdot | \cdot \rangle$ is used for the analogous bilinear forms in which complex conjugation is omitted:
\begin{equation}
\langle f | g \rangle = \int f(x)^{\T} g(x) \D x.
\end{equation}
Tranpose (denoted by the superscript $\scriptstyle\T$) of a densely defined operator is defined in terms of $\langle \cdot | \cdot \rangle$ in the same way as the adjoint (denoted by $*$) is defined in terms of the scalar product. We use superscript $\tperp$ for orthogonal complement with respect to $\langle \cdot | \cdot \rangle$ and $\perp$ for orthogonal complement with respect to $( \cdot | \cdot )$. Overline always denotes complex conjugation, for example we have $X^\perp = \overline{X^\tperp}$ for a subspace $X$.

We will write $\R_+ = ]0, \infty [$, $\C_{\pm} = \{ z \in \C \, | \,
\pm \IM(z) >0 \}$, $\N=\{0,1,\dots\}$. Omission of zero will be denoted by $\times$, e.g.\ $\R^\times =\R \backslash \{ 0 \}$. Indicator function of a subset $S \subset \R$ will be denoted by $\one_S$. We label elements of the Riemann sphere $\CP^1$ using homogeneous coordinates, i.e.\ $[a:b] \in \CP^1$ is the complex line spanned by $(a,b) \in \C^2 \backslash \{ 0 \}$.

Operators of multiplication of a~function in $L^2(\R_+, \C^n)$ and
$L^2(\R, \C^n)$ by its argument will be denoted by $X$ and $K$,
respectively. Dilation group action on $L^2(\R_+,\C^n)$ is defined by
$U_{\tau} f (x) = \e^{\frac{\tau}{2}} f(\e^{\tau} x)$. We denote its
self-adjoint generator by $A$, so that $U_{\tau} = \e^{\i \tau
  A}$. Operator $O$ is said to be homogeneous of degree $\nu$ if $U_{\tau} O U_{\tau}^{-1} = \e^{\nu \tau} O$. Inversion operator $J$ is defined by $(Jf)(x) = \frac{1}{x} f \left( \frac{1}{x} \right)$. It is unitary and satisfies $J^2=1$, $JAJ^{-1}=-A$.

Complex power functions $z \mapsto z^a$ are holomorphic and defined on
$\C \backslash ] - \infty, 0 ]$. Domains of holomorphy of special
functions used in the text are specified in Appendix
\ref{whittaker}.

In our paper we will often use the concept of a holomorphic map with
values in closed operators, which we now briefly recall
\cite{Kato,DeWr}. We will give two equivalent definitions of this concept: the first
is ``more elegant'', the second ``more practical''.
To formulate the first definition note that
the Grassmannian (the set of closed subspaces)
$\mathrm{Grass}(X)$ of a Hilbert space $X$
carries the structure of a
complex Banach manifold \cite{Douady}.

Consider Hilbert spaces $X_2, X_3$ be Hilbert spaces and a complex
manifold $\cM$. We say that a function
 $\cM\ni p\mapsto T_p$ of closed operators $X_2 \to X_3$ is
 holomorphic if and only if $p \mapsto \mathrm{graph}(T_p) \in \mathrm{Grass}(X_2 \times X_3)$ is a holomorphic map. 

 Equivalently, $\cM\ni p\mapsto T_p$
 is holomorphic if for every $p_0 \in \cM$ there exists
a~neighborhood $\cM_0$ of $p_0$ in $\cM$, a Hilbert space $X_1$ and a
holomorphic family $\cM_0\ni p\mapsto S_p$ of bounded injective
operators $S_p : X_1 \to X_2$  such that $\Ran(S_p)= \Dom(T_p)$ and $T_p S_p$ form
a~holomorphic family of bounded operators.

\section{Blown-up quadric}\label{Blown-up quadric}

Formal Dirac-Coulomb Hamiltonians depend on parameters
$(\omega,\lambda)\in\C^2$.
In order to specify their realizations as closed homogeneous operators, it is necessary to choose a square root of $\omega^2-\lambda^2$. 
For this reason homogeneous Dirac-Coulomb Hamiltonians are
parametrized by points of a~certain complex manifold.
This section is devoted to its definition and basic properties.

Let us first introduce a certain null quadric in $\C^3$:
\begin{equation}
\cM^\pre:= \left \{ (\omega, \lambda, \mu) \in \C^3 \, | \, \omega^2 = \lambda^2 + \mu^2
  \right \}.
\label{quadric}\end{equation}
By the holomorphic implicit function theorem, $\cM^{\pre}$ is a complex two-dimensional submanifold of $\C^3$ away from the singular point $ (0,0,0)$ (also denoted $0$ for brevity).

We consider also the so-called {\em blowup} of $\cM^\pre$ at the singular
point, defined by
\begin{equation}
  \cM = \left \{ (\omega, \lambda, \mu,[a:b]) \in \C^3 \times \CP^1 \, | \, \begin{bmatrix}  \omega + \lambda &  \mu \\ \mu &   \omega - \lambda \end{bmatrix} \begin{bmatrix} a \\ b \end{bmatrix} = \begin{bmatrix} 0 \\ 0 \end{bmatrix}
  \right \}.
\end{equation}
Fibers of the projection map $\cM \to \CP^1$ are described by triples $(\omega, \lambda, \mu) \in \C^3$ subject to two linearly independent linear equations, whose coefficients are holomorphic functions on local coordinate patches of $\CP^1$. Therefore $\cM$ is a holomorphic line bundle over $\CP^1$, embedded in the trivial bundle $\C^3 \times \CP^1$. In particular it is a two-dimensional complex manifold.

Equation $\begin{bmatrix}  \omega + \lambda &  \mu \\ \mu &   \omega - \lambda \end{bmatrix} \begin{bmatrix} a \\ b \end{bmatrix} = \begin{bmatrix} 0 \\ 0 \end{bmatrix}$ has a solution different than $(a,b) = (0,0)$ if and only if the quadratic equation defining $\cM^{\pre}$ is satisfied. Thus there is a~projection map $\cM \to \cM^{\pre}$. Its restriction to the preimage of $\cM^\pre\backslash\{  0 \}$ is an isomorphism and will be treated as an identification. The preimage of zero, called the zero fiber and denoted $\cZ$, is an isomorphic copy of $\CP^1$.

We will often use the short notation $p=(\omega, \lambda,\mu,[a:b])$
for elements of $\cM $. If $p \notin \cZ$, then $[a:b]$ is uniquely determined by $(\omega, \lambda, \mu)$ and we abbreviate $p = (\omega, \lambda, \mu)$. In turn for $p$ in the zero fiber we write $p = [a:b]$.

We will now describe  useful coordinate systems on $\cM$. The coordinates 
\begin{equation}
    z = \frac{a}{b}, \qquad \omega+ \lambda
    \label{cp1}
\end{equation}
are valid on $\{ b \neq 0 \}$ -- the open subset of $\cM$ which is the complement of
\begin{equation}
\{ b = 0 \} = \{ ( \omega, -\omega, 0,[1:0])   \}_{\omega \in \C}.
\end{equation}
More precisely, the following map is an isomorphism of complex manifolds:
\begin{equation}
\C^2 \ni (\omega + \lambda, z) \mapsto \left( \frac{(\omega+ \lambda) (1+ z^2)}{2}, \frac{(\omega+\lambda) (1- z^2)}{2}, - (\omega+\lambda) z , [z:1] \right) \in \{ b \neq 0 \}.
\end{equation}
We note that
\begin{equation}
  z = - \frac{\mu}{\omega+ \lambda} = - \frac{\omega - \lambda}{\mu}
\end{equation}
whenever the denominators are nonzero. 

Analogously, on $\{ a \neq 0 \}$, the complement of
\begin{equation}
\{ a = 0 \} = \{ ( \omega, \omega, 0,[0:1])   \}_{\omega \in \C},
\end{equation}
we use the coordinates $z^{-1}$ and $\omega - \lambda$.

Sets $\{ a \neq 0 \}$, $\{ b \neq 0 \}$ cover the whole $\cM$. On their intersection we have
\begin{equation}
\omega - \lambda= (\omega+ \lambda) z^2.
\label{eq:transition}
\end{equation}

We note that the locus $\{ \lambda = 0 \}$ is the union of three Riemann surfaces:
\begin{equation}
\{ \lambda = 0 \} = \cZ \cup \{ a = b \} \cup \{ a = - b \}.
\label{eq:lambda_zero_decomp}
\end{equation}
It is singular at the intersection points:
\begin{equation}
\{[1:1] \}=\cZ \cap \{ a = b \},\quad\{[1:-1]\}=\cZ \cap \{ a = - b
  \}.
\end{equation}
On~the other hand, the  level sets $\{ \lambda = \lambda_0 \}$ with $\lambda_0 \neq 0$ are nonsingular. Similarly, we have
\begin{equation}
    \{ \mu = 0 \} = \cZ \cup \{ a = 0 \} \cup \{ b = 0 \}.
\end{equation}

\begin{remark}Consider the {\em tautological line bundle}
    $\mathcal{N}\to\CP^1$, i.e.\ the space of pairs
${\big((a',b'),[a:b]\big) \in \C^2 \times \CP^1}$ such that $(a',b') \in [a:b]
$. Setting $z:=\frac{a'}{b'}$, we obtain two charts $(b',z)$ and
$(a',z^{-1})$, which cover $\mathcal{N}$. The clutching formula for
$\mathcal{N}$ is
$a'=b' z$, which can be compared with the clutching formula \eqref{eq:transition}
for $\mathcal{M}$. Thus we see that as a holomorphic vector bundle $\cM$ is isomorphic to the tensor 
square of $\mathcal{N}$. 
\end{remark}


Later we will encounter the meromorphic functions on $\cM$
\begin{equation}
N_p^\pm = \frac{z \pm \i}{\Gamma(1 + \mu \mp \i \lambda)}.
\end{equation}
We define the {\em exceptional sets} as their zero loci:
\begin{align}
\cE^{\pm} := &\{  N_p^\pm =0  \} =\bigcup_{n=0}^\infty\cE_n^\pm,\\
  \cE_0^\pm:=& \{p\in\cM\,|\, a = \mp \i b \}=\{p\in\cM\,|\,z=\mp\i\},\notag\\
\cE_n^\pm:=&              \{ p \in \cM \, | \, \mu \mp \i \lambda
             = -n \},\quad n=1,2,\dots. \nonumber 
\end{align}
Away from $\cZ$, the condition
$p\in\cE_0^\pm$ is equivalent to $\mu \mp \i
\lambda =0$. Thus for $p \notin \cZ$ we have $p \in \cE^{\pm}$ if and
only if $\mu \mp \i \lambda \in - \mathbb N$.
Moreover,
\begin{equation} \cE^{\pm}\cap\cZ=\cE_0^\pm\cap\cZ=
\{ [\mp \i : 1] \}.\end{equation} In particular $\cZ \cap \cE^+\cap \cE^- = \emptyset$.
Clearly, the sets $\cE_n^\pm$, $n=0,1,2,\dots$, are connected
components of $\cE^\pm$. Each $\cE_n^\pm$ is isomorphic to $\C$. Indeed, $\cE_0^\pm$ is a fiber of $\cM \to \CP^1$ and $\cE_n^\pm$ with $n \geq 1$ is globally parametrized by $\omega$.

\begin{lemma} \label{eplus_eminus_intersection}
$\cE^+ \cap \cE^-$ is a countably infinite discrete subset of $\cM$ on which $2 \mu + 1 \in - \mathbb N$. In~particular $\mu \leq - \frac{1}{2}$.
\end{lemma}
\begin{proof}
Suppose that $p \in \cM$ is such that $\mu + \i \lambda = -n$, $\mu - \i \lambda = -m$ with $n,m \in \mathbb N$. Then
\begin{equation}
(\omega, \lambda, \mu) = \left( \pm nm, \frac{m-n}{2 \i}  , - \frac{m+n}{2} \right), \qquad (n,m) \in \N^2 \backslash \{ (0,0) \},
\label{eq:ep_em_int_par}
\end{equation}
from which the discreteness and countability of $\cE^+ \cap \cE^-$ is clear. If both $n,m$ are zero, then $\mu = \lambda =0$ and hence also $\omega = 0$. In this case we have $p \in \cZ \cap \cE^+ \cap \cE^- = \emptyset$---contradiction. Thus at least one of $n,m$ is nonzero, and we have $2 \mu + 1 = 1 - n - m \in - \mathbb N$. Conversely, if $(n,m) \in \mathbb N^2$ is different than $(0,0)$, then \eqref{eq:ep_em_int_par} defines one or two (if $nm \neq 0$) points of $\cE^+ \cap \cE^-$, so this set is infinite.
\end{proof}


We define the \emph{principal scattering amplitude} as the ratio
\begin{equation}
S_p = \frac{N_p^-}{N_p^+} = \frac{z-\i}{z+\i} \frac{\Gamma(1 + \mu - \i \lambda)}{\Gamma(1 + \mu + \i \lambda)} = \frac{(\omega - \lambda + \i \mu) \Gamma(1 + \mu - \i \lambda)}{(\omega - \lambda - \i \mu) \Gamma(1 + \mu + \i \lambda)}.
\end{equation}
It satisfies $\overline{S_{\overline p}} = S_p^{-1}$, hence it has a unit modulus for $p = \overline p$. Furthermore,
\begin{equation}
\cE^- \backslash \cE^+ = \{ S_p = 0 \}, \qquad \cE^+ \backslash \cE^- = \{ S_p = \infty \}, \qquad \cE^- \cap \cE^+ = \{ S_p \text{ indeterminate} \}.
\end{equation}

We introduce an involution on $\cM$ by
\begin{equation}\label{tau}
\tau (\omega, \lambda, \mu, [a:b]) = (\omega, \lambda, - \mu, [-a:b]).
\end{equation}

\section{Eigenfunctions and Green's kernels} \label{eigen}

\subsection{Zero energy}

The 1d Dirac-Coulomb Hamiltonian with parameters $\omega,\lambda \in
\C$ is given by the expression
\begin{equation}
  D_{\omega, \lambda} = 
\begin{bmatrix}
-\frac{\lambda+\omega}{x} & - \partial_x \\
\partial_x & -\frac{\lambda-\omega}{x} 
\end{bmatrix}.
\label{formal}\end{equation}
When we consider \eqref{formal} as acting on distributions on $\R_+$,
we will call it the {\em formal operator}. In what follows we will define various realizations of this operator, with domain and range contained in $L^2(\R_+,\C^2)$, preferably closed. They will
have additional indices.

First consider its eigenequation for eigenvalue zero
\begin{equation}
D_{\omega,\lambda}\xi=0.\label{eqo}
\end{equation}
The space of distributions on $\R_+$ solving \eqref{eqo} will be denoted $\Ker(D_{\omega,\lambda})$. 
The following lemma shows that  $\Ker(D_{\omega,\lambda})$ consists of smooth solutions.
\begin{lemma} \label{reg_lemma}
Let $f$ be a distributional solution on $\R_+$ of the equation $f'(x)
= M(x)f(x)$ for some $M \in C^{\infty}(\R_+, \mathrm{End}(\C^n))$. Then $f$ is a smooth function.
\end{lemma}
\begin{proof}
Fix $x_0 \in \R_+$ and $\epsilon \in \left] 0 , \frac{x_0}{2} \right[$. We choose $\chi_2 \in C_c^{\infty}(\R_+)$ equal to $1$ on $[x_0 - 2 \epsilon, x_0 + 2\epsilon]$ and $\chi_1 \in C_c^{\infty}(\R_+)$ supported in $[x_0 - 2 \epsilon, x_0 + 2\epsilon]$ and equal to $1$ on $[x_0 - \epsilon, x_0 + \epsilon]$. Clearly $\chi_2 \chi_1 = \chi_1$ and $\chi_2 \chi_1' = \chi_1'$. Put $f_j = \chi_j f$ for $j=1,2$. Since $f_2$ is compactly supported, it~belongs to $H^s(\R_+,\C^n)$ for some $s \in \R$. We have $f_1 = \chi_1 f_2$, so also $f_1 \in H^s(\R_+,\C^n)$. Now evaluate
\begin{equation}
    f_1 ' = \chi_1' f + \chi_1 f' = \chi_2 \chi_1' f + \chi_2 \chi_1 M f = (\chi_1' + \chi_1 M) f_2.
\end{equation}
Since $\chi_1' +  \chi_1 M \in C_c^{\infty}(\R_+)$, this implies that $f_1 \in H^{s+1}(\R_+,\C^n)$. Next we may repeat this argument with $\frac{\epsilon}{2}$ playing the role of new $\epsilon$, $\chi_1$ as the new $\chi_2$ and arbitrarily chosen new~$\chi_1$. Then the new $f_1$ is in $H^{s+2}(\R_+, \C^n)$. Proceeding like this inductively we conclude that for every $s \in \R$ there exists $\chi \in C_c^{\infty}(\R_+)$ equal to $1$ on a neighbourhood of $x_0$ such that $\chi \, f$ belongs to $H^s(\mathbb R_+,\C^n)$. Taking $s> \frac{1}{2}+k$ we conclude from Sobolev embeddings that $f$ is of class $C^k$ on a~neighbourhood of $x_0$, perhaps after modifying it on a set of measure zero. Since this is true for every $k \in \mathbb N$ and every $x_0 \in \R_+$, $f$ is smooth.
\end{proof}

For $p\in \cM_{} $,
we introduce two types of solutions of \eqref{eqo}:
\begin{subequations}
\begin{align}\label{sol1}
\eta_p^\uparrow(x)&:=   \frac{x^{\mu}}{\omega+ \lambda} \begin{bmatrix}-\mu\\\omega+\lambda\end{bmatrix}=
    x^{\mu}\begin{bmatrix}z \\ 1\end{bmatrix}, \\
\label{sol2}
\eta_p^\downarrow(x)&:=  \frac{x^{\mu}}{\omega - \lambda}\begin{bmatrix}\omega-\lambda\\-\mu\end{bmatrix}=
  x^{\mu}\begin{bmatrix}1 \\ z^{-1} \end{bmatrix}.
\end{align}
\end{subequations}
They are nowhere vanishing meromorphic functions of $p$ for every $x$:
\begin{align*}
\cM \ni p \mapsto \eta_p^\uparrow(x) &\quad\text{has a pole on } \{ b = 0 \}, \\
  {\cM}_{} \ni p\mapsto \eta_p^\downarrow (x)&\quad\text{has a pole on } \{ a = 0 \}.
\end{align*}
On $\{ a \neq 0 \} \cap \{ b \neq 0 \}$ we have $\eta_p^{\downarrow} = z^{-1} \eta_p^{\uparrow}$.

There exist also exceptional solutions, defined only for $\mu=0$:
\begin{subequations}
\begin{align}
\label{sol3}
\vartheta_{\omega}^\uparrow(x)&:= -\ln(x) \begin{bmatrix}0\\ 2\omega\end{bmatrix}
+\begin{bmatrix}1\\0\end{bmatrix},\quad \omega-\lambda=0;\\
\label{sol4}
\vartheta_{ \omega}^\downarrow(x)&:=  -\ln(x) \begin{bmatrix}2\omega\\0\end{bmatrix}
+\begin{bmatrix}0\\1\end{bmatrix},\quad \omega+\lambda=0
  .\end{align}
  \end{subequations}

  The nullspace of $D_{\omega,\lambda}$, that is,
$\Ker(D_{\omega,\lambda})$ has the following bases:
\begin{align*}
\mu\neq0 : & \qquad\big(  \eta_{\omega,\lambda,\mu}^\uparrow,\quad 
              \eta_{\omega,\lambda,-\mu}^\uparrow \big)\quad \text{and}\quad
              \big(  \eta_{\omega,\lambda,\mu}^\downarrow,\quad 
              \eta_{\omega,\lambda,-\mu}^\downarrow \big),\\
\omega=\lambda
\neq0 : & \qquad\big(  \eta_{\omega,\omega,0}^\uparrow,\quad 
                                        \vartheta_{\omega}^\uparrow \big),\\
 \omega=-\lambda\neq0 : & \qquad  \big(  \eta_{\omega,- \omega,0}^\downarrow,\quad 
                                          \vartheta_{\omega}^\downarrow\big),\\
(
  \omega,\lambda)=(0,0): & \qquad \big(  \vartheta_{0}^\uparrow,\quad 
                                          \vartheta_{0}^\downarrow \big)  =  \left(\begin{bmatrix}1\\0\end{bmatrix},\quad \begin{bmatrix}0\\1\end{bmatrix}\right).
\end{align*}

The canonical bisolution of $D_{\omega,\lambda}$ \eqref{canonical} at $k=0$ takes the form
\begin{align}
  G_{\omega,\lambda}^\leftrightarrow(0;x,y)
  &=\frac12\begin{bmatrix}
    \frac{\omega-\lambda}{\mu} \Big(\Big(\frac{x}{y}\Big)^\mu-\Big(\frac{y}{x}\Big)^\mu\Big)&
    \Big(\frac{x}{y}\Big)^\mu+\Big(\frac{y}{x}\Big)^\mu\\[2ex]
    -\Big(\frac{x}{y}\Big)^\mu-\Big(\frac{y}{x}\Big)^\mu
    &-\frac{\omega+\lambda}{\mu} \Big(\Big(\frac{x}{y}\Big)^\mu-\Big(\frac{y}{x}\Big)^\mu\Big)
    \end{bmatrix}.
\end{align}

\subsection{Nonzero energy}

Now consider the eigenequation for the
eigenvalue $k \in \C^\times $:
\begin{equation}
  ( D_{\omega, \lambda} -k)f=0.
\label{equo}  \end{equation}

Acting on \eqref{equo} with $ D_{\omega, - \lambda}+k$ we obtain
\begin{equation}\begin{bmatrix}- \partial_x^2 + \frac{\omega^2 - \lambda^2}{x^2}
- \frac{2 \lambda k}{x} - k^2 &\frac{\omega  - \lambda}{x^2}\\[2ex]
\frac{\omega  + \lambda}{x^2}&- \partial_x^2 + \frac{\omega^2 - \lambda^2}{x^2}
- \frac{2 \lambda k}{x} - k^2
\end{bmatrix}f(x)=0.
\label{eq:Whit_mat}
\end{equation}
At first we focus on the case $\mu^2=\omega^2 - \lambda^2 \neq 0$, in
which 
$\begin{bmatrix}0 &\omega  - \lambda\\
\omega  + \lambda&0
\end{bmatrix}$
 is a diagonalizable matrix. Decomposing $f(x)$ in its eigenbasis
\begin{equation}
f(x) = f_+(x) \begin{bmatrix} \omega - \lambda \\ \mu \end{bmatrix} + f_-(x) \begin{bmatrix} \omega - \lambda \\ -\mu \end{bmatrix}
\label{eq:f_decomp}
\end{equation}
we find that functions $f_{\pm} (x)$ satisfy the Whittaker equations
\begin{equation}
    \left( - \partial_x^2 + \frac{\left( \mu \pm \frac{1}{2} \right)^2 - \frac{1}{4}}{x^2} - \frac{2 \lambda k}{x} - k^2 \right) f_{\pm}(x) =0.\\\
    \label{eq:Whit}
\end{equation}
This second order differential equation is satisfied by the Whittaker
functions \eqref{eq:def_I} and \eqref{eq:def_K}:
\begin{equation}
f_{\pm}(x) = c_{\pm, 1} \cI_{- \i \lambda, \mu \pm \frac{1}{2}}(2 \i kx) + c_{\pm, 2} \cK_{- \i \lambda, \mu \pm \frac{1}{2}}(2 \i kx). 
\label{eq:Whit_sol}
\end{equation}

For generic values of parameters, the four functions appearing in
\eqref{eq:Whit_sol} are linearly independent and thus
\eqref{eq:Whit_sol} is the general  solution of \eqref{eq:Whit_mat}. Inspection of its expansion for $x \to 0$ reveals that (again, for generic parameters) it is annihilated by $ D_{\omega, \lambda} -k$ if and only if
\begin{equation}
    \i \omega \, c_{-,1} = c_{+,1}, \qquad  \omega \, c_{-,2} = ( \lambda + \i \mu ) c_{+,2}.
    \label{eq:c_coeff}
\end{equation}

\begin{remark}
 Equation \eqref{equo} simplifies for $\mu =0$, but instead of treating it separately we will construct solutions valid on the whole $\cM$ by analytic continuation. For similar reasons we disregard non-generic cases mentioned above Equation \eqref{eq:c_coeff}.
\end{remark}

Let us introduce a family of solutions of the eigenequation \eqref{equo} defined for $k \in \C \backslash [0, \i \infty [$:
\begin{subequations}\label{xizeta-}
\begin{gather}
\xi_p^-(k,x) = \frac{\Gamma (1 + \mu + \i \lambda )}{2 \mu (\omega - \lambda + \i \mu)} \left( \i \omega \cI_{- \i \lambda, \mu + \frac12} (2 \i k x) \begin{bmatrix} \omega - \lambda \\ \mu \end{bmatrix} + \cI_{- \i \lambda, \mu - \frac12} (2 \i k x) \begin{bmatrix}
 \omega  - \lambda \\ - \mu
\end{bmatrix} \right), \label{eq:xi_def} \\
\zeta_{p}^-(k,x) = \frac{\omega \, \cK_{- \i \lambda , \mu +
\frac{1}{2}}(2 \i k x)}{\mu (\omega - \lambda - \i \mu)}     \begin{bmatrix}
\omega - \lambda \\ \mu 
\end{bmatrix} + \frac{(\lambda + \i \mu) \cK_{- \i \lambda , \mu -
\frac{1}{2}}(2 \i k x)}{\mu (\omega - \lambda - \i \mu)}  \begin{bmatrix}
\omega - \lambda \\ - \mu \label{eq:zeta_def}
\end{bmatrix}.
\end{gather}
\end{subequations}
As an alternative to the presented derivation, one may check directly that they satisfy \eqref{equo} using recursion relations from Appendix \ref{recur}.

The second family, defined for $k \in \C \backslash [0, - \i \infty[$, is obtained by reflection:
\begin{equation}
\xi_p^+(k,x) = \overline{\xi_{\overline{p}}^-(\overline k , x)}, \qquad \zeta_p^+(k,x) = \overline{\zeta_{\overline{p}}^-(\overline k , x)}.
\end{equation}
Explicit expressions in terms of Whittaker functions take the form
\begin{subequations}\label{xizeta+}
\begin{gather}
\xi_p^+(k,x) = \frac{\Gamma (1 + \mu - \i \lambda )}{2 \mu (\omega - \lambda - \i \mu)} \left( -\i \omega \cI_{ \i \lambda, \mu + \frac12} (-2 \i k x) \begin{bmatrix} \omega - \lambda \\ \mu \end{bmatrix} + \cI_{ \i \lambda, \mu - \frac12} (-2 \i k x) \begin{bmatrix}
 \omega  - \lambda \\ - \mu
\end{bmatrix} \right), \\
\zeta_{p}^+(k,x) = \frac{\omega \, \cK_{ \i \lambda , \mu +
\frac{1}{2}}(-2 \i k x)}{\mu (\omega - \lambda + \i \mu)}     \begin{bmatrix}
\omega - \lambda \\ \mu 
\end{bmatrix} + \frac{(\lambda - \i \mu) \cK_{ \i \lambda , \mu -
\frac{1}{2}}(-2 \i k x)}{\mu (\omega - \lambda + \i \mu)}  \begin{bmatrix}
\omega - \lambda \\ - \mu
\end{bmatrix}.
\end{gather}
\end{subequations}

\begin{lemma} \label{xi_zeta_props}
Let us fix $k,x$. $\xi^{+}_p(k,x)$ and $\xi^{-}_p(k,x)$ are
meromorphic functions of $p \in \cM$, nonsingular away from $\cE^+$
and $\cE^-$, respectively. $\zeta_p^+(k,x)$ and $\zeta_p^-(k,x)$ are
holomorphic functions on the whole $\cM$. Furthermore,
$\zeta_p^\pm(\cdot)$ 
satisfy $\zeta_p^\pm = \zeta_{\tau(p)}^\pm$, where $\tau$ was defined  
in \eqref{tau}, and are nonzero functions for every $p \in
\cM$.
\end{lemma}
\begin{proof} It is sufficient to prove the claim for the family with
  superscript minus. Meromorphic dependence on $p$ is
  clear. Definitions of $\xi_p^-$ and $\zeta_p^-$ can be manipulated to the form
\begin{subequations}
\begin{gather}
\xi_{p}^-(k,x)   = \i \, \cI_{- \i \lambda, \mu + \frac{1}{2}}(2 \i k x) \frac{1}{ N_p^-} \begin{bmatrix} -1 \\ z \end{bmatrix} + \frac{\cI_{- \i \lambda, \mu - \frac{1}{2}}(2 \i k x) - \i \lambda \, \cI_{- \i \lambda,\mu + \frac{1}{2}}(2 \i k x)}{\mu} \frac{1}{ N_p^-}\begin{bmatrix} z \\ 1 \end{bmatrix}, \label{eq:xi_reg} \\
\zeta_{p}^-(k,x)  = \cK_{- \i \lambda, \mu + \frac{1}{2}} (2 \i k
x) \begin{bmatrix} \i \\ 1 \end{bmatrix} - \frac{(\omega+\lambda) (z +
  \i)}{2} \frac{\cK_{- \i \lambda, \mu - \frac{1}{2}}(2 \i k x) -
  \cK_{- \i \lambda, \mu + \frac{1}{2}}(2 \i k x)
}{\mu} \begin{bmatrix} z \\ 1 \end{bmatrix} \label{eq:zeta_reg}\\
 = \cK_{- \i \lambda, \mu -\frac{1}{2}} (2 \i k
x) \begin{bmatrix} \i \\ 1 \end{bmatrix} - \frac{(\omega+\lambda) (-z +
  \i)\big(\cK_{- \i \lambda, \mu - \frac{1}{2}}(2 \i k x) -
  \cK_{- \i \lambda, \mu + \frac{1}{2}}(2 \i k x)\big)
}{2\mu} \begin{bmatrix} -z \\ 1 \end{bmatrix}. \label{eq:zeta_reg-}
\end{gather}
\end{subequations}
Functions $\mu \mapsto \frac{\cI_{- \i \lambda, \mu - \frac{1}{2}}(2
  \i k x) - \i \lambda \, \cI_{- \i \lambda,\mu + \frac{1}{2}}(2 \i k
  x)}{\mu}$ and $\mu \mapsto \frac{\cK_{- \i \lambda, \mu -
    \frac{1}{2}}(2 \i k x) - \cK_{- \i \lambda, \mu + \frac{1}{2}}(2
  \i k x) }{\mu}$ have removable singularities at $\mu = 0$, as seen
from identities \eqref{eq:I_sym}, \eqref{eq:K_sym}. Therefore $\zeta_p^-(k,x)$ is regular for $z \neq \infty$. If in addition $p \notin \cE^-$, then also $(N_p^-)^{-1}$ and hence $\xi_p^-$ is nonsingular.  

Next write $(\omega+\lambda) (z+ \i) \begin{bmatrix} z \\
  1 \end{bmatrix} = (\omega - \lambda)  (1+ \i z^{-1}) \begin{bmatrix}
  1 \\ z^{-1} \end{bmatrix}$ in \eqref{eq:zeta_reg}. Then it is clear
that $\zeta_p^-(k,x)$ is nonsingular for $z=\infty$.
  
Moreover,
$\frac{z}{ N_p^-} = \Gamma(1 + \mu + \i \lambda)(1- \i z^{-1})^{-1}$ is nonsingular for $z=\infty$ if $p\not\in\cE_0^-$. Hence
$\xi_{p}^-(k,x) $ is nonsingular for $z=\infty$ if $p \notin \cE^-$.

The statement about the symmetry $\mu\to-\mu$ follows from the
comparison of \eqref{eq:zeta_reg} and \eqref{eq:zeta_reg-}.
The last claim follows from \eqref{eq:zeta_asymptotic} below.
\end{proof} 
\begin{remark}
Proof of Lemma \ref{xi_zeta_props} shows that functions
$\frac{\xi_p^{\pm}(k,x)}{\Gamma(1+ \mu \mp \i \lambda)}$ are singular
on subsets of $\cM$ smaller than $\cE^{\pm}$, namely
  on $\cE_0^\pm$.
 Moreover,  $\frac{(\omega - \lambda \mp \i \mu) \xi_p^{\pm}(k,x)}{\Gamma(1+ \mu
    \mp \i \lambda)}$ is holomorphic everywhere on $\cM$, however it vanishes on $\cZ \cup \{ a = 0 \}$.
\end{remark}

Near the origin, $\xi_p^\pm$ has the leading term proportional to $(kx)^\mu$, except for $2 \mu +1 \in - \mathbb N$:
\begin{equation}
\xi_p^\pm(k,x) \sim \frac{1}{N_p^{\pm}} \frac{(\mp 2 \i k x)^{\mu}}{\Gamma(2 \mu +1)}  \begin{bmatrix} z \\ 1 \end{bmatrix}+ O(( k x)^{\mu+1})
\label{eq:xi_asymptotic}
\end{equation}
If $k \in \C_\pm$, it grows exponentially at infinity:
\begin{equation}
\xi_p^\pm(k,x) \sim  \frac{1}{2} \e^{\mp \i k x} (\mp 2 \i k x)^{\mp \i \lambda} \begin{bmatrix} 1 \\ \mp \i \end{bmatrix}+ O(\e^{\mp \i k x}( kx)^{\mp \i \lambda -1}).
\label{eq:xi_large_x}
\end{equation}
Under the same assumption, $\zeta_p^\pm$ is exponentially decaying:
\begin{equation}
\zeta_p^\pm(k,x) \sim \e^{\pm \i k x} (\mp 2 \i k x)^{\pm \i \lambda}  \begin{bmatrix} \mp \i \\ 1 \end{bmatrix}+O(\e^{\pm \i k x} ( kx)^{\pm \i \lambda -1}).
\label{eq:zeta_asymptotic}
\end{equation}
Behaviour of this function near the origin is much more complicated, see \eqref{eq:K_small_arg}. Here we note only that for $\RE(\mu) >0$ one has
\begin{equation}
\zeta_{p}^\pm(k,x) \sim   \frac{N_p^{\pm}}{2} \Gamma(2 \mu + 1) (\mp 2 \i k x )^{- \mu}  \begin{bmatrix} -1 \\ z^{-1} \end{bmatrix}  + o (( kx)^{- \mu}).
\label{eq:zeta_small_x}
\end{equation}

For $k \in \C \backslash \i \R$ both families of solutions are defined. The following lemma provides relations between them. It is convenient to introduce $\varepsilon_k = \sgn(\RE(k))$, which distinguishes connected components of $\C \backslash \i \R$.

\begin{lemma}
For every $k \in \C \backslash \i \R$ we have
\begin{subequations}
\begin{align}
\xi_p^+(k,x)& = \e^{- \i \varepsilon_k \pi \mu} S_p \, \xi_p^- (k,x),
  \label{eq:xipm}\\
  \xi_p^+(k,x)&=\frac{\e^{-\varepsilon_k \pi\lambda}}{2\i}\big(\zeta_p^-(k,x)-\e^{-\i \varepsilon_k \pi\mu}S_p\zeta_p^+(k,x)\big) \label{eq:xipzz} ,\\
  \xi_p^-(k,x)&=\frac{\e^{-\varepsilon_k \pi\lambda}}{2\i}\big(\e^{\i\varepsilon_k \pi\mu}S_p^{-1}\zeta_p^-(k,x)-\zeta_p^+(k,x)\big), \label{eq:ximzz} \\
  \zeta_p^\pm(k,x) &= \mp 2 \i \e^{\varepsilon_k \pi \lambda} \xi_p^\mp(k,x) + \e^{\pm \i \varepsilon_k \pi \mu} (S_p)^{\mp 1} \zeta_p^\mp(k,x). \label{eq:zxiz}
\end{align}
\end{subequations}
\end{lemma}
\begin{proof}
Equation \eqref{eq:xipm} follows immediately from \eqref{eq_miracle}. To derive \eqref{eq:xipzz}, we express $\xi_p^\pm$ and $\zeta_p^\pm$ in terms of trigonometric Whittaker functions and use the connection formula \eqref{eq_missing}. Then \eqref{eq:ximzz} is obtained by reflection or by combining with \eqref{eq:xipm}. Equation \eqref{eq:zxiz} is obtained from \eqref{eq:xipzz} and \eqref{eq:ximzz} by inverting and multiplying $2 \times 2$ matrices.
\end{proof}


\begin{lemma}
$\xi_p^{\pm}$ and $\xi_{\tau(p)}^\pm$, two
  eigenvectors of the monodromy, can be used to express 
$\zeta_p^{\pm}$:
\begin{equation}
\zeta_p^\pm(k,x) = - \frac{2 \pi \omega}{\Gamma(1 + \mu \mp \i \lambda) \Gamma(1 - \mu \mp \i \lambda)} \frac{\xi_p^\pm(k,x) - \xi_{\tau(p)}^\pm(k,x)}{\sin(2 \pi \mu)}. \label{eq:zxixi}
\end{equation}
The analytic continuation of $\zeta_p^\pm$ along a loop winding around the origin counterclockwise gives
\begin{equation}
\zeta_p^\pm (\e^{2 \pi \i} k , x) = \e^{-2 \pi \i \mu} \zeta_p^\pm(k,x) - \frac{4 \pi \i \omega}{\Gamma(1 + \mu \mp \i \lambda) \Gamma(1 - \mu \mp \i \lambda)} \xi_p^\pm(k,x). \label{eq:zeta_mono}
\end{equation}
\end{lemma}
\begin{proof}
Relation \eqref{eq:zxixi} may be derived from \eqref{eq:def_K}. Then \eqref{eq:zeta_mono} follows immediately.
\end{proof}

\begin{lemma} \label{det_lemma}
The following relations hold:
\begin{subequations}
\begin{align}
\det \begin{bmatrix}
\xi_{p}^\pm(k,x) & \zeta_{p}^\pm(k,x)
\end{bmatrix} & =1, \label{eq:det1} \\
\det \begin{bmatrix}
\zeta_{p}^+(k,x) & \zeta_{p}^-(k,x)
\end{bmatrix} & = -2\i\e^{\varepsilon_k\pi\lambda}. \label{eq:det2}
\end{align}
\end{subequations}
In particular, $\xi_p^\pm(k, \cdot), \zeta_p^\pm(k, \cdot)$ form a basis of solutions of $\eqref{equo}$ for $p \notin \cE^\pm$ and $k \notin [0, \mp \i \infty [$, while $\zeta_p^+(k, \cdot)$ and $\zeta_p^-(k, \cdot)$ form a basis whenever $k \notin \i \R$.
\end{lemma}
\begin{proof}
Equation \eqref{equo} may be rewritten in the 
form $f'(x) = M(x) f(x)$, where $M(x)$ is a~traceless
matrix. Therefore for any two solutions $f,g$ the determinant $\det \begin{bmatrix} f(x) & g(x) \end{bmatrix}$ is independent of~$x$. To calculate it for $f = \xi_{p}^\pm(k,\cdot)$, $g = \zeta_{p}^\pm(k, \cdot )$, we use their asymptotic forms for $x \to 0$. By holomorphy, it~is sufficient to carry out the computation for $\RE(\mu) > 0$. Then we may use \eqref{eq:xi_asymptotic} and \eqref{eq:zeta_small_x}. To obtain \eqref{eq:det2}, we combine \eqref{eq:xipzz} with \eqref{eq:det1}.
\end{proof}

We remark that restrictions on $k$ in Lemma \ref{det_lemma} may be
omitted if the functions $\xi_p^\pm$ and $\zeta_p^\pm$ are analytically continued in suitable way.


\begin{lemma} \label{exceptional_sol}
The following relation holds for $p \in \cE^{\pm}$:
\begin{equation}
\zeta^{\pm}_p(k,x) = \mp 2 \i \e^{\mp \i \pi (\mu \mp \i \lambda)} (S_p)^{\mp 1} \xi^{\pm}_p(k,x).
\label{eq:zeta_xi_dep}
\end{equation}
In particular $(S_p)^{\mp 1} \xi^{\pm}_p(k,x)$ is nonsingular on $\cE^{\pm}$. 
\end{lemma}
\begin{proof}
It is sufficient to consider the lower sign. If $p \in \cE^-$, then either $1 + \mu + \i \lambda \in - \mathbb N$ or $z=\i$ (and hence $\mu + \i \lambda =0$). In the former case we use \eqref{eq:Laguerre} for both terms in \eqref{eq:zeta_def}. In~the latter case \eqref{eq:Laguerre} may be used only for the second term, but the first term in both \eqref{eq:xi_def} and \eqref{eq:zeta_def} vanishes. This establishes \eqref{eq:zeta_xi_dep}. 
\end{proof}

The following function will be called the {\em two-sided Green's
  kernel}. It  is defined if $k \in \C_\pm$ and $p \notin \cE^{\pm}$:
\begin{equation}
G_p^{\bowtie}(k;x,y) = 
 -\one_{\R_+}(y -x) \xi_{p}^\pm(k,x) \zeta_{p}^\pm(k,y)^\T - \one_{\R_+}(x-y ) \zeta_{p}^\pm(k,x) \xi_{p}^\pm(k,y)^\T.
 \label{eq:Gdef}
\end{equation}  
It is a holomorphic function of $p \in \cM \backslash \cE^{\pm}$ satisfying
\begin{equation}
    G_{p}^{\bowtie}(k;x,y ) = G_{p}^{\bowtie}(k;y ,x)^\T \qquad \text{and} \qquad \overline{G_{p}^{\bowtie}(k,x,y )} = G_{\overline{\vphantom{k}p}}^{\bowtie}(\overline{k};x,y ).
    \label{eq:GK_sym}
\end{equation}
Later on, with some restrictions on parameters, it will be interpreted
as the resolvent of certain closed realizations of $D_p$.

\section{Minimal and maximal operators} \label{sec:minmax}

We consider the operator
\begin{equation}
 D_{\omega, \lambda} = 
\begin{bmatrix}
-\frac{\lambda+\omega}{x} & - \partial_x \\
\partial_x & -\frac{\lambda-\omega}{x} 
\end{bmatrix}.
\end{equation}
on distributions on $\R_{+} = ]0,\infty[$ valued in $\C^2$. We will
construct out of it several densely defined operators on $L^2(\R_+,\C^2)$.

Firstly, we let $D_{\omega, \lambda}^\mathrm{pre}$
be the restriction of $ D_{\omega, \lambda}$ to
$C_c^{\infty}=C_c^{\infty}(\R_+, \C^2)$, called the {\em preminimal
  realization of $D_{\omega,\lambda}$}. We have $D_{ \overline
  {\vphantom{\lambda}\omega}, \overline \lambda}^\mathrm{pre} \subset
D_{\omega, \lambda}^{\mathrm{pre}*}$, so~$D_{\omega,
  \lambda}^{\mathrm{pre}*}$ is densely defined. Thus $D_{\omega,
  \lambda}^\mathrm{pre}$ is closable. Its closure will be denoted by
$D_{\omega, \lambda}^{\min}$. Next, $D_{\omega, \lambda}^{\max}$ is
defined as the restriction of $ D_{\omega, \lambda}$ to $\Dom
(D_{\omega, \lambda}^{\max}) = \{ f \in L^2(\R_+,\C^2) \, | \,  D_{\omega,
  \lambda} f \in L^2(\R_+,\C^2) \}$. It is easy to check that $D_{\omega,
  \lambda}^{\min} \subset D_{\omega,\lambda}^{\max} = D_{ \overline
  {\vphantom{\lambda}\omega}, \overline
  \lambda}^{\mathrm{pre}*}$. Furthermore, $\overline{D_{\omega,
      \lambda}^\mathrm{pre}} = D_{ \overline
    {\vphantom{\lambda}\omega}, \overline \lambda}^\mathrm{pre}$ and
  analogously for $D_{\omega, \lambda}^{\min}$ and $D_{\omega,
    \lambda}^{\max}$. As a~consequence, 
  \begin{equation}
D_{\omega,\lambda}^{\min\T}=D_{\omega,\lambda}^{\max},\qquad D_{\omega,\lambda}^{\max\T}=D_{\omega,\lambda}^{\min}.
    \end{equation}
Operators $D_{\omega, \lambda}^\pre, D_{\omega,
  \lambda}^{\min}$ and $D_{\omega, \lambda}^{\max}$ are all
homogeneous of order $-1$.

We choose $\mu \in \C$ satisfying $\mu^2 = \omega^2 - \lambda^2$. Note that in general $\mu$ is not uniquely determined by $\omega, \lambda$. For the moment
it does not matter which one we take.

\begin{theorem} \label{domain_comparison}
\begin{enumerate}\item
    If $|\RE(\mu)|\geq\frac12$, then
\begin{equation}  
  \Dom 
(D_{\omega, \lambda}^{\max})=\Dom 
(D_{\omega, \lambda}^{\min}).
\end{equation}
\item If  $|\RE(\mu)|<\frac12$, then
\begin{equation}  
  \dim\Dom
(D_{\omega, \lambda}^{\max})/\Dom
(D_{\omega, \lambda}^{\min})=2.
\end{equation}
Besides, if $\chi\in C_\mathrm{c}^\infty([0,\infty[)$ equals $1$ near $0$,
then
  \begin{equation}\Dom
(D_{\omega, \lambda}^{\max})=\Dom(
D_{\omega, \lambda}^{\min})+ \{f\chi \, | \, f\in\Ker(D_{\omega, \lambda})\}
.\label{domain}\end{equation}
\end{enumerate}
\end{theorem}

We will prove the above theorem in the next section. Now we would like to discuss its consequences. If $|\RE(\mu)| < \frac{1}{2}$, we are especially interested in operators $D_{\omega, \lambda}^\bullet$ satisfying
\begin{equation}
D_{\omega, \lambda}^{\min}\subsetneq D_{\omega, \lambda}^\bullet\subsetneq
D_{\omega, \lambda}^{\max}.
\label{eq:intermediate}
\end{equation}
By the above theorem, they are in $1-1$ correspondence with rays in
$\Ker(D_{\omega, \lambda})$.

More precisely, let $f\in \Ker(D_{\omega, \lambda})$, $f\neq0$. Define $D_{\omega, \lambda}^f$ as the restriction of
  $D_{\omega, \lambda}^{\max}$ to
 \begin{equation}  
  \Dom(D_{\omega, \lambda}^f):=\Dom(D_{\omega, \lambda}^{\min})+\C f\chi.
  \end{equation}
Then $D_{\omega, \lambda}^f$ is independent of the choice of $\chi$ and satisfies
\begin{equation}
D_{\omega, \lambda}^{\min}\subsetneq D_{\omega, \lambda}^f\subsetneq
D_{\omega, \lambda}^{\max}. 
\label{eq:Df_strict_inc}
\end{equation}
Every $D_{\omega, \lambda}^{\bullet}$ satisfying \eqref{eq:intermediate} is of the form $D_{\omega, \lambda}^f$ for some $f$ and we have $D_{\omega, \lambda}^f = D_{\omega, \lambda}^g$ if and only if $f$ and $g$ are proportional to each other.

We will now investigate the domain of the
minimal operator. Note that if we know the domain of $D_{\omega, \lambda}^{\min}$, then the domain of $D_{\omega, \lambda}^{\max}$ is also
known from Theorem \ref{domain_comparison}. From now on we do not use this result until its proof is presented.

The following two facts are well-known:

\begin{lemma}{Hardy's inequality:} If $ f \in H_0^1(\R_+)$, then
\begin{equation}
\int_{0}^{\infty} \frac{|f(x)|^2}{x^2} \D x \leq 4 \int_{0}^{\infty} |f'(x)|^2 \D x.
\end{equation}
\end{lemma}

\begin{lemma} \label{closed_lemma}
If~$R, S$ are closed operators such that $R$ has bounded inverse, then $RS$ is closed.
\end{lemma}

The above two lemmas are used in the following
  characterization of the minimal domain:

\begin{proposition} \label{dom_prop}
$H_0^1(\R_+, \C^2) \subset \Dom(D_{\omega, \lambda}^{\min})$, with an equality if $|\RE (\mu)| \neq \frac{1}{2}$.
\end{proposition}
\begin{proof}
The inclusion follows from Hardy's inequality. To prove the second part of the statement, we use Lemma \ref{closed_lemma}. Consider $R = A -M_{\omega, \lambda}$, $S = \frac{\sigma_2}{x}$, where $M_{\omega, \lambda} = \begin{bmatrix} \frac{\i}{2} & - \i \lambda -\i \omega \\ \i \lambda - \i \omega & \frac{\i}{2} \end{bmatrix}$. $R$~is a~bounded perturbation of $A$, so~$\Dom(R)= \Dom(A)$ and $R$ is closed, while $S$ is self-adjoint on the domain $\Dom(S) = \left \{ f \in L^2(\R_+,\C^2) \, | \, \frac{1}{x} f(x) \in L^2(\R_+,\C^2) \right \}$. One checks that $RS = \left. D_{\omega, \lambda} \right|_{\Dom(RS)}$. Next we show that $\Dom(RS) = H_0^1(\R_+,\C^2)$.

If $f \in \Dom(RS)$, then $x \mapsto \frac{f(x)}{x}$ belongs to $ \Dom(A)$. Since $x \partial_x \frac{f(x)}{x} = f'(x) - \frac{f(x)}{x}$, this entails that $f' \in L^2$. Thus $f \in H^1(\R_+, \C^2) \cap \Dom(S) = H_0^1(\R_+,\C^2)$. Conversely, if $f \in H_0^1(\R_+,\C^2)$, then $f \in \Dom(S)$ by Hardy's inequality, while the last computation implies that $Sf \in \Dom(A)$. Thus $f \in \Dom(RS)$. 

We have shown that $\Dom(RS) = H_0^1(\R_+,\C^2)$, which is dense in $\Dom(D_{\omega, \lambda}^{\min})$ with the graph topology. Thus $D_{\omega, \lambda}^{\min}$ is the closure of $RS$. We have to check that $R$ has bounded inverse.

If $\mu \neq 0$, then $M_{\omega, \lambda}$ is a diagonalizable matrix
with eigenvalues $c_{\pm}=\frac{\i}{2} \pm \i \mu$, which have nonzero
imaginary part if $|\RE(\mu)| \neq \frac{1}{2}$. Therefore the operator $A- M_{\omega, \lambda}$ is similar to $A - \begin{bmatrix} c_+ & 0 \\ 0 & c_- \end{bmatrix}$, which clearly is boundedly invertible. If $\mu =0$, then $N_{\omega,\lambda}=M_{\omega, \lambda} - \frac{\i}{2}$ is a nilpotent matrix, $N_{\omega, \lambda}^2=0$. Therefore $(A- M_{\omega, \lambda})^{-1} = (A - \frac{\i}{2})^{-1} + (A - \frac{\i}{2})^{-1} N_{\omega, \lambda} (A - \frac{\i}{2})^{-1}$. 
\end{proof}

\begin{corollary}
$D_{\omega, \lambda}^{\min}$ and $D_{\omega, \lambda}^{\max}$ are holomorphic families of closed operators for $|\RE(\mu)| \neq \frac12$.
\end{corollary}
\begin{proof}
Away from the set $|\RE(\mu)| = \frac12$, the operators $D_{\omega,
  \lambda}^{\min}$ have a constant domain. By Hardy's inequality,
$D_{\omega, \lambda} f$ is a holomorphic family of elements of $L^2(\R_+,\C^2)$
for any $f \in H_0^1(\R_+,\C^2)$. Hence $D_{\omega, \lambda}^{\min}$ form a
holomorphic family of bounded operators $H_0^1(\R_+,\C^2) \to L^2(\R_+,\C^2)$. The claim for
$D_{\omega, \lambda}^{\max}$ follows by taking adjoints (see e.g. Theorem 3.42 in \cite{DeWr}).
\end{proof}

We denote by $\sigma_\pp(B)$ the {\em point spectrum} of an operator $B$,
that is
\begin{equation}
\sigma_\pp(B):=\{k\in\C \, | \, \dim(\Ker(B-k)\geq1\}.
\end{equation}
If $\dim(\Ker(B-k))=1$, we say that $k$ is a {\em nondegenerate
eigenvalue}. 

In the following proposition we give a complete description of the
point spectrum of the maximal and minimal operator. With no loss of
generality, we assume that 
$\RE(\mu) > - \frac{1}{2}$. Note that the definition of $\cE^\pm$ is not symmetric with respect to $\mu \mapsto - \mu$!

\begin{proposition} \label{eig_prop}
One of the following mutually exclusive statements is true:
\begin{enumerate}
\item $\RE(\mu) \geq \frac{1}{2}$ and 
$(\omega,\lambda,\mu)\in\cE^\pm$. Then 
\[\sigma_\pp  (D_{\omega,\lambda}^{\max})=\sigma_\pp  (D_{\omega,\lambda}^{\min})=\C_\pm.\]
\item $\RE(\mu) \geq \frac{1}{2}$ and 
  $(\omega,\lambda,\mu)\not\in
    \cE^+\cup\cE^-$. Then 
\[\sigma_\pp  (D_{\omega,\lambda}^{\max})=\sigma_\pp  (D_{\omega,\lambda}^{\min})=\emptyset.\]
    \item $\RE(\mu) < \frac{1}{2}$ and $|\IM(\lambda)| \leq
      \frac{1}{2}$. Then
\[\sigma_\pp  (D_{\omega,\lambda}^{\max})=\C\backslash\R,\qquad\sigma_\pp  (D_{\omega,\lambda}^{\min})=\emptyset.\]
    \item $\RE(\mu) < \frac{1}{2}$ and $|\IM(\lambda)| >
      \frac{1}{2}$. Then
\[\sigma_\pp  (D_{\omega,\lambda}^{\max})=\C^\times ,\qquad\sigma_\pp  (D_{\omega,\lambda}^{\min})=\emptyset.\]
\end{enumerate}
Besides, all eigenvalues of $D_{\omega,\lambda}^{\max}$ and
$D_{\omega,\lambda}^{\min}$ are nondegenerate.
\end{proposition}
\begin{proof}
The four possibilities listed above are clearly mutually exclusive and cover all cases. Indeed, case $p \in \cE^+ \cap \cE^-$ is ruled out by Lemma \ref{eplus_eminus_intersection}.

By Lemma \ref{reg_lemma}, every $f \in \Ker (D_{\omega, \lambda}^{\max} -k)$ is a smooth function satisfying the differential equation $( D_{\omega, \lambda}-k) f =0$, in~which derivatives may be understood in the classical sense. Space of solutions of this equation is two-dimensional. 

By discussion in Section \ref{eigen}, there exist no nonzero solutions in $L^2(\R_+,\C^2)$ for $k=0$. In the remainder of the proof we restrict attention to $k \neq 0$.

First suppose that $\RE(\mu)\geq\frac12$. If $p \notin \cE^+ \cup \cE^-$, then $\xi_p^+$ (as well as $\xi_p^-$) is the unique up to scalars solution square integrable in a neighbourhood of zero, since other solutions have leading term proportional to $x^{- \mu}$. It is not in $L^2(\R_+,\C^2)$. Now let $p \in \cE^{\pm}$. If $\pm \IM(k) \leq 0$, we can argue in the same way using function $\xi_p^\mp$. In the case $k \in \C_\pm $ solution $\zeta_p^{\pm}$ is square integrable, whereas solutions not proportional to it grow exponentially at infinity. If $\RE(\mu) > \frac{1}{2}$, then we have also $\zeta_p^\pm \in H_0^1(\R_+,\C^2) \subset \Dom(D_{\omega, \lambda}^{\min})$. 

If $\RE(\mu) = \frac{1}{2}$, then $\zeta_p^{\pm} \notin H_0^1(\R_+,\C^2)$. We will now show that nevertheless $\zeta_p^{\pm} \in \Dom(D_{\omega, \lambda}^{\min})$. We define $\zeta_{p,\epsilon}^\pm(k,x) = \min \{ x^{\epsilon} ,1 \} \zeta_p^\pm(k,x)$ for $\epsilon >0$. Then $\zeta_{p,\epsilon}^\pm \in H_0^1(\R_+,\C^2) \subset \Dom (D_{\omega, \lambda}^{\min})$. We will show that $\zeta_{p,\epsilon}^\pm$ converges to $\zeta_p^\pm$ in the graph topology of $\Dom(D_{\omega, \lambda}^{\max})$ as $\epsilon 
\to 0$. Indeed, convergence in $L^2(\R_+,\C^2)$ is clear. Furthermore,
\begin{equation}
D_{\omega, \lambda}^{\min} \zeta_{p,\epsilon}^\pm(k,x) = k \zeta^\pm_{p,\epsilon}(k,x) + \epsilon \, \one_{[0,1]}(x) x^{\epsilon-1} \begin{bmatrix} 0 & -1 \\ 1 & 0 \end{bmatrix} \zeta_p^\pm(k,x),
\end{equation}
where $\one_{[0,1]}$ is the characteristic function of $[0,1]$. The first term converges to $k \zeta_p^\pm = D_{\omega, \lambda}^{\max} \zeta_p^\pm$. We show that the second term converges to zero by estimating
\begin{subequations} \label{eq:dom_min_calc}
\begin{gather}
 \int_0^{\infty} \left|  \epsilon \, \one_{[0,1]}(x) x^{\epsilon-1} \begin{bmatrix} 0 & -1 \\ 1 & 0 \end{bmatrix} \zeta_p^\pm(k,x) \right|^2 \D x = \epsilon^2 \int_0^1 \frac{|\zeta_p^\pm(k,x)|^2}{x} x^{2 \epsilon -1} \D x \tag{\ref{eq:dom_min_calc}}  \\
 \leq \epsilon^2 \sup_{y \in [0,1]}  \frac{|\zeta_p^\pm(k,y)|^2}{y} \cdot \int_0^1 x^{2 \epsilon -1} \D x = \frac{\epsilon}{2} \sup_{y \in [0,1]}  \frac{|\zeta_p^\pm(k,y)|^2}{y} . \nonumber
\end{gather}
\end{subequations}

Now suppose that $|\RE(\mu) |< \frac{1}{2}$. Then all solutions are square integrable in a neighbourhood of the origin, but they do not belong to $H_0^1(\R_+,C^2) = \Dom(D_{\omega, \lambda}^{\min})$. If $k \in \C_\pm$, then $\zeta_p^\pm$ is square integrable and solutions not proportional to it grow at infinity.

It only remains to consider the case of nonzero $k \in \R$. There exist solutions with leading terms for $x \to \infty$ proportional to $\e^{- \i k x} (kx)^{- \i \lambda}$ and $\e^{\i k x} (kx)^{ \i \lambda}$. If $|\IM(\lambda)| > \frac{1}{2}$, then one of these two is square integrable.
\end{proof}

We note that Proposition \ref{eig_prop} partially describes also ranges of $D_{\omega,\lambda}^{\min}$ and $D_{\omega,\lambda}^{\max}$, since
\begin{gather}
\Ran(D_{\omega, \lambda}^{\min}-k)^{\tperp} =  \Ker(D_{\omega, \lambda}^{\max}-k), \\
\Ran(D_{\omega, \lambda}^{\max}-k)^{\tperp} =  \Ker(D_{\omega, \lambda}^{\min}-k),
\end{gather}

\begin{corollary} \label{empty_res}
Operators $D_{\omega, \lambda}^{\substack{\min}}$ and $D_{\omega, \lambda}^{\substack{\max}}$ have empty resolvent sets if $|\RE (\mu)| < \frac{1}{2}$.
\end{corollary}

\section{Homogeneous realizations and the resolvent} \label{sec:homog}

\subsection{Definition and basic properties}

We consider the following open subset of $\cM$:
\begin{equation}
\cM_{-\frac{1}{2}}:=\left\{ p \in \cM \ | \ \RE(\mu) > - \frac{1}{2} \right\}.
\end{equation}
As before, choose $\chi\in C_\mathrm{c}^\infty(\R_+)$ equal to $1$ near $0$.
If $p = (\omega, \lambda, \mu,[a:b]) \in \cM_{- \frac12}$, we define
$D_p$ to be the restriction of $D_{\omega,\lambda}^{\max}$ to
\begin{equation}
\Dom(D_p):=\Dom(D_{\omega,\lambda}^{\min})+\C x^\mu \begin{bmatrix}
    a \\ b \end{bmatrix} \chi.
    \label{eq:Dp_dom}
    \end{equation}
This definition is correct because $x^{\mu} \begin{bmatrix} a \\ b \end{bmatrix}\chi$ is an element of $\Dom(D_{\omega, \lambda}^{\max})$ for $\RE(\mu) > - \frac{1}{2}$. If~$\RE(\mu) \geq \frac12$, then it belongs to $\Dom(D_{\omega, \lambda}^{\min})$, so we have $D_p = D_{\omega, \lambda}^{\min}$.

  \begin{theorem}\label{thoe}
    Let $p\in\cM_{{}-\frac12}$.  Then the operator $D_p$ does not
    depend on the choice of $\chi$,
    is closed, self-transposed and
\begin{equation}
\sigma(D_p) = \begin{cases}
\R & \text{for } p \notin \cE^+ \cup \cE^-, \\
\C_\pm\cup\R & \text{for } p \in \cE^{\pm},
\end{cases} \qquad \qquad 
\sigma_\pp(D_p) = \begin{cases}
\emptyset & \text{for } p \notin \cE^+ \cup \cE^-, \\
\C_\pm & \text{for } p \in \cE^{\pm}.
\end{cases}
\end{equation}    
    If $\pm \IM(k) >0$ and $p\not\in\cE^{\pm}$,
    then the integral kernel $G_p^{\bowtie}(k;x,y)$ introduced in  \eqref{eq:Gdef} defines
    a~bounded operator $G_p^{\bowtie}(k)$
    and
\begin{equation}
G_p^{\bowtie}(k)=(D_p-k)^{-1}.
\end{equation}
For $k\in\C_\pm$, the map
 $\cM_{{}-\frac12} \backslash\cE^\pm\ni p\mapsto (D_p-k)^{-1}$ is a
 holomorphic family of bounded operators.

 Therefore,
  $\cM_{{}-\frac12} \ni p\mapsto D_p$ is a holomorphic family of closed operators.
\end{theorem}

\begin{proof}
It is sufficient to consider the case $\IM(k)<0$. Let $p \notin
\cE^-$. We prove the boundedness separately for the integral operators with kernels $G^{\bowtie}_{p}(k)$ restricted to four regions forming a~partition of $\R_+ \times \R_+$ (up to an inconsequential overlap on a set of measure zero). Throughout the proof we use notation $x_< = \min \{ x, y \}$, $x_> = \max \{ x, y \}$. Symbols $c_{p}, c_{p}'$ will be used for positive constants which are locally bounded functions of $p$.

First we consider the region $x,y \leq |k|^{-1}$. Inspecting the asymptotics of Whittaker functions for small argument we conclude that $|G^{\bowtie}_{p}(k;x,y)| \leq c_{p} (|k|x_<)^{\RE(\mu)} (|k| x_>)^{-|\RE(\mu)|}$. Using this inequality and elementary integrals we estimate
\begin{equation}
    \int_{\left[ 0, |k|^{-1} \right]^2} |G^{\bowtie}_{p}(k;x,y)|^2 \D x \D y  \leq \frac{c_{p}^2}{|k|^2} \frac{1}{1 + 2 \RE(\mu)} \frac{1}{1+ \RE(\mu) - |\RE(\mu)|}.
\end{equation}
Therefore the Hilbert-Schmidt norm of the corresponding operator is bounded by $\frac{c_p'}{|k|}$.

Next, in the region $y \leq |k|^{-1} \leq x$ we have $|G^{\bowtie}_{p}(k)| \leq c_{p} (|k| y)^{\RE(\mu)} (|k|x)^{\IM(\lambda)} \e^{-|\IM(k)| x} $. Thus
\begin{equation}
 \int \limits_{[|k|^{-1}, \infty[ \times [0,|k|^{-1}]} |G^{\bowtie}_{p}(k;x,y)|^2 \D x \D y  \leq \frac{c_{p}^2}{|k|^2} \int_{[1, \infty[ \times [0,1]} \e^{- 2 \frac{|\IM(k)|}{|k|} t} t^{2 \, \IM(\lambda)} t'^{2 \RE(\mu)}  \D t \D t',
\end{equation}
which is a convergent integral depending continuously on $\lambda, \mu$. Again, the corresponding operator is Hilbert-Schmidt with locally bounded norm. By the symmetry property \eqref{eq:GK_sym} the same is true for the region $x \leq |k|^{-1} \leq y$.

Finally for $x, y \geq |k|^{-1}$ we have $|G^{\bowtie}_{p}(k;x,y)| \leq c_{p} \e^{-|\IM(k)| (x_> - x_<)}  \frac{x_>^{\IM(\lambda)}}{x_<^{\IM(\lambda)}}$. Hence
\begin{equation}
    \int_{|k|^{-1}}^{\infty} |G^{\bowtie}_{p}  (k;x,y)| \D y   \leq c_{p} \left( \int_{|k|^{-1}}^x \e^{ |\IM(k)| (y-x)} \frac{y^{- \IM(\lambda)}}{x^{- \IM(\lambda)}} \D y  + \int_x^{\infty} \e^{-|\IM(k)| (y-x)} \frac{y^{\IM(\lambda)}}{x^{\IM(\lambda)}}  \D y  \right)   . 
   \label{eq:Schur_int}
\end{equation}
If $\IM(\lambda) \leq 0$, then $ \frac{y^{\mp \IM(\lambda)}}{x^{\mp \IM(\lambda)}}  \leq 1$ under these integrals, so elementary calculation gives
\begin{equation}
    \int_{|k|^{-1}}^{\infty} |G^{\bowtie}_{p}  (k;x,y)| \D y  \leq \frac{2 c_p}{|\IM(k)|}.
\end{equation}
Next we consider the case $\IM(\lambda) >0$. Integration by parts in the first term of \eqref{eq:Schur_int} gives
\begin{equation}
\int_{|k|^{-1}}^x \e^{ |\IM(k)| (y-x)}  \frac{y^{- \IM(\lambda)}}{x^{- \IM(\lambda)}}  \D y   \leq \frac{1}{|\IM(k)|} +  \frac{\IM(\lambda)}{|\IM(k)|} \int_{|k|^{-1}}^x \e^{|\IM(k)|(y-x)} \frac{x^{\IM(\lambda)}}{y^{\IM(\lambda)+1}} \D y  .
\end{equation}
The integrand of this integral is maximized at one of the two endpoints, so
\begin{subequations} \label{eq:random_label}
\begin{align}
     \int_{|k|^{-1}}^x \e^{|\IM(k)|(y-x)} \frac{x^{\IM(\lambda)}}{y^{\IM(\lambda)+1}} \D y  & \leq \max \left \{ \e^{\frac{|\IM(k)|}{|k|}(1-|k|x)   } \frac{(|k|x)^{\IM(\lambda)+1}}{x} , \frac{1}{x} \right \} \int_{|k|^{-1}}^x \D y  \tag{\ref{eq:random_label}} \\
    & \leq \max \left \{ \e^{1 - |\IM(k)|x   } (|k|x)^{\IM(\lambda)+1}  , 1 \right \}. \nonumber
\end{align}
\end{subequations}
Optimizing with respect to $x$ we conclude that
\begin{equation}
\int_{|k|^{-1}}^x \e^{|\IM(k)|(y-x)} \frac{x^{\IM(\lambda)}}{y^{\IM(\lambda)+1}} \D y     \leq \max \left \{ \e \left( \frac{\IM(\lambda)+1}{|\IM(k)|} \right)^{\IM(\lambda)+1}  , 1 \right \} .
\end{equation}
In the second integral in \eqref{eq:Schur_int}, we integrate by parts $n \geq \IM(\lambda)$ times:
\begin{equation}
\int_x^{\infty} \e^{-|\IM(k)| (y-x)} \frac{y^{\IM(\lambda)}}{x^{\IM(\lambda)}}  \D y  =\sum_{j=0}^{n-1} \frac{c_j x^{-j}}{| \IM(k)|^{j+1}} + \frac{c_n}{|\IM(k)|^{n}} \int_{y}^{\infty} \e^{-|\IM(k)| (y-x)} \frac{y^{\IM(\lambda)-n}}{x^{\IM(\lambda)}} \D y ,
\end{equation}
where $c_j:=\IM(\lambda)
(\IM(\lambda)-1)\cdots (\IM(\lambda)-j+1)$.
Next we estimate $y^{\IM(\lambda)-n} \leq x^{\IM(\lambda)-n}$ and $x^{-j} \leq |k|^j$ under the remaining integral. Then simple calculation gives
\begin{equation}
\int_x^{\infty} \e^{-|\IM(k)| (y-x)} \frac{y^{\IM(\lambda)}}{x^{\IM(\lambda)}}  \D y  \leq \sum_{j=0}^{n} \frac{\IM(\lambda)_j |k|^j}{| \IM(k)|^{j+1}}.
\end{equation}
The same estimates are true for $\int_{|k|^{-1}}^{\infty} |G^{\bowtie}_{p}(k;x,y)|\D x$. The claim follows by Schur's criterion. 
This proves the boundedness of $G^{\bowtie}_p(k)$.

Equation \eqref{eq:GK_sym} implies that (whenever $G^{\bowtie}_{p}(k)$ is
defined) we have $\langle f|G^{\bowtie}_{p}(k)g\rangle = \langle G^{\bowtie}_{p}(k) f|g\rangle$ for $f,g \in C_c^{\infty}(\R_+,\C^2)$. By continuity, the same is true for all $f,g \in L^2(\R_+,\C^2)$. Thus $G^{\bowtie}_{p}(k)$ is self-transposed.

Next we check that $\langle f|D_p g\rangle = \langle D_p f|g\rangle$ for $f,g \in \Dom(D_p)$. To this end, we evaluate
\begin{equation}
\langle f|D_p g\rangle - \langle D_p f |g\rangle = \i \int_0^{\infty} \frac{\D}{\D x} \left( f(x)^\T \sigma_2 g(x) \right) \D x.
\end{equation}
If either $f$ or $g$ is in $C_c^{\infty}(\R_+,\C^2)$, the right hand side is zero. By continuity with respect to the graph norm, the
same is true for all $f,g \in \Dom(D_{\omega, \lambda}^{\min})$. Since
$\sigma_2$ is a skew-symmetric matrix, the right hand side vanishes
also for $f,g$ proportional to $\chi x^{\mu} \begin{bmatrix} a \\
  b \end{bmatrix}$. Thus $D_p$ is self-transposed.

Let $f \in L^2(\R_+,\C^2)$. We pick a sequence $f_i \in C_c^{\infty}(\R_+,\C^2)$ such that $f_i \to f$. Then $G^{\bowtie}_{p}(k) f_i \to G^{\bowtie}_{p}(k) f$ and 
\begin{equation}
D_{p}G^{\bowtie}_{p}(k) f_i =f_i + k G^{\bowtie}_{p}(k) f_i \to f + k G^{\bowtie}_{p}(k) f.
\end{equation}
Since $D_{p}$ is closed, this implies that $f \in \Dom(D_{p})$ and $D_{p}G^{\bowtie}_{p}(k) f = f + k G^{\bowtie}_{p}(k) f$. Therefore we have $\Ran(G^{\bowtie}_{p}(k)) \subset \Dom(D_p)$ and $(D_{p}-k)G^{\bowtie}_{p}(k)=1$.

For any $f \in L^2(\R_+,\C^2)$ and $g \in \Dom(D_{p})$ we have
\begin{equation}
\langle f| G^{\bowtie}_{p}(k) (D_p -k)g\rangle = \langle
G^{\bowtie}_{p}(k)
f|(D_p-k)g\rangle=\langle(D_p-k)G^{\bowtie}_{p}(k)f|g\rangle=\langle f|g\rangle.
\end{equation}
Since $f$ was arbitrary, $G^{\bowtie}_{p}(k) (D_{p}-k) g = g$. Thus $k \notin \sigma(D_p)$ and $G^{\bowtie}_{p}(k)=(D_{p}-k)^{-1}$. 

To show that $(D_p-k)^{-1}$ is unbounded for $k \in \mathbb R^\times $, we fix $\epsilon>0$ and consider the function \[ f_{\epsilon}(x) = \e^{- \epsilon x} \xi_{p}( k , x). \] Then $f_{\epsilon} \in \Dom(D_{p})$ and $|(D_{p}-k) f_{\epsilon} (x)|= \epsilon |f_{\epsilon} (x)|$, so $\frac{\| (D_{p}-k) f_{\epsilon} \|}{\| f_{\epsilon} \|} = \epsilon$. Hence $k \in \sigma(D_p)$. Since $\sigma(D_p)$ is closed, $\R \subset \sigma(D_p)$.

Finally, let $p \in \cE^\pm$, $k \in \C_\pm$. Then $\zeta_p^\pm(k, \cdot )$ belongs to $\Ker(D_p - k)$.
\end{proof}

\begin{corollary}
We have $D_p^* = D_{\overline p}$. In particular $D_p$ is self-adjoint if $p = \overline p$.
\end{corollary}

We are now ready to prove Theorem \ref{domain_comparison}.

\begin{proof}[Proof of Theorem \ref{domain_comparison}]
We choose some $k$ in the resolvent set of $D_p$.

If $|\RE(\mu)| \geq \frac{1}{2}$, then $D_{\omega, \lambda}^{\min}
=D_p$, so it suffices to show that $D_p = D_{\omega,
  \lambda}^{\max}$. Indeed, $D_p - k$ is surjective, so the ranges of $D_p-k$ and $D_{\omega,\lambda}^{\max}-k$ coincide. By Proposition \ref{eig_prop} also kernels are equal.

Next we consider the case  $|\RE(\mu)| < \frac{1}{2}$.

We easily check that $\chi x^{\mu} \begin{bmatrix} a \\
  b \end{bmatrix} \notin H_0^1(\R_+,\C^2)$, which by Proposition \ref{dom_prop} for  $|\RE(\mu)| < \frac{1}{2}$
coincides with
  $\Dom(D_{\omega,
    \lambda}^{\min})$. 
Hence,  $\Dom(D_{\omega, 
  \lambda}^{\min})$
  is a codimension one subspace of $\Dom(D_p)$.

Next, $D_{p}-k$ and $D_{\omega, \lambda}^{\max}-k$ have the same
 range--the whole Hilbert space. Besides, $\dim \Ker(D_{\omega, \lambda}^{\max}-k)=1$ by
 Proposition \ref{eig_prop}. Hence
 $\Dom(D_p)$ is a codimension one subspace of $\Dom(D_{\omega,
   \lambda}^{\max})$. 
\end{proof}

\begin{proposition}
Family $D_p$ has the following symmetries
\begin{subequations}
\begin{align}
\sigma_1 D_{\omega, \lambda, \mu, [a:b]} \sigma_1 &= - D_{\omega, - \lambda, \mu ,[b:a]}, \\
\sigma_2 D_{\omega, \lambda, \mu, [a:b]} \sigma_2 & = D_{- \omega, \lambda, \mu, [-b:a]}, \\
\sigma_3 D_{\omega, \lambda, \mu , [a:b]} \sigma_3 &= -D_{-\omega, - \lambda, [-a:b]},
\end{align}
\label{eq:D_sym}
\end{subequations}
where $\sigma_j$ are the Pauli matrices.
\end{proposition}
\begin{proof}
Matrix multiplication gives
\begin{equation}
\sigma_1 D_{\omega, \lambda} \sigma_1 = - D_{\omega, - \lambda}, \qquad \sigma_2 D_{\omega, \lambda} \sigma_2 = D_{- \omega, \lambda}, \qquad \sigma_3 D_{\omega, \lambda} \sigma_3 =-D_{- \omega, - \lambda}. 
\end{equation}
Using \eqref{eq:Dp_dom} one checks that the domains of operators on
the left and right hand side of \eqref{eq:D_sym} agree.
\end{proof}

\subsection{Essential spectrum}

\begin{proposition} \label{resolvents_compact_difference}
Let $p,p' \in \cM_{{}-\frac12}$ and $k\not\in\sigma(D_p)\cup\sigma(D_{p'})$. Then $G^{\bowtie}_{p}(k)-G^{\bowtie}_{p'}(k)$ is a Hilbert-Schmidt operator.
\end{proposition}
\begin{proof}
The proof of Theorem \ref{thoe} shows that it suffices to show that the integral operator with kernel $G^{\bowtie}_{p}(k) - G^{\bowtie}_{p'}(k)$ restricted to the region $x,y \geq |k|^{-1}$ is Hilbert-Schmidt. Furthermore, we~may assume that $\IM(k)<0$. Using formulas \eqref{eq:xi_large_x} and \eqref{eq:zeta_asymptotic} we obtain the following asymptotic expansion for $x, y  \to \infty$:
\begin{equation}
- \xi_{p}^-(k,x) \otimes \zeta_{p}^-(k,y) \sim \frac{1}{2 \i} \begin{bmatrix} 1 \\ \i  \end{bmatrix} \otimes \begin{bmatrix} 1 \\ - \i \end{bmatrix} \cdot \left( \frac{x}{y } \right)^{ \i \lambda} \e^{\i k(x-y )}. 
\end{equation}
It follows that we have
\begin{equation}
|G^{\bowtie}_{p}(k;x,y ) - G^{\bowtie}_{p'}(k;x,y )| \leq c \, \e^{- |\IM(k)| (x_> - x_<)} \left| \left( \frac{x_<}{x_>} \right)^{ \i \lambda} - \left( \frac{x_<}{x_>} \right)^{ \i \lambda'} \right|,
\end{equation}
with some constant $c$ independent of $x,y $. Therefore
\begin{subequations} \label{eq:I_def}
\begin{align}
I & := \int_{[|k|^{-1}, \infty]^2} |G^{\bowtie}_{p}(k;x,y ) - G^{\bowtie}_{p'}(k;x,y )|^2 \D x \D y  \tag{\ref{eq:I_def}} \\
& \leq 2c^2 \int_{|k|^{-1}}^{\infty}  \int_{y }^{\infty}  \e^{- 2 \, \IM(z) (x-y )} \left| \left( \frac{x}{y } \right)^{ -\i \lambda} - \left( \frac{x}{y } \right)^{- \i \lambda'} \right|^2 \D x \D y . \nonumber
\end{align}
\end{subequations}
Next we change variables to $y , t$ with $x = t y $. This gives
\begin{subequations} \label{eq:I_estimate}
\begin{align}
I & \leq  \int_1^{\infty} \int_{|k|^{-1}}^{\infty} y  \e^{- 2 | \IM(k)| (t-1) y } |t^{-\i \lambda} - t^{-\i \lambda'}|^2 \D y  \D t \tag{\ref{eq:I_estimate}} \\
& = \int_1^{\infty} \frac{|k|+2 |\IM(k)|(t-1)}{4 |\IM(k)|^2 |k| (t-1)^2} \e^{- 2 \frac{|\IM(k)|}{|k|}(t-1)} |t^{-\i \lambda} - t^{-\i \lambda'}|^2 \D t, \nonumber
\end{align}
\end{subequations}
where we have computed an elementary integral over $y $. The remaining integrand is bounded for $t \to 1$ and decays exponentially for $t \to \infty$. Therefore the integral converges.
\end{proof}

Resolvents of operators $D_p$ for distinct
  $p\in\cM_{-\frac12}$ are close to each other in the sense specified
  by Proposition \ref{resolvents_compact_difference}. Therefore, it is
  useful to know that for some $p$
  their integral kernels are particularly simple. These are provided in the Appendix \ref{app:elementary}.

By the essential spectrum (resp. essential spectrum of index zero) of a closed operator $R$ we mean the set $\sigma_{\ess}(R)$ (resp. $\sigma_{\ess,0}(R)$) of all $k \in \mathbb C$ such that $R-k$ is not a Fredholm operator (resp.~Fredholm operator of index zero). Clearly $\sigma_{\ess}(R) \subset \sigma_{\ess,0}(R)$.

\begin{lemma} \label{ess_spec_lemma}
Let $R,S$ be closed operators such that there exists $k_0$ in the intersection of resolvent sets of $R$ and $S$ such that $(R-k_0)^{-1}-(S-k_0)^{-1}$ is a compact operator. Then $\sigma_{\ess}(R) = \sigma_{\ess}(S)$ and $\sigma_{\ess,0}(R) = \sigma_{\ess,0}(S)$.
\end{lemma}
\begin{proof}
By assumption, $(S-k_0)^{-1}$ and $(R-k_0)^{-1}$ have the same essential spectra. The spectral mapping theorem proven in \cite{Buoni} gives
\begin{subequations} \label{eq:spectra_equality}
\begin{align}
\sigma_{\mathrm{ess}}(S) & = \{ k \in \mathbb C \, | \, \exists q \in \sigma_{\mathrm{ess}}((S-k_0)^{-1}) \  (k-k_0)q =1  \} \tag{\ref{eq:spectra_equality}} \\
& = \{ k \in \mathbb C \, | \, \exists q \in \sigma_{\mathrm{ess}}((R-k_0)^{-1}) \  (k-k_0)q =1  \} = \sigma_{\mathrm{ess}}(R). \nonumber
\end{align}
\end{subequations}
The same argument works also for $\sigma_{\ess,0}$.
\end{proof}

\begin{corollary} \label{ess_spec}
For any $p\in  \mathcal
  M_{{}-\frac12}$ we have
$\sigma_{\ess}(D_p)=\sigma_{\ess,0}(D_p) = \mathbb R$.
\end{corollary}
\begin{proof}
There exists $p$ such that $\sigma(D_p)=\R$. By Lemma
\ref{ess_spec_lemma}, it is sufficient to prove our statement for such $p$. Clearly, $\sigma_{\ess}(D_p)  \subset \sigma_{\ess,0}(D_p) \subset \sigma(D_p) = \mathbb R$. If $k \in \R$, then $D_p -k$ is injective and its range is dense, hence not closed, for otherwise $(D_p-k)^{-1}$ would be bounded.
\end{proof}

\begin{corollary} \label{ess_spec_min_max}
Let $\omega, \lambda$ be such that $|\RE \sqrt{\omega^2 - \lambda^2}| < \frac12$. Then $\sigma_{\ess}(D_{\omega,\lambda}^{\min}) = \sigma_{\ess}(D_{\omega,\lambda}^{\max}) = \R$. If $k \in \C \setminus \R$, then $D_{\omega,\lambda}^{\min} - k$ and $D_{\omega,\lambda}^{\max} - k$ are Fredholm operators with indices $-1$ and $+1$, respectively. If $D$ is an operator satisfying $D_{\omega,\lambda}^{\min} \subsetneq D \subsetneq D_{\omega,\lambda}^{\max}$, then $\sigma_{\ess}(D) = \sigma_{\ess,0}(D)=\R$.
\end{corollary}
\begin{proof}
Follows from Theorem \ref{domain_comparison} and Corollary \ref{ess_spec}.
\end{proof}

\subsection{Limiting absorption principle}

Let $s\in\mathbb{R}$. The
Hilbert space $L^2_s(\R_+,\C^2)$ is defined as the completion of
$C_c^{\infty}(\R_+,\C^2)$ with respect to the norm induced by the
scalar product $( f | g )_s = \int_0^{\infty} (1+x^2)^{s}
\overline{f(x)} g(x) \D x$. For any $t \in \mathbb R$ we have a
unitary operator $\langle X \rangle^t : L^2_s(\R_+,\C^2) \to
L^2_{s-t}(\R_+,\C^2)$ given by $(\langle X \rangle^t f)(x) =
(1+x^2)^{\frac{t}{2}} f(x)$, alternatively regarded as an (unbounded for $t>0$) positive operator on $L^2(\R_+,\C^2)$. 

\begin{proposition} \label{limiting_absorption}
Let $p \in \cM_{- \frac12} \backslash \cE^{\pm}$, $k \in \R^\times $. The limit $G^{\bowtie}_{p}(k \pm \i 0) := \lim \limits_{\epsilon \downarrow 0} G^{\bowtie}_{p}( k \pm \i \epsilon)$ exists as a~compact operator $L^2_s(\R_+,\C^2) \to L^2_{-s}(\R_+,\C^2)$ for any $s > |\mathrm{Im}(\lambda)| + \frac{1}{2}$ and depends continuously on $p,k$.

If $\RE(\mu) >0$, then $\R^\times $ may be replaced by $\R$ in the
above statement and $G_p^{\bowtie}(\pm \i 0)$ has the kernel
\begin{equation}
G_p^{\bowtie}(0;x,y) = \frac12 \one_{\R_+}(y-x) \left( \frac{x}{y} \right)^{\mu} \begin{bmatrix} z & -1 \\ 1 & - z^{-1} \end{bmatrix} + \frac12 \one_{\R_+}(x-y) \left( \frac{y}{x} \right)^{\mu} \begin{bmatrix} z & 1 \\ -1 & - z^{-1} \end{bmatrix}.
\end{equation}
If $\RE(\mu) \leq 0$, then $\| G_p^{\bowtie}(k) \|_{B(L^2_s,L^2_{-s})} = O(|k|^{2 \RE(\mu)})$ for $k \to 0$.

Therefore in both cases we have $\| G_p^{\bowtie}(\cdot \pm \i 0)  \|_{B(L^2_s, L^2_{-s})} \in L^1_{\mathrm{loc}}(\R)$.
\end{proposition}
\begin{proof}
It is sufficient to cover the case of $k$ approaching the real axis
from below. Asymptotics of $G^{\bowtie}(k;x,y)$ are such that
$(1+x^2)^{- \frac{s}{2}} (1+y ^2)^{- \frac{s}{2}}
G^{\bowtie}_{p}(k;x,y )$ is an $L^2(\mathbb R_+^2 ,
\mathrm{End}(\C^2))$ function. Dominated convergence theorem implies
that it depends continuously (in the~$L^2$ sense) on $p,k$, including the boundary set $\IM(k)=0$. Therefore $\langle X \rangle^{-s} G^{\bowtie}_{p}(k) \langle X \rangle^{-s}$ is a~continuous family of Hilbert-Schmidt (and hence compact) operators on $L^2(\R_+,\C^2)$, so $G^{\bowtie}_{p}(k)$ defines an operator $L^2_s(\R_+,\C^2) \to L^2_{-s}(\R_+,\C^2)$ which may be written as a composition of two unitaries and a~compact operator.

The second part follows from the asymptotics of $\xi^{\pm}_p$ and $\zeta^{\pm}_p$ functions for small arguments and the dominated convergence theorem.
\end{proof}

\subsection{Generalized eigenvectors}

Point spectrum of $D_p$, when present, possesses quite
counter-intuitive properties. Note that in this subsection an
important role is played by the bilinear product
$\langle\cdot|\cdot\rangle$.

\begin{proposition} \label{orthog}
Let $n,m\in\mathbb{N}$. If $f \in \Ker((D_p-k)^n)$, $g \in \Ker((D_p-k')^m)$, then $\langle f | g \rangle =0$.
\end{proposition}
\begin{proof}
Assume at first that $k' \neq k$. We induct on $m$. If $m=1$, then
\begin{equation}
0 = \langle (D_p - k)^n f | g \rangle = \sum_{j=0}^n \binom{n}{j} (k'-k)^{n-j} \langle f | (D_p - k')^j g \rangle = (k'-k)^n \langle f | g \rangle.
\end{equation}
Cancelling $(k'-k)^n$ we obtain the induction base. Assume that the
claim is true for $m$ and let $g \in \Ker((D_p-k')^{m+1})$. By a similar calculation
\begin{equation}
0= \sum_{j=0}^n \binom{n}{j} (k'-k)^{n-j} \langle f | (D_p - k')^j g \rangle = (k-k')^n \langle f | g \rangle,
\end{equation}
where the last equality follows from $(D_p-k')^j g \in
\Ker((D_p-k')^m)$ for $j \geq 1$ and the induction hypothesis. This
completes the proof for $k \neq k'$.

So far we used only the self-transposedness of $D_p$. Next
  we will also use
  its homogeneity.

Let $k'=k$. Then for any $\tau \in \R^{\times}$ we have $U_{\tau} g \in \Ker((D_p-k'')^m)$ for some $k'' \neq k$. Hence $\langle f | U_{\tau} g \rangle =0$. Now take $\tau \to 0$.
\end{proof}

\begin{proposition}
If $p \in \cE^\pm$ and $k \in \C_\pm$, then for every $n \in \mathbb N$ we have $\dim(\Ker((D_p-k)^n)) = n$.
\end{proposition}
\begin{proof}
We proceed by induction on $n$. Case $n=0$ is trivial and $n=1$ is
already established. By the inductive hypothesis, there exists $f \in
\Ker((D_p-k)^n) \backslash \Ker((D_p-k)^{n-1})$, unique up to elements
of $\Ker((D_p-k)^{n-1})$ and multiplication by nonzero scalars. Then
$f \in \Ker(D_p -k)^{\tperp}$ by Proposition \ref{orthog}. On the other hand
$\Ker(D_p-k)^{\tperp} = \left( \Ran(D_p-k)^{\tperp} \right)^{\tperp} =
\Ran(D_p - k)$. Here the last equality holds because $D_p -k$ has closed range, see Corollary \ref{ess_spec}. Thus there exists $g \in \Dom(D_p -k)$, unique up to elements of $\Ker(D_p - k )$, such that $(D_p-k) g =f$. Clearly, $g \in \Ker((D_p-k)^{n+1}) \backslash \Ker((D_p-k)^n)$ and we have a vector space decomposition $\Ker((D_p-k)^{n+1}) = \C g \oplus \Ker((D_p-k)^{n})$ .
\end{proof}

\begin{question}
Let $p \in \cE^\pm$, $k \in \C_{\pm}$. We denote the $L^2$ closure of $\bigcup \limits_{n=0}^{\infty} \Ker((D_p - k)^n)$ by $\mathcal N_p(k)$. By Lemma \ref{orthog} we have $\mathcal N_p(k) \subset \mathcal N_p(k)^\tperp$. In Appendix \ref{app:elementary} we have verified that in the case $\omega=0$ subspace $\mathcal N_p(k)$ does not depend on the choice of $k \in \C_\pm$ and $\mathcal N_p(k) = \mathcal N_p(k)^\tperp $ (equivalently, $\mathcal N_p(k) \oplus \overline{\mathcal N_p(k)} = L^2(\R_+,\C^2)$). We leave open the question whether these assertions remain true for $\omega \neq 0$.
\end{question}

\section{Diagonalization}
\label{Diagonalization}

Let $k \in \R^\times $. Recall that $\varepsilon_k = \sgn(\RE(k))$.
On
  the real line,
 it is convenient  to rewrite the formulas for $\xi^\pm$ and $\zeta^\pm$
(\ref{xizeta-}, \ref{xizeta+})
 in terms of trigonometric Whittaker functions
  (\ref{Jbm-definition}, \ref{Hpm-definition}):
  \begin{subequations}
\begin{gather}\label{gather1}
\xi_p^{\pm}(k,x) = \frac{\i^{\mp \varepsilon_k \mu}}{2 N_p^{\pm} \mu} \left( \varepsilon_k \omega \cJ_{\varepsilon_k \lambda, \mu + \frac{1}{2}}(2|k|x) \begin{bmatrix} -z \\ 1 \end{bmatrix} + \cJ_{\varepsilon_k \lambda, \mu - \frac{1}{2}} (2|k|x) \begin{bmatrix} z \\ 1 \end{bmatrix}  \right), \\
\zeta_p^{\pm}(k,x) = \frac{\i^{\pm \varepsilon_k \mu}}{2} \left( \pm \i \varepsilon_k  (z \pm \i) \cH^{\pm \varepsilon_k}_{\varepsilon_k \lambda, \mu + \frac12} (2 |k| x) \begin{bmatrix} -1 \\ z^{-1} \end{bmatrix} + (z \mp \i ) \cH^{\pm \varepsilon_k}_{\varepsilon_k \lambda, \mu - \frac12} (2 |k| x) \begin{bmatrix} 1 \\ z^{-1} \end{bmatrix}  \right).
\end{gather}
\end{subequations}
For $\mu$ near $0$ it is convenient instead of 
\eqref{gather1} to use a version of \eqref{eq:xi_reg}:
\begin{small}
\begin{gather}\label{gather2}
\xi_p^{\pm}(k,x) = \frac{\i^{\mp \varepsilon_k \mu}}{2 N_p^{\pm}}
\left( \varepsilon_k \cJ_{\varepsilon_k \lambda, \mu +
    \frac{1}{2}}(2|k|x) \begin{bmatrix} 1 \\ -z \end{bmatrix} +
  \frac{\cJ_{\varepsilon_k \lambda, \mu - \frac{1}{2}} (2|k|x)
    +\varepsilon_k\lambda \cJ_{\varepsilon_k \lambda, \mu + \frac{1}{2}} (2|k|x)}{\mu}
    \begin{bmatrix} z \\ 1 \end{bmatrix}  \right).
   \end{gather}
   \end{small}
   
The leading terms of $\xi_p^{\pm}$ and $\zeta_p^\pm$ for large $kx$ are
\begin{subequations}
\begin{align}
\xi_p^{\pm}(k,x)& \sim \frac{\e^{- \frac{\varepsilon_k \pi \lambda}{2}}}{2} \left( \e^{\mp \i k x} (2 |k|x )^{\mp \i \lambda} \begin{bmatrix} 1 \\ \mp \i \end{bmatrix} 
+ (S_p \e^{- \i \varepsilon_k \pi \mu})^{\pm 1} \e^{\pm \i k x} (2 |k| x)^{\pm \i \lambda} 
\begin{bmatrix} 1 \\ \pm \i \end{bmatrix} 
  \right), \label{eq:xi_real_asymp} \\
  \label{eq:zeta_real_asymp} 
  \zeta_p^{\pm}(k,x)& \sim\mp\i \, \e^{\frac{\varepsilon_k \pi \lambda}{2}} \e^{\pm \i k x} (2 |k|x )^{\pm \i \lambda} \begin{bmatrix} 1 \\ \pm \i \end{bmatrix}. 
\end{align}
\end{subequations}

Because of the long-range nature of the perturbation and of the
presence of spin degrees of freedom, it is not obvious what should be chosen as the definition of the outgoing and incoming waves. Let us call
$\i\zeta^+(k,x)$ the \textit{outgoing
  wave}
and $-\i\zeta^-(k,x)$ the
\textit{incoming wave}.
Then the ratio of the outgoing wave and the incoming wave in $\xi^+(k,x)$
is
$\e^{- \i \varepsilon_k \mu} S_p$ 
and can be called the {\em (full) scattering amplitude at energy $k$.}

\begin{proposition}
  Let $p \in \cM_{- \frac{1}{2}} \backslash (\cE^+ \cup \cE^-)$,  
$k \in \R^\times $, $s > |\IM(\lambda)| + \frac{1}{2}$. Then the {\em spectral density}
\begin{equation}
\Pi_p(k):=(2 \pi \i)^{-1} \left( G^{\bowtie}_{p}(k + \i 0) - G^{\bowtie}_{p}(k - \i 0)
    \right)
    \end{equation}
  is
well defined as a compact operator $ L^2_s(\R_+,\C^2) \to L^2_{-s}(\R_+,\C^2)$ and has the integral kernel
\begin{equation}
    \Pi_{p}(k;x,y ) = \frac{\e^{\varepsilon_k \pi \lambda}}{ \pi} \xi_{p}^{+}(k,x) \xi_{p}^{-}(k,y)^\T = \frac{\e^{\varepsilon_k \pi \lambda}}{ \pi} \xi_{p}^{-}(k,x) \xi_{p}^{+}(k,y)^\T.
    \label{eq:Pi_factorization}
\end{equation}
As $k \to 0$, it admits the expansion
\begin{equation}
\Pi_p(k) = \e^{\varepsilon_k \pi \lambda} |k|^{2 \mu} \Pi_p^0 + O(|k|^{2 \mu +1}),
\end{equation}
where the remainder is estimated in the $B(L^2_s,L^2_{-s})$ norm and
$\Pi_p^0$ has the integral kernel
\begin{equation}
\Pi_p^0(x,y) = \frac{ (4xy)^{ \mu}}{\pi \, \Gamma(2 \mu +1 )^2 N_p^+ N_p^-}  \begin{bmatrix} z^2 & z \\ z & 1 \end{bmatrix}.
\end{equation}
\end{proposition}
\begin{proof}
The first statement follows from Proposition \ref{limiting_absorption}. By \eqref{eq:GK_sym}, it is sufficient to prove \eqref{eq:Pi_factorization} for $x < y $. Plugging \eqref{eq:xipm} into \eqref{eq:Gdef} we find
\begin{equation}
G_p^{\bowtie}(k + \i 0;x,y) - G_p^{\bowtie}(k- \i 0;x,y) = \xi_p^-(k,x) \left( \zeta_p^-(k,y) - \e^{-\i \varepsilon_k \pi \mu} S_p \zeta_p^+(k,y) \right)^\T.
\end{equation}
Plugging in \eqref{eq:xipzz} we obtain \eqref{eq:Pi_factorization}.

The last part of the statement follows from asymptotics of $\xi$ functions for small arguments and the dominated convergence theorem.
\end{proof}

We refer to Appendix \ref{app:operators_Rplus} for definitions used in
the lemma below. Note also the identity
$\xi_p^\pm(k,x)=\xi_p^\pm(\varepsilon_k,|k|x)$, which allows us to
restrict our attention to
$\xi_p^\pm(\varepsilon_k,x)$. The following fact follows immediately from Lemma \ref{melino} and \eqref{gather2}.

\begin{lemma} \label{xi_Mel}
$ \xi_p^{\pm}(\varepsilon_k, x)$, $p \notin \cE^\pm$, is a tempered
distribution in $x \in \R_+$, in the sense explained
  in Appendix \ref{app:operators_Rplus}. Its
 Mellin transform is
\begin{align}
\Xi_p^\pm (\varepsilon_k, s):= & \int_0^{\infty} \e^{ - 0 x} x^{- \frac{1}{2} - \i s} \xi_p^{\pm}(\varepsilon_k,x) \D x \label{eq:xi_Mellin} \\
=  & \i^{\mp \varepsilon_k \mu - \frac32 - \mu + \i s} 2^{\mu -1} \Gamma \left( \frac12 + \mu - \i s \right) \frac{1}{N_p^{\pm}\mu} \nonumber \\
\times &  \Bigg(  2 \varepsilon_k \omega\Big(\frac12 + \mu - \i s\Big) \, {}_{2}\mathbf{F}_1 \left( 1 + \mu + \i \varepsilon_k \lambda , { \scriptstyle \frac32} + \mu - \i s ; 2 \mu + 2 ; 2 + \i 0 \right) \begin{bmatrix} - z \\ 1 \end{bmatrix}  \nonumber \\
& + \i\, {}_2 \mathbf{F}_1 \left(  \mu + \i \varepsilon_k \lambda , {\scriptstyle \frac12} + \mu - \i s ; 2 \mu ; 2 + \i 0 \right) \begin{bmatrix} z \\ 1 \end{bmatrix} \Bigg), \nonumber
\end{align}
 is analytic in $s$ and bounded by $c_p^{\pm} (1+s^2)^{\frac12 |\IM(\lambda)|}$ locally uniformly in $p$.
\end{lemma}

We define $\cU_p^{\pm, \pre}$, $p \in \cM \backslash \cE^{\pm}$, as
the integral operator $C_\mathrm{c}^{\infty}(\R_+, \C^2) \to C^{\infty}(\R) $
with the  kernel
\begin{equation}
\cU^\pm_p(k,x) = \frac{\e^{\frac12 \varepsilon_k \pi \lambda}}{ \sqrt{\pi}} \xi_p^{\pm}(k,x)^\T.
\end{equation}
By construction, the kernel of the spectral density operator factors as
\begin{equation}
\Pi_p(k ; x , y) = \cU_p^+(k,x)^{\T} \cU_p^-(k,y) = \cU_p^-(k,x)^{\T} \cU_p^+(k,y).
\label{eq:Pi_Factorization_2}
\end{equation}
We note also the relations
\begin{equation}
\cU_p^+(k,x) = \e^{- \i \varepsilon_k \pi \mu } S_{p} \, \cU_p^-(k,x) , \qquad \overline{\cU_{\overline p}^{\pm}(k,x)} = \cU_{p}^{\mp}(k,x) \label{eq:V_clutch}
\end{equation}
 and the intertwining property
\begin{equation}
(\cU_p^{\pm, \pre} D_p f)(k) = k (\cU_p^{\pm, \pre} f)(k), \qquad  f \in C_\mathrm{c}^\infty(\R_+, \C^2).
\end{equation}

Recall from Subsection
  \ref{Remarks about notation} that $J$ is the inversion and $A$ is the generator of 
dilations, and $K$ is the multiplication operator on $L^2(\mathbb{R})$
by the variable $k\in\mathbb{R}$.

Below we will consider level sets $\{ \lambda = \lambda_0 \} \subset
\cM_\frac12$. Recall from the discussion around equation
\eqref{eq:lambda_zero_decomp} that it is a submanifold for
$\lambda_0 \neq 0$, but for $\lambda_0 =0$ it is the union of three submanifolds singular along the intersection. We will say that a function on the locus $\{ \lambda =0 \}$ is holomorphic if its restriction to each of the three components is holomorphic.


\begin{proposition}
$\cU_p^{\pm, \pre}$ are densely defined closable operators $L^2(\R_+,
\C^2) \to L^2(\R)$ with the closures given by
\begin{equation}
\cU_p^{\pm} f(k) = \frac{\e^{\frac12 \varepsilon_k \pi \lambda}}{\sqrt{\pi}}  \Xi_p^{\pm \T}(\varepsilon_k, A) J f (|k|), \qquad k \in \R.
\label{eq:Ucl}
\end{equation}
Hence $\cU^\pm_p (1+A^2)^{- \frac12 |\IM(\lambda)|}$ is bounded. In
particular $\cU^\pm_p$ are bounded if $\lambda \in
\R$. If $\lambda_0 \in \R$, they form a holomorphic
  family of operators on the level set
  $\{ \lambda = \lambda_0 \} \backslash \cE^{\pm}$.
  Furthermore, $\cU_p^{\pm *} = \cU_{\overline p}^{\mp \T}$.
\end{proposition}
\begin{proof}
The first part follows from Lemma \ref{xi_Mel} and discussion in Appendix \ref{app:operators_Rplus}. Now fix $\lambda_0 \in \R$ and consider $p$ in a component $S$ of the level set $\{ \lambda = \lambda_0 \} \backslash \cE^{\pm}$. If $f \in C_c^{\infty}(\R)$, $g \in C_c^{\infty}(\R_+, \C^2)$, then $( f | \cU_{p}^{\pm} g )$ is a holomorphic function of $p \in S$. Since $C_c^{\infty}$ spaces are dense in $L^2$ and $\cU_p^{\pm}$ are bounded locally uniformly in $p$, $\cU_p^{\pm}$ is a holomorphic operator-valued function. The last claim follows from the formula \eqref{eq:V_clutch}.
\end{proof}

In a sense, operators $\cU^\pm_p$ diagonalize $D_p$ for $p \in \cM_{-
  \frac12} \backslash \cE^\pm$. If $p = \overline p$, then $D_p$ are
self-adjoint and $\cU_p^\pm$ are unitary. If we assume only that
$\lambda$ is real, then  $\cU_p^\pm$ are still bounded with bounded inverses, so they are almost as good as in the self-adjoint case. This will be made precise below.

\begin{proposition} \label{spectral_measure}
If $p = \overline p$, then for any bounded interval $[a,b] \subset \mathbb R$ and $f,g \in C_c^{\infty}(\R_+,\C^2)$ 
\begin{equation}
( g|\one_{[a,b]}(D_p) f ) = \int_a^b \int_{0}^{\infty} \int_{0}^{\infty} g(x)^* \Pi_{p}(k;x,y )  f(y ) \D y  \D x \D k.
\label{eq:spectral_projectors}
\end{equation}
Besides, $\cU_{p}^\pm$ is a unitary operator and
\begin{equation}
D_{p} = \cU_{p}^{\pm *} K \cU_{p}^{\pm}. \label{pasd}
\end{equation}
\end{proposition}
\begin{proof}
Since the point spectrum of $D_{p}$ is trivial for $p=\overline p$, Stone's formula gives
\begin{equation}
( g|\one_{[a,b]}(D_p) f ) = \lim_{\epsilon \downarrow 0} \frac{1}{2 \pi \i} \int_{[a,b] \times \R_+^2} g(x)^* (G_{p}^{\bowtie}(k + \i \epsilon;x,y ) - G_{p}^{\bowtie}(k - \i \epsilon;x,y )) f(y ) \D y  \D x \D k.
\end{equation}
It follows from the asymptotics of functions $\xi_p^{\pm}$ and
$\zeta_p$ that on $[a,b] \times \supp(f) \times \supp(g)$ we have
$|G_{p}^{\bowtie}( k \pm \i \epsilon;x,y )| \leq c |k|^{
  \RE(\mu)-|\RE(\mu)|}$ with $c$ independent of $k$. This function is
integrable, because $ \RE(\mu)-|\RE(\mu)| > -1$. Therefore by the
dominated convergence theorem, the limit $\epsilon \downarrow 0$ may
be taken under the integral.
This proves \eqref{eq:spectral_projectors}.

Let us prove the unitarity of $\cU_p^\pm$.
Let $f \in C_c^{\infty}(\R_+,\C^2)$ and let $[a,b]$ be a bounded interval. Then
\begin{align}
 \int_a^b |\cU_{p}^{\pm} f(k)|^2 \D k &= \int_0^{\infty} \int_a^b \int_0^{\infty} f(x)^*\Pi_{p}(k;x,y ) f(y ) \D y  \D k \D x \label{eq:parseval}  \\
&= \int_0^{\infty} f(x)^* (\one_{[a,b]}(D_p)f)(x) \D x = ( f| \one_{[a,b]}(D_p) f ), \nonumber
\end{align}
where in the first step we used the definition of $\cU_{p}^\pm$,
conjugation formula \eqref{eq:V_clutch} and the
factorization \eqref{eq:Pi_Factorization_2}. The order of integrals is immaterial, because the integrand is compactly supported and its only possible singularity (at $k=0$, if $0 \in [a,b]$) is integrable. In the second step we used Proposition~\ref{spectral_measure}. Taking the limit $b \to \infty$, $a \to - \infty$ we find
\begin{equation}
\int_{- \infty}^{\infty} |\cU_{p}^{\pm} f(k)|^2 \D k = ( f| f ).
\end{equation}
Hence $\cU_{p}^{\pm}$ is an isometry. Equation \eqref{eq:parseval} implies that
\begin{equation}
\one_{[a,b]}(D_p) = \cU_p^{\pm *} \one_{[a,b]}(K) \cU_p^{\pm}.
\end{equation}

It remains to show that $\cU_p^{\pm}\cU_p^{\pm*}=1$. The proof of
this fact follows closely the proof of (3.37) of Theorem 3.16 in
\cite{DeRi18_01}. 
\end{proof}

\begin{proposition}
If $p \in \cM_{- \frac12} \backslash (\cE^+ \cup \cE^-)$ is such that
$\lambda \in \R$, then $(\cU_p^\pm)^{-1}=\cU^{\mp\T}$ and
\begin{subequations}
\begin{gather}
D_p = \cU_p^{\pm -1} K \cU_p^{\pm }. \label{eq:holo_diag}
\end{gather}
\end{subequations}
In particular $D_p$ is similar to a self-adjoint operator.
\end{proposition}
\begin{proof}
We fix $\lambda_0 \in \R$. Then $\cU_p^{\pm} \cU_p^{\mp \T}-1$ and $\cU_p^{\pm \T} \cU_p^{\mp}-1$ form
holomorphic families of bounded operators on (one-dimensional) $\{
\lambda = \lambda_0 \} \backslash (\cE^+ \cup \cE^-)$. They vanish on the set of real points, which has an accumulation point in each component of the domain. Thus they vanish everywhere. 

Now take $k \in \C \backslash \R$. Arguing as in the previous paragraph we obtain
\begin{equation}
\cU_p^{\pm -1} (K-k)^{-1} \cU_p^\pm = (D_p - k)^{-1},
\end{equation}
from which \eqref{eq:holo_diag} follows immediately.
\end{proof}

\begin{question} If $\lambda\in\R$, then $D_p$ is similar to a
  self-adjoint operator. Hence it enjoys a very good functional
  calculus--for any bounded Borel function $f$ the operator $f(D_p)$ is well defined
  and bounded.

  If $\lambda \notin \R$ this is probably no longer true, because the
diagonalizing operators $\cU_p^\pm$ are unbounded. However, they are
unbounded  in a controlled manner: they are continuous on the domain
of some power of the dilation operator. One may hope that this is
sufficient to allow for a rich functional calculus for Dirac-Coulomb
Hamiltonians. We pose an open problem: for a~given  $\IM(\lambda)$,
characterize
functions that allow for a functional calculus for $D_p$.
In particular, one could ask  when $\i D_p$ generates a $C^0$ semigroup of bounded operators.
\end{question}

\section{Numerical range and dissipative properties}
\label{Numerical range and dissipative properties}
In this section we give a complete analysis
of the numerical range of various
realizations of 1d Dirac-Coulomb Hamiltonians studied in this paper.

\begin{proposition} \label{num_range_minimal}
One of the following mutually exclusive statements is true:
\begin{enumerate}
    \item $\omega$ and $\lambda$ are real. Then $\Num(D_{\omega, \lambda}^\pre) = \R$.
    \item $|\IM (\omega)| < |\IM(\lambda)|$. Then $\Num(D_{\omega, \lambda}^\pre) = \C_{- \sgn(\IM(\lambda))}$.
    \item $|\IM (\omega)| = |\IM (\lambda)| \neq 0$. Then $\Num(D_{\omega, \lambda}^\pre) = \C_{- \sgn(\IM(\lambda))} \cup \{ 0 \}$.
    \item $|\IM (\omega)| > |\IM (\lambda)|$. Then $\Num(D_{\omega, \lambda}^\pre) = \C$.
\end{enumerate}
The same is true with $D_{\omega, \lambda}^\pre$ replaced by $D_{\omega, \lambda}^{\min}$ throughout.
\end{proposition}
\begin{proof}
Integrating by parts we find that for $f = \begin{bmatrix} f_1 \\ f_2 \end{bmatrix} \in C_\mathrm{c}^{\infty}(\R_+,\C^2)$ we have
\begin{equation}
\IM ( f| D_{\omega, \lambda} f ) = - \IM(\lambda + \omega) \int_0^{\infty} \frac{|f_1(x)|^2}{x} \D x - \IM (\lambda - \omega) \int_0^{\infty} \frac{|f_2(x)|^2}{x} \D x .
\end{equation}
In the four cases listed in the proposition we have: both terms are zero in Case $1.$, both terms are nonzero (except for $f=0$) and have the same sign as $- \IM(\lambda)$ in Case $2.$, one term is zero and the other has the same sign as $- \IM(\lambda)$ in Case $3.$ and the two terms have opposite signs in the last case. Therefore inclusions of numerical ranges in the specified sets are clear, except for the third case. Then in order for $\IM ( f| D_{\omega, \lambda} f )$ to vanish, one of the two $f_j$ has to be zero. It is easy to check that this implies $( f| D_{\omega, \lambda} f ) =0$ (but not $f =0$).

We have to show that the obtained inclusions are saturated. The
homogeneity of $D_{\omega, \lambda}^\pre$ implies that
$\Num(D_{\omega, \lambda}^\pre)$ is a convex cone. Thus to establish
the result in Case $1.$ it is sufficient to show that both signs of $(
f | D_{\omega, \lambda} f )$ are possible. We choose a nonzero $\varphi \in C_c^{\infty}(\R_+,\C^2)$ with $\| \varphi \|_{H_0^1}=1$ and put $f_{\pm, t}(x) = \begin{bmatrix} \varphi(x-t) \\ \pm \varphi'(x-t) \end{bmatrix}$ for $t \geq 0$. Then $\| f_{\pm, t} \|_{L^2}=1$ and 
\begin{equation}
    ( f_{\pm, t} | D_{\omega,\lambda} f_{\pm, t} ) = \pm 2 \int_0^{\infty} |\varphi'(x)|^2 \D x - \int_0^{\infty} \frac{1}{x+t} \left( (\lambda+\omega) |\varphi(x)|^2 + (\lambda - \omega) |\varphi'(x)|^2 \right)  \D x .
    \label{eq:hform}
\end{equation}
The first term is nonzero, has sign $\pm$ and does not depend on $t$, while the other converges to zero for $t \to \infty$. Therefore $\pm ( f_{\pm, t} |D_{\omega, \lambda} f_{\pm, t} ) \geq c_{\pm} >0$ for large enough $t$.

Next we suppose that $|\IM(\omega)| \leq |\IM (\lambda)| \neq 0$. It
is sufficient to show that $\C_-$ is included in the numerical range
for $\IM (\lambda)<0$. Arguing as below \eqref{eq:hform}, we~deduce
that there exist constants $t_0 > 0$ and $c_{\pm} >0$ such that $\pm
\RE \, ( f_{\pm, t} | D_{\omega,\lambda} f_{\pm, t} ) \geq c_{\pm}$
for $t \geq t_0$. Let $\delta = \IM \, (f_{ \pm, t_0}
|D_{\omega,\lambda} f_{ \pm, t_0} )$. Then $\delta >0 $. The function $t \mapsto \IM \, ( f_{\pm, t} | D_{\omega,\lambda} f_{\pm, t} )$ is continuous and converges to zero for $t \to \infty$, so~for every $\epsilon \in ]0, \delta]$ there exists $t \geq t_0$ such that $\IM \, ( f_{\pm, t} | D_{\omega,\lambda} f_{\pm, t} ) = \epsilon$. By convexity of numerical ranges this implies $[-c_- + \i \epsilon, c_+ + \i \epsilon]  \subset \Num (D_{\omega,\lambda})$. Homogeneity implies that for every $s>0$ we have $\left[ -\frac{c_- s}{\epsilon} + \i s, \frac{c_+ s}{\epsilon} + \i s \right] \subset \Num (D_{\omega, \lambda})$. Every $k$ with $\IM(k)=s$ is in this interval for small enough $\epsilon$. 

Similar argument shows that in Case $4.$ there exist $c_{\pm} >0$ and $\delta>0$ such that for every $\epsilon \in ] 0 , \delta]$ there exist $g_{\pm,\epsilon} \in C_c^{\infty}(\R_+,\C^2)$ with $\| g_{\pm,\epsilon} \|_{L^2}=1$, $\pm \RE ( g_{\pm,\epsilon}|D_{\omega, \lambda} g_{\pm,\epsilon} ) \geq c_{\pm}$ and $\left| \IM ( g_{\pm,\epsilon}| D_{\omega, \lambda} g_{\pm,\epsilon} )\right| \leq \epsilon$. On the other hand for nonzero $f \in C_c^{\infty}(\R_+,\C^2)$ with $f_1=0$ or $f_2=0$ we have that $( f| D_{\omega, \lambda} f )$ is proportional to $\omega-\lambda$ or $- \omega- \lambda$, respectively, with a positive proportionality constant. Using homogeneity we can even construct functions $f$ with the proportionality constant equal to $1$ and $\| f \|=1$. Next we observe that if $\epsilon$ is taken to be sufficiently small, the convex hull of $( g_{+, \epsilon} | D_{\omega, \lambda} g_{+, \epsilon} )$, $( g_{-, \epsilon} | D_{\omega, \lambda} g_{-, \epsilon} )$, $\omega - \lambda$ and $- \omega - \lambda$ contains zero in its interior. Therefore the smallest convex cone containing it is the whole $\C$.

To prove the last statement, first note that $\Num(D_{\omega, \lambda}^{\min})$ is contained in the closure of $\Num(D_{\omega, \lambda}^\pre)$. Therefore in Cases $1.$ and $4.$ there is nothing to prove. We consider Case $2.$ We have to show that if $g \in \Dom(D_{\omega, \lambda}^{\min})$ is such that $\IM ( g| D_{\omega, \lambda} g ) =0$, then $g=0$. We choose $\epsilon>0$ and $f \in C_c^{\infty}(\R_+,\C^2)$ such that $\| f - g \|_{\Dom(D_{\omega, \lambda}^{\min})} < \epsilon$. Then
\begin{equation}
\IM ( f| D_{\omega, \lambda} f ) = \IM \left( ( g | D_{\omega, \lambda}^{\min} (f-g) ) + ( f-g | D_{\omega, \lambda}^{\min} g ) + (f-g | D_{\omega, \lambda}^{\min} (f-g) ) \right),
\end{equation}
so $|\IM ( f| D_{\omega, \lambda} f )| \leq 2 \epsilon \| g \|_{\Dom(D_{\omega, \lambda}^{\min})} + \epsilon^2$. On the other hand for any $t >0$ we have
\begin{subequations} \label{eq:imag_definit}
\begin{align}
& |\IM ( f| D_{\omega, \lambda} f )|  \geq \frac{|\IM(\lambda+\omega)|}{t} \int_0^t |f_1(x)|^2 \D x + \frac{|\IM(\lambda-\omega)|}{t} \int_0^t |f_2(x)|^2 \D x \tag{\ref{eq:imag_definit}} \\
& \geq \frac{|\IM(\lambda+\omega)|}{t} \left( \int_0^t |g_1(x)|^2 \D x - 2 \epsilon \| g \|^2 \right) + \frac{|\IM(\lambda-\omega)|}{t} \left( \int_0^t |g_2(x)|^2 \D x - 2 \epsilon \| g \|^2 \right). \nonumber
\end{align}
\end{subequations}
Comparing the two derived inequalities and taking $\epsilon \to 0$ we find that
\begin{equation}
    \int_0^t |g_1(x)|^2 \D x =     \int_0^t |g_2(x)|^2 \D x = 0. 
\end{equation}
Since $t$ was arbitrary, $g=0$. Case $3.$ may be handled analogously.
\end{proof}

It is convenient to describe the numerical ranges of operators $D_p$
in terms of $[a:b] $. It~can be related to parameters $\omega,
\lambda, \mu$ by recalling that $[a:b] = [- \mu : \omega + \lambda]$
if $\omega +\lambda \neq 0$ and $[a:b] = [\omega - \lambda : - \mu]$
if $\omega - \lambda \neq 0$. No such expression exists on the zero
fiber. We will also choose a representative $(a,b) \in [a:b]$. We note
that the condition $\IM(\overline b a) =0$ is equivalent to the existence of a real representative $(a,b)$, which is also equivalent to the statement that $[a:b]$ belongs to the real projective line $\mathbb{RP}^1$. If $[a:b] \notin \mathbb{RP}^1$, then $\sgn(\IM(\overline b a)) = \sgn \left( \IM \left( \frac{a}{b} \right) \right)$.

\begin{proposition}
The numerical range of $D_p$ may be characterized as follows.
\begin{enumerate}
    \item If $\omega, \lambda \in \R$ and $[a:b] \notin \mathbb{RP}^1$, then $\Num(D_p) = \R \cup \C_{- \IM \left( \overline b a \right)}$.
    \item If $\RE(\mu)=0$ and $\IM(\overline b a) \IM(\lambda)<0$, then $\Num(D_p) =
      \C$.
    \item If $\RE(\mu)<0$ and $[a:b] \notin \mathbb{RP}^1$, then $\Num(D_p) = \C$.
    \item In every other case $\Num(D_p) = \Num(D_{\omega, \lambda}^{\min})$.
\end{enumerate}
\end{proposition}
\begin{proof}
If $p = \overline p$, then $D_p$ is self-adjoint, so $\Num(D_p) \subset \R = \Num(D_{\omega, \lambda}^{\min}) \subset \Num(D_p)$. If~$|\IM(\omega)| > |\IM(\lambda)|$, then $\C = \Num(D_{\omega, \lambda}^{\min}) \subset \Num(D_{p})$. 

Let $\eta(x) = x^{\mu} \begin{bmatrix} a \\ b \end{bmatrix}$ and consider $f = \begin{bmatrix} f_1 \\ f_2 \end{bmatrix} \in C_c^{\infty}(\R_+,\C^2) + \mathrm{span} \{ \chi \eta \}$. Then
\begin{equation}
\IM ( f | D_{p} f ) = \IM \int_{0}^{\infty} \left[ \frac{\D}{\D x} \left( \overline{f_2(x)} f_1(x)  \right) - \frac{(\lambda + \omega) |f_1(x)|^2 + (\lambda - \omega) |f_2(x)|^2}{x} \right] \D x.
\label{eq:form}
\end{equation}

By construction, there exist $x_0 >0$ and $c \in \C$ such that for $x < x_0$ we have $f(x) = c \, \eta(x)$, and hence $\overline{f_2(x)} f_1(x) = \IM(\overline b a) x^{2 \RE(\mu)}$. If $\RE(\mu)>0$ or $\IM(\overline b a)=0$ (which is equivalent to $[a:b] \in \mathbb{RP}^1 \subset \CP^1$), then $\overline{f_2(x)} f_1(x)$ vanishes for $x$ sufficiently large and for $x \to 0$. Therefore $\int_{0}^{\infty} \frac{\D}{\D x} \left( \overline{f_2(x)} f_1(x) \right) \D x =0$ and the proof goes as for Proposition \ref{num_range_minimal}. 

We consider the case $\RE(\mu) = 0$ and $\IM(\overline b a) \neq 0$. Then
\begin{equation}
\IM ( f | D_p f ) = - |c|^2 \IM(\overline b a) - \IM \int_{0}^{\infty}  \frac{(\lambda + \omega) |f_1(x)|^2 + (\lambda - \omega) |f_2(x)|^2}{x}  \D x.
\label{eq:im_mu_im_fhf}
\end{equation}
If $\omega, \lambda \in \R$, then $\IM ( f | D_p f ) = - |c|^2 \IM(\overline b a)$ and we have $\R = \Num(D_{\omega, \lambda}^{\min}) \subset \Num(D_p)$, so~$\Num(D_p) = \{ k \in \C \, | \, \IM(\overline b a) \IM(k) \leq 0 \}$. In the case $|\IM(\omega)| \leq |\IM(\lambda)| \neq 0$ there are two possibilities. If $\IM(\overline b a) \IM(\lambda)>0$, then all terms in \eqref{eq:im_mu_im_fhf} have the same sign and one has $\Num(D_p) = \Num(D_{\omega, \lambda}^{\min})$. Otherwise $\Num(D_p) = \C$. Indeed, consider $f = \frac{\chi \eta}{\| \chi \eta \|}$ with shrinking support of $\chi \geq 0$. A simple calculation shows that for these functions the integrand in \eqref{eq:im_mu_im_fhf} vanishes, while the first term grows without bound.

Next, we suppose that $\RE(\mu) <0$, $\IM(\overline b a) \neq 0$. Put $f = \varphi \eta$ with $\varphi \in C^{\infty}([0 ,\infty [, \R)$ vanishing exponentially at infinity. Then $f \in \Dom(D_p)$ and $(D_p f)(x) = \varphi'(x) \begin{bmatrix} 0 & -1 \\ 1 & 0 \end{bmatrix} \eta(x)$. Thus
\begin{equation}
( f | D_p f ) = 2 \i \,  \IM(\overline b a) \int_0^{\infty} \varphi(x) \varphi'(x) x^{2 \RE(\mu)} \D x.  
\end{equation}
If $\varphi \neq 0$ vanishes at zero, the integral is positive, as can be seen by integrating by parts:
\begin{equation}
\int_0^{\infty} \varphi(x) \varphi'(x) x^{2 \RE(\mu)} \D x =  \frac{1}{2} \int_0^{\infty} \frac{\D}{\D x} (\varphi(x)^2) x^{2 \RE(\mu)} \D x = - \RE(\mu) \int_0^{\infty} \varphi(x)^2 x^{2 \RE(\mu)-1} \D x > 0.  
\end{equation}
On the other hand, for $\varphi(x) = \e^{- \frac{x}{2}}$ the integral is negative:
\begin{equation}
\int_0^{\infty} \varphi(x) \varphi'(x) x^{2 \RE(\mu)} \D x = - \frac{\Gamma(2 \RE(\mu)+1)}{2} <0.
\end{equation}
By Proposition \ref{num_range_minimal} and the fact that $\Num(D_p)$ is a convex cone, we have $\Num(D_p) = \C$.
\end{proof}

We adopt the convention saying that operators with the numerical range
contained in the closed upper half-plane are called {\em
  dissipative}. Dissipative operators which are not properly contained
in another dissipative operator are said to be {\em maximally
  dissipative}. This condition is equivalent to the inclusion of the spectrum in the closed upper half plane. Maximally dissipative operators may also be characterized as operators $D$ such that $\i D$ is the generator of a semigroup of contractions.

\begin{corollary}
$\pm D_p$ is a dissipative operator if and only if one of the following (mutually exclusive) statements holds:
\begin{itemize}
\item $\omega, \lambda \in \R$ and $\mp\IM(\overline b a) \geq0$.
\item $\pm \IM(\lambda)<0$, $|\IM(\omega)| \leq |\IM(\lambda)|$ and $\RE(\mu) >0$. 
\item $\pm \IM(\lambda)<0$, $|\IM(\omega)| \leq |\IM(\lambda)|$, $\RE(\mu) =0$ and $\pm \IM(\overline b a) \leq 0$.
\item $\pm \IM(\lambda)<0$, $|\IM(\omega)| \leq |\IM(\lambda)|$, $\RE(\mu)<0$ and
  $\IM(\overline b a) =0$.
\end{itemize}
Furthermore, if these conditions are satisfied then $\pm D_p$ is maximally dissipative. 
\end{corollary}

\begin{corollary}
Let $\omega, \lambda$ be such that $\pm D_{\omega, \lambda}^{\min}$ is dissipative, i.e.\ $|\IM(\omega)| \leq |\IM(\lambda)|$, $\pm \IM(\lambda) \leq 0$. There exists $p \in \cM_{- \frac12}$ such that $D_{\omega, \lambda}^{\min} \subset D_p$ and $\pm D_p$ is maximally dissipative. In particular $\pm D_{\omega, \lambda}^{\min}$ admits a maximally dissipative extension which is homogeneous and contained in $\pm D_{\omega, \lambda}^{\max}$.
\end{corollary}
\begin{proof}
We present the proof for the upper sign. The other part of the statement then follows by taking complex conjugates. If $\omega, \lambda \in \R$, it is possible to choose $p$ with $\mp \IM(\overline b a) \geq 0$. Now let $\IM(\lambda) <0$, $|\IM(\omega)| \leq |\IM(\lambda)|$. If $\omega^2 - \lambda^2 \notin ] - \infty, 0 ]$, we can choose $\mu$ with $\RE(\mu) >0$.

Next suppose that $\omega^2 - \lambda^2 \leq 0$. If the inequality is strict, then there exist two possible choices of $\mu$ differing by a sign, so the condition $\IM(\overline b a) \leq 0$ is satisfied for at least one choice. If $\omega^2 - \lambda^2 =0$, then either $\omega + \lambda$ or $\omega - \lambda$ vanishes. We may assume that it is not true that both vanish, because this is covered by the case $\omega, \lambda \in \R$. Then $[a:b]=[0:1]$ or $[a:b]=[1:0]$. 
\end{proof}

\section{Mixed boundary conditions} \label{mixed-bc}

In this section we discuss operators $D_{\omega,\lambda}^f$ introduced around equation \eqref{eq:Df_strict_inc}. Hence $\omega,\lambda$ are restricted to the region $| \RE \sqrt{\omega^2 - \lambda^2}| < \frac12$. 

\begin{proposition}
$D_{\omega,\lambda}^f$ is closed, self-transposed and $\sigma_{\ess}(D_{\omega,\lambda}^f) = \sigma_{\ess,0}(D_{\omega,\lambda}^f) = \R$.
\end{proposition}
\begin{proof}
The self-transposedness follows from \cite[Proposition
3.21]{DeGe19_01}. The statement about the essential spectrum follows from Corollary \ref{ess_spec_min_max}.
\end{proof}

Operators $D_{\omega,\lambda}^f$ can be organized in a holomorphic family as follows. Let
\begin{equation}
    \cM^{\mix} = \left \{ (\omega,\lambda,[a:b]) \in \C^2 \times \CP^1 \, | \, \RE \sqrt{\omega^2 - \lambda^2}| < \frac12 \right \}.
\end{equation}
We define $D^\mix_{\omega,\lambda,[a:b]}$ to be $D_{\omega,\lambda}^{f_{\omega,\lambda,[a:b]}}$, where $f_{\omega,\lambda,[a:b]}$ is a (unique up to a multiplicative constant) solution of $D_{\omega,\lambda} f_{\omega,\lambda,[a:b]} =0$ whose value at $x=1$ belongs to the ray $[a:b]$ in $\C^2$.

\begin{proposition}
$D_{\omega,\lambda,[a:b]}^\mix$ form a~holomorphic family of operators on $\cM^\mix$. One has $\overline{D_{\omega,\lambda,[a:b]}^\mix} = D_{\overline{\vphantom{\lambda} \omega},\overline \lambda,[\overline{\vphantom{\lambda} a}:\overline{\vphantom{\lambda} b}]}^\mix$, so $D_{\omega,\lambda,[a:b]}^\mix$ is self-adjoint if and only if $\omega,\lambda \in \R$, $[a:b] \in \mathbb{RP}^1$. 
\end{proposition}
\begin{proof}
Only the holomorphy of $D_{\omega,\lambda,[a:b]}^\mix$ requires some justification. Define
\begin{equation}
    T_{\omega,\lambda,[a:b]} : H_0^1(\R_+,\C^2) \oplus \C \ni (g,t) \mapsto g + t \chi f_{\omega,\lambda,[a:b]},
\end{equation}
where $\chi \in C_c^\infty([\R_+,\infty[)$ is equal to $1$ near $0$. It is easy to check that $T_{\omega,\lambda,[a:b]}$ form a~holomorphic family of bounded injective operators with $\Ran(T_{\omega,\lambda,[a:b]}) = \Dom(D_{\omega,\lambda,[a:b]}^{\mix})$ such that $D_{\omega,\lambda,[a:b]}^{\mix} T_{\omega,\lambda,[a:b]}$ form a holomorphic family of bounded operators. 
\end{proof}


Next we describe the point spectra of nonhomogeneous operators
$D_{\omega,\lambda}^f$. For this purpose it is not very convenient to
use the parametrization by points of $\cM^\mix$.

Below we treat the logarithm, denoted $\mathrm{Ln}$,  as a set-valued function, more
precisely,
\begin{equation}
  \mathrm{Ln}(z):=\{u\ |\ z=\e^u\}.
  \end{equation}

\begin{proposition}
Consider the point spectrum of $D_{\omega,\lambda}^f$ for various $\omega,\lambda, f$. All eigenvalues are non-degenerate and zero is never an eigenvalue. For $k \neq 0$, we split the discussion into several cases. We say that a pair $(k,\pm)$ is admissible if either $k \in \R^\times$, $|\IM(\lambda)| > \frac12$ and $\pm = \sgn(\IM(\lambda))$ or $k \in \C \setminus \R$ and $\pm = \sgn(\IM(k))$.
\begin{enumerate}
\item Case $\mu \neq 0$. We select
  select a square root $\mu = \sqrt{\omega^2 - \lambda^2}$, or
  equivalently, we fix $p \in \cM_{-\frac12}$ lying over $\omega,\lambda$.  
All nonhomogeneous realizations of $D_{\omega,\lambda}$ correspond 
to
\begin{equation}
  f(x) = \begin{bmatrix} \omega - \lambda \\ - \mu \end{bmatrix}
x^\mu + \kappa \begin{bmatrix} \omega - \lambda \\ \mu \end{bmatrix}
x^{- \mu}\end{equation} with $\kappa\in\mathbb{C}^\times$. Let
     \begin{equation}
         c_{p,\pm} = \frac{\omega}{\lambda \mp \i \mu} \frac{\Gamma(2 \mu +1)}{\Gamma(-2 \mu +1)} \frac{\Gamma(1 - \mu \mp \i \lambda)}{\Gamma(1 + \mu \mp \i \lambda)}
     \end{equation}
     Away from $\mu=0$, $c_{p,\pm}$ is a holomorphic function of $\omega,\lambda, \mu$ valued in $\C \cup \{ \infty \}$. $k$~is an eigenvalue if and only if $\kappa (\mp 2 \i k)^{2 \mu} = c_{p, \pm}$ and $(k, \pm)$ is admissible. $D_{\omega,\lambda}^f$ has no eigenvalues in $\C_\pm$ if $c_{p, \pm} \in \{ 0 ,\infty \}$. Away from these loci, eigenvalues in $\C_\pm$ vary continuously with parameters, possibly (dis)appearing on the real axis. They form a~discrete subset of a half-line if $\mu \in \i \R$, of a circle if $\mu \in \R$ and of a~logarithmic spiral otherwise. If $\mu \not \in \i \mathbb{R}$, the set of eigenvalues is finite. More precisely, it is given by the union of the following two sets: 
\begin{align}
\left\{k=\pm\frac{\i}{2}\e^w\ |\
w\in\frac{1}{2\mu}\mathrm{Ln}(c_{p,\pm}),\quad
-\frac\pi2<\mathrm{Im} (w)<
  \frac\pi2 \right\},&\quad|\IM(\lambda)|\leq\frac12,\\
\left\{k=\pm\frac{\i}{2}\e^w\ |\
w\in\frac{1}{2\mu}\mathrm{Ln}(c_{p,\pm}),\quad
-\frac\pi2\leq\mathrm{Im}(w)\leq
\frac\pi2 \right\},&\quad|\IM(\lambda)|>\frac12.
\end{align}
\item Case $\mu=0$, $(\omega,\lambda)\neq(0,0)$.
  All nonhomogeneous realizations of $D_{\omega,\lambda}$ are
  parametrized
  by $\nu\in\mathbb{C}$ and
  \begin{align}
    f(x) = \begin{bmatrix} 1 \\ 0 \end{bmatrix} -2 \lambda(
    \ln(\e^{2\gamma} x)  + \nu) \begin{bmatrix} 0 \\
      1 \end{bmatrix} &\quad\text{for}\quad
                        \omega=\lambda \neq 0,\\
    f(x) = \begin{bmatrix} 0 \\ 1  \end{bmatrix} + 2 \lambda(
    \ln(\e^{2 \gamma} x) + \nu) \begin{bmatrix} 1 \\
      0 \end{bmatrix}&\quad\text{for}\quad\omega = - \lambda \neq
                       0.\end{align}
                     In both cases $k$ is an eigenvalue if and only if
                     $\ln(\mp 2 \i k) + \psi(1 \mp \i \lambda) \mp
                     \frac{\i}{2 \lambda} = \nu$ and $(k,\pm)$ is
                     admissible.
                     There is at most one eigenvalue in $\C_+$ and at most one eigenvalue in $\C_-$. The eigenvalue in $\C_\pm$ exists if and only if $\pm \i \lambda \not \in \N$ and $\RE \left( \exp \left( \nu - \psi(1 \mp \i \lambda) \mp \frac{\i}{2 \lambda} \right) \right) >0 $.
     \item Case $\omega = \lambda =0$, $f(x) = \begin{bmatrix} 1 \\ \kappa \end{bmatrix}$. $k$ is an eigenvalue if and only if $k \not \in \R$ and $\kappa = \i \, \sgn(\IM(k))$.
\end{enumerate}
\end{proposition}
\begin{proof}
An eigenvector of $D_{\omega,\lambda}$ square integrable away from the origin is necessarily of the form $\zeta_p^\pm(k,\cdot)$ with an admissible $(k , \pm)$. It belongs to the domain of $D_{\omega,\lambda}^f$ if its asymptotic form for $x \to 0$, obtained from \eqref{eq:K_small_arg}, is proportional to $f$. This yields conditions described in 1.-3. 

Function $c_{p,\pm}$ is meromorphic. In the region $|\RE(\mu)| < \frac12$ functions $\frac{\omega}{\Gamma(1+\mu \mp \i \lambda)}$ and $\frac{\lambda \mp \i \mu}{\Gamma(1 - \mu \mp \i \lambda)}$ do not simultaneously vanish anywhere, while $\frac{\Gamma(2 \mu +1)}{\Gamma(-2 \mu +1)}$ is holomorphic and nowhere vanishing. Hence $c_{p, \pm}$ is not of the indeterminate form $\frac{0}{0}$ anywhere. 
\end{proof}


Let us note that eigenfunctions corresponding to real eigenvalues (which exist only for $|\IM(\lambda)| > \frac12$) decay at infinity only as fast as $x^{- |\IM(\lambda)|}$, not exponentially. 

Consider a homogeneous operator $D_p$ with $p \in \cE^\pm$ and its deformations $D_{\omega,\lambda}^f$, with $f$ parametrized by $\kappa$ so that $D_{\omega,\lambda}^f=D_p$ for $\kappa =0$. Then for $\kappa=0$ the point spectrum of $D_{\omega,\lambda}^f$ is $\C_\pm$, but for every $\kappa \neq 0$ it is disjoint from $\C_\pm$.





\appendix

\section{1-dimensional Dirac operators} \label{sec:dir1d}

\subsection{General formalism}
By a {\em 1d Dirac operator on the halfline} we will mean a differential
operator of the form
\begin{equation}
 D = 
\begin{bmatrix}
a(x) & - \partial_x \\
\partial_x & b(x)
\end{bmatrix},
\end{equation}
where $a,b$ are smooth functions on $\R_{+} = ]0,\infty[$. 
In this subsection we treat it as a formal operator acting, say, on the space of distributions on $\R_{+} $ valued
in $\C^2$.
We first  describe a few integral kernels closely related to $D$.

Let $k\in\C$ and
\begin{equation}
\xi(k,x)=\begin{bmatrix}\xi_{\uparrow}(k,x)\\\xi_{\downarrow}(k,x)\end{bmatrix},\qquad 
  \zeta(k,x)=\begin{bmatrix}\zeta_{\uparrow}(k,x)\\\zeta_{\downarrow}(k,x)\end{bmatrix}
\end{equation}
be a pair of linearly independent solutions of the Dirac equation:
\begin{equation}
(D-k)\xi(k, \cdot)=(D-k)\zeta(k, \cdot)=0.
\end{equation}
Let
\begin{equation}
d(k,x):=\det\left[\xi(k,x),\zeta(k,x)\right]=
  \det  \begin{bmatrix}\xi_{\uparrow}(k,x)&\zeta_{\uparrow}(k,x)\\\xi_{\downarrow}(k,x)&\zeta_{\downarrow}(k,x).
\end{bmatrix}
\end{equation}
Then $d(k,x)$ does not depend on $x$, so that one can write $d(k)$ instead.
We define
\begin{align}\label{canonical}
&  G^\leftrightarrow(k;x,y):=
  d(k)^{-1}\xi(k,x)\zeta(k,y)^\T-  d(k)^{-1}\zeta(k,x)\xi(k,y)^\T\\[2ex]
  &=
  d(k)^{-1}\begin{bmatrix}\xi_{\uparrow}(k,x)\zeta_{\uparrow}(k,y) 
    &\xi_{\uparrow}(k,x)\zeta_{\downarrow}(k,y)\\[1ex]
    \xi_{\downarrow}(k,x)\zeta_{\uparrow}(k,y)&\xi_{\downarrow}(k,x)\zeta_{\downarrow}(k,y) \end{bmatrix}
    -  d(k)^{-1}\begin{bmatrix}\zeta_{\uparrow}(k,x) \xi_{\uparrow}(k,y)
      &\zeta_{\uparrow}(k,x)\xi_{\downarrow}(k,y)
      \\[1ex]
      \zeta_{\downarrow}(k,x)\xi_{\uparrow}(k,y)
      &\zeta_{\downarrow}(k,x)\xi_{\downarrow}(k,y)
\end{bmatrix}. \notag
\end{align}
Note that $G^\leftrightarrow(k,x,y)$ is uniquely defined by
\begin{equation}
(D-k) G^\leftrightarrow(k;x,y)=0,\qquad G^\leftrightarrow(k;x,x)=
\begin{bmatrix}0 
    &1\\
    -1&0 \end{bmatrix}.
\end{equation}
We will call it the {\em canonical bisolution.}

We also have the {\em forward and backward Green's operators} given by the kernels
\begin{subequations}
\begin{align}
   G^\rightarrow(k;x,y)=
  & G^\leftrightarrow(k;x,y)\one_{\R_+}(x-y),\\
   G^\leftarrow(k;x,y)=
  & -G^\leftrightarrow(k;x,y)\one_{\R_+}(y-x).  
\end{align}
\end{subequations}
They are uniquely defined by
\begin{subequations}
\begin{align} (D-k)  G^\rightarrow(k;x,y)=\delta(x-y) \begin{bmatrix}1 
    &0\\
    0&1 \end{bmatrix}, &\quad x<y\Rightarrow G^\rightarrow(x,y)=0;
\\
  (D-k)  G^\leftarrow(k;x,y)=\delta(x-y) \begin{bmatrix}1 
    &0\\
    0&1 \end{bmatrix},& \quad x>y\Rightarrow G^\leftarrow(x,y)=0.\end{align}
    \end{subequations}
Note that $G^\leftrightarrow, G^\leftarrow, G^\rightarrow$ do not depend on the choice of $\xi,\zeta$.

Using the eigensolutions $\xi,\zeta$, we can introduce yet  another important integral kernel:
\begin{align}\notag G^{\bowtie}(k;x,y):=&
  -d(k)^{-1}\begin{bmatrix}\xi_{\uparrow}(k,x)\zeta_{\uparrow}(k,y) 
    &\xi_{\uparrow}(k,x)\zeta_{\downarrow}(k,y)\\
    \xi_{\downarrow}(k,x)\zeta_{\uparrow}(k,y)&\xi_{\downarrow}(k,x)\zeta_{\downarrow}(k,y) \end{bmatrix}\one_{\R_+}(y-x)
\\ & -       d(k)^{-1}\begin{bmatrix}\zeta_{\uparrow}(k,x) \xi_{\uparrow}(k,y)
      &\zeta_{\uparrow}(k,x)\xi_{\downarrow}(k,y)
      \\
      \zeta_{\downarrow}(k,x)\xi_{\uparrow}(k,y)
      &\zeta_{\downarrow}(k,x)\xi_{\downarrow}(k,y)
\end{bmatrix}\one_{\R_+}(x-y).
\end{align}
It is also  {\em Green's kernel}, because
  it  satisfies
\begin{align}(D-k) G^{\bowtie}(k;x,y)=\delta(x-y) \begin{bmatrix}1 
    &0\\
    0&1 \end{bmatrix}.\end{align}
$ G^{\bowtie}(k;x,y)$
depends on the choice of the pair of 1-dimensional subspaces $\C\xi(k, \cdot)$,
$\C\zeta(k, \cdot)$ of $\Ker(D-k)$. The resolvents of various closed
realizations
of $D$ are often of this form.

Two classes of 1d Dirac operators have special properties.
The case $a(x)=b(x)$ can be fully diagonalized:
\begin{align} \label{diago}
\begin{bmatrix}
a(x) & - \partial_x \\
\partial_x & a(x)
\end{bmatrix}&=\frac1{\sqrt2}\begin{bmatrix}
  1&\i\\\i&1\end{bmatrix}
            \begin{bmatrix}
              -\i\partial_x+a(x)&0\\0&
              \i\partial_x+a(x)\end{bmatrix}\frac1{\sqrt2}\begin{bmatrix}
  1&-\i\\-\i&1\end{bmatrix}.          
\end{align}
We will analyze 1d Dirac-Coulomb operators of this form in Subsection 
\ref{app:elementary}.

The case $a(x)=-b(x)$ can be brought to an antidiagonal
form, used in supersymmetry:
\begin{align} 
\begin{bmatrix}
a(x) & - \partial_x \\
\partial_x & -a(x)
\end{bmatrix}&=\frac1{\sqrt2}\begin{bmatrix}
  1&-1\\1&1\end{bmatrix}
            \begin{bmatrix}
              0&-\partial_x-a(x)\\
              \partial_x-a(x)&0\end{bmatrix}\frac1{\sqrt2}\begin{bmatrix}
  1&1\\-1&1\end{bmatrix}.        
  \label{eq:susy_H}  
\end{align}
We will analyze 1d Dirac-Coulomb operators of this form in Subsection 
\ref{app.super}.

\subsection{Homogeneous first order scalar operators}
\label{Homogeneous first order scalar operators}

Let $\alpha\in\C$. In this subsection we discuss the differential operator
\begin{equation}
A_{\alpha} =  x^{\alpha} \partial_x x^{- \alpha} = \partial_x -\frac{ \alpha}{x}
\end{equation}
acting on scalar functions. It will be a building block of some
special 1d Dirac-Coulomb operators considered in subsections
\ref{app:elementary}
and \ref{app.super}.

Let us briefly recall  basic results about  realizations of
$A_\alpha$ as a closed operator in $L^2(\R_+)$ following
\cite{BuDeGe11_01}. Proofs of all statements stated in this subsection
without justification can be found therein.
(In \cite{BuDeGe11_01} a different convention was used:
$A_{\alpha} = -\i\partial_x +\frac{\i \alpha}{x}$. Thus
$A_\alpha^{\rm new}=\i A_\alpha^{\rm old}$.)

We let $A_{\alpha}^{\min}$ be the closure (in the sense of operators on $L^2(\R_+)$) of the restriction of $A_{\alpha}$ to $C_c^{\infty}(\R_+)$ and $A_{\alpha}^{\max}$ the restriction of $A_{\alpha}$ to $\Dom(A_{\alpha}^{\max}) = \{ f \in L^2(\R_+) \, | \, A_{\alpha} f \in L^2(\R) \}$. Operators $A_{\alpha}^{\min}$ and $-A_{ -\overline \alpha}^{\max}$ are adjoint to each other.

\begin{proposition} \label{A_dom}
We have $A_{\alpha}^{\min} = A_{\alpha}^{\max}$ if and only if
$|\RE(\alpha)| \geq \frac12$. If $|\RE(\alpha)| < \frac12$, then
$\Dom(A_{\alpha}^{\max}) = \Dom(A_{\alpha}^{\min}) + \C \chi
x^{\alpha}$, where $\chi\in C_\mathrm{c}^\infty(\R_+)$ and $\chi=1$
near $0$. If $\RE(\alpha) \neq \frac{1}{2}$, then $\Dom(A_{\alpha}^{\min})= H_0^1(\R_+)$.
\end{proposition}

Closed realizations of $A_\alpha$ are of two types,
described in the following pair of propositions.

\begin{proposition}
Let $\RE(\alpha) > - \frac12$,
\begin{enumerate}
\item
$\sigma(-\i A_{\alpha}^{\max}) = \R \cup \C_+$
and one has \begin{equation}
(A^{\max}_{\alpha} - \i k)^{-1}f (x) = -  \int_{x}^{\infty} \e^{\i k(x-y)} \left( \frac{x}{y} \right)^{\alpha} f(y) \D y, \qquad \IM(k) < 0 .
\end{equation}
\item If $k \in \C_+$ and $n \geq 1$, then $\Ker\left( (
    A_{\alpha}^{\max} - \i k)^n \right)$ is the space of functions of
  the form $x^{\alpha} \e^{\i k x} q(x)$ with $q$ polynomial of degree
  at most $n-1$. In particular $\bigcup \limits_{n=0}^{\infty}
  \Ker\left( ( A_{\alpha}^{\max} - \i k)^n \right)$ is dense in
  $L^2(\R_+)$.
  \item If $k \in \C_+$, then $A_{\alpha}^{\max} - \i k$ is a Fredholm operator of index $1$.
\end{enumerate}\label{pri1}
\end{proposition}

\begin{proposition} Let $\RE(\alpha) < \frac12$.
    \begin{enumerate}
\item $\sigma(-\i A_{\alpha}^{\min}) = \R \cup \C_-$ and one has
\begin{equation}
(A^{\min}_{\alpha} - \i k)^{-1}f (x) =  \int_{0}^{x} \e^{\i
  k(x-y)} \left( \frac{x}{y} \right)^{\alpha} f(y) \D y, \qquad \IM(k)
> 0 .
\end{equation}
\item $A_{\alpha}^{\min}$ has no eigenvectors.
  \item
If $k \in \C_-$, then $A_{\alpha}^{\min} - \i k$ is a Fredholm
operator of index $-1$.
\end{enumerate}\label{pri2}
\end{proposition}

\noindent{\em Proof of Propositions \ref{pri1} and \ref{pri2}.}
Statements 1. are proven in \cite{BuDeGe11_01}.

2. requires justification only for the first part in the first proposition. We factorize
\begin{equation}
x^{\alpha} \e^{\i k x} q(x) = x^{\i \, \IM(\alpha)} \e^{\i \, \RE(k) x} \left( x^{\RE(\alpha)} \e^{- \IM(k) x} q(x) \right).
\end{equation}
Functions in the parenthesis form a dense set, because for any real
numbers $c >0$, $\beta > - 1$ functions $e^{- \frac{c x}{2} }
x^{\frac{\beta}{2}} L_n^{(\beta)}(cx)$, with $L_n^{(\beta)}$
Laguerre polynomials, form an orthogonal basis (see e.g. \cite{Szego}). Clearly density is unaffected by the prefactor, which amounts to the action of a~certain unitary operator on $L^2$. 

Let us show 3. We consider first the case $|\RE(\alpha)| < \frac12$. Then we have explicit inverses modulo rank one operators. 

If $\IM(k)>0$, then $A^{\min}_{\alpha} - \i k$ is invertible and its inverse is a right inverse for $A_{\alpha}^{\max} - \i k$. Thus $A_{\alpha}^{\max} - \i k$ is surjective. We already know that its kernel is one-dimensional. 

If $\IM(k)<0$, then $(A_{\alpha}^{\max} - \i k)^{-1} : L^2(\R_+) \to \Dom(A_{\alpha}^{\max})$ is continuous. The range of $A_{\alpha}^{\min} - \i k$ is the preimage of $\Dom(A_{\alpha}^{\min})$, which is a closed subspace of $\Dom(A_{\alpha}^{\max})$ of codimension one. Hence $\Ran(A_{\alpha}^{\min} - \i k)$ is a closed subspace of $L^2(\R_+)$ of codimension one.


To extended the result beyond the strip $|\RE(\alpha)|< \frac{1}{2} $, note that $(A_{\alpha}^{\min} - \i k)^{-1} -
(A_{\beta}^{\min} - \i k)^{-1}$ (resp. $(A_{\alpha}^{\max} - \i
k)^{-1} - (A_{\beta}^{\max} - \i k)^{-1}$) has a square-integrable integral kernel for
$\RE(\alpha), \, \RE(\beta) < \frac{1}{2}$ and $k \in \C_+$
(resp. $\RE(\alpha), \, \RE(\beta) >  - \frac{1}{2}$ and $k \in
\C_-$). Therefore it is a Hilbert-Schmidt operator, in particular
compact. By Corollary \ref{ess_spec}, the essential spectrum of
$A_{\alpha}^{\min}$ and $A_{\alpha}^{\max}$ does not depend on
$\alpha$. From the case $|\RE(\alpha)| < \frac12$ we know that it is
$\R$. The statement about the value of the index is clear. 
\qed

\begin{proposition} \label{A_semigroups}
$ A_{\alpha}^{\max}$ is the generator of a $C^0$-semigroup if and only if $\RE(\alpha) \geq 0$. If this condition is satisfied, it generates the semigroup of contractions
\begin{equation}
(\e^{ t A_{\alpha}^{\max}} f)(x) = x^{\alpha} (x+t)^{- \alpha} f(x+t), \qquad t \geq 0,
\label{eq:Amax_semigroup}
\end{equation}
$-  A_{\alpha}^{\min}$ is the generator of a $C^0$-semigroup if and only if $\RE(\alpha) \leq 0$. If this condition is satisfied, it generated a semigroup of contractions
\begin{equation}
(\e^{-  t A_{\alpha}^{\min}}f)(x) = x^{\alpha} (x-t)^{- \alpha} f(x-t), \qquad t \geq 0.
\end{equation}
Here we put $f(x-t)=0$ if $x-t<0$.

If $|\RE(\alpha)| < \frac12$, the operators $- A_{\alpha}^{\max}$ and $A_{\alpha}^{\min}$ are not generators of $C^0$-semigroups.
\end{proposition}
\begin{proof}
We present a proof of the statements concerning
$A_{\alpha}^{\max}$. The others can be proven analogously. It is
elementary to check that for $\RE(\alpha) \geq 0$ the right hand side
of \eqref{eq:Amax_semigroup} defines a $C^0$-semigroup of contractions
with the generator $A_{\max}^{\alpha}$. If~$\RE(\alpha) < 0$, we~consider the same expression for $f \in C_c^{\infty}(]1, \infty[)$. Then for $t \leq 1$ it is the unique solution of the Cauchy problem $ \frac{\D}{\D t} f_t = A_{\alpha}^{\max} f_t$, $f_0 = f$. However, there exists no constant $c$ such that $\| f_t \| \leq c \| f \|$ for every $t \in [0,1]$ and $f$. Thus $ A_{\alpha}^{\max}$ is not a generator. If~$|\RE(\alpha)| < \frac12$, then $\sigma( - A_{\alpha}^{\max})$ is the right closed complex half-plane, so $-  A_{\alpha}^{\max}$ is not a generator. 
\end{proof}

\subsection{Dirac-Coulomb Hamiltonians with \texorpdfstring{$\omega=0$}{omega=0}}
 \label{app:elementary}

Dirac-Coulomb Hamiltonians with $\omega=0$ can be reduced to operators
$A_\alpha$ studied in  Subsection \ref{Homogeneous first order scalar operators}.
Therefore, they can be analyzed using elementary functions only.

Let us set $W:=\frac{1}{\sqrt{2}} \begin{bmatrix} 1 & \i \\ \i & 1 \end{bmatrix} $.
Using  \eqref{diago} we obtain for all $\lambda$
\begin{align}
D_{0,\lambda}^{\min} = W \begin{bmatrix}-\i A_{\i \lambda}^{\min} & 0 
  \\ 0 & \i A_{- \i \lambda}^{\min} \end{bmatrix} W^{-1},\qquad
         D_{0,\lambda}^{\max} = W \begin{bmatrix}-\i A_{\i \lambda}^{\max} & 0 
  \\ 0 & \i A_{- \i \lambda}^{\max} \end{bmatrix} W^{-1}.
\end{align}

Consider now the homogeneous holomorphic family. Note first that
 $\omega=0$ implies $\mu=\pm\i\lambda$.
We set $D_{\lambda}^{\pm}:=D_{0,
  \lambda, \pm \i \lambda, [\mp \i : 1]}$. Note that $(0,
  \lambda, \pm \i \lambda, [\mp \i : 1])\in\cE^{\pm}$. We have:
\begin{subequations}
\begin{align}
D_{\lambda}^+ &= W \begin{bmatrix}-\i A_{\i \lambda}^{\max} & 0 \\ 0 & \i A_{- \i \lambda}^{\min} \end{bmatrix} W^{-1}, \qquad \RE(\i \lambda) > - \frac{1}{2}, \\
D_{\lambda}^- &= W \begin{bmatrix} -\i A_{\i \lambda}^{\min} & 0 \\ 0 & \i A_{- \i \lambda}^{\max} \end{bmatrix} W^{-1}, \qquad \RE(- \i \lambda) > - \frac{1}{2},
\label{eq:D_elementary}
\end{align}
\end{subequations}

Below $\sigma_2$ is the Pauli matrix $\begin{bmatrix} 0 & -\i \\ \i & 0 \end{bmatrix}$. Matrices $\frac{1 \pm \sigma_2}{2}$ are its spectral projections.

\begin{proposition}
We have $\sigma(D_{\lambda}^+) = \R \cup \C_+$ and 
\begin{equation}
(D_{\lambda}^+ - k)^{-1} = (-\i A_{\i \lambda}^{\max}-k )^{-1}
\frac{1+\sigma_2}{2} - (-\i A_{- \i \lambda}^{\min}+k)^{-1} \frac{1 - \sigma_2}{2}, \qquad \IM(k)<0,
\end{equation}
whereas $\sigma(D_{\lambda}^-) = \R \cup \C_-$ and
\begin{equation}
(D_{\lambda}^- - k)^{-1} = (-\i A_{\i \lambda}^{\min}-k )^{-1}
\frac{1+ \sigma_2}{2} - (-\i A_{- \i \lambda}^{\max}+k)^{-1} \frac{1 - \sigma_2}{2}, \qquad \IM(k)>0.
\end{equation}
\end{proposition}
\begin{proof}
Follows from identities \eqref{eq:D_elementary} and Proposition \ref{pri1}, \ref{pri2}.
\end{proof}

  \begin{proposition}
    $D_{\lambda}^{\pm} - k $ with $k \in \C_{\pm}$ are Fredholm of
    index $0$. \end{proposition}
\begin{proof} Indeed, by Propositions \ref{pri1} and \ref{pri2}
they are direct sums of two Fredholm operators with indices $1$ and
$-1$. \end{proof}

\begin{proposition}
Let $k \in \C_\pm$. Then $\bigcup \limits_{n=0}^{\infty} \Ker( (D_{\lambda}^\pm - k)^n)$ is a dense subspace of $L^2(\R_+) \begin{bmatrix} \mp \i \\ 1 \end{bmatrix}$.
\end{proposition}

\begin{proposition}
$ \i D_{\lambda}^+$ is the generator of a $C^0$-semigroup if and only if $\IM(\lambda) \leq 0$. Then it generates the semigroup of contractions
\begin{equation}
\e^{\i t D_{\lambda}^+} = \e^{ t A_{\i \lambda}^{\max}} \frac{1 + \sigma_2}{2} + \e^{-  t A^{\min}_{- \i \lambda}} \frac{1 - \sigma_2}{2}, \qquad t \geq 0.
\end{equation}
$- \i D_{\lambda}^-$ is the generator of a $C^0$-semigroup if and only if $\IM(\lambda) \geq 0$. Then it generates the semigroup of contractions
\begin{equation}
\e^{-\i t D_{\lambda}^-} = \e^{- t A_{\i \lambda}^{\min}} \frac{1 + \sigma_2}{2} + \e^{  t A^{\max}_{- \i \lambda}}\frac{1 - \sigma_2}{2}, \qquad t \geq 0.
\end{equation}
Operators $- \i D_{\lambda}^+$ and $\i D_{\lambda}^-$ are not generators of $C^0$-semigroups.
\end{proposition}

\subsection{Hankel transformation}

The following proposition is proven e.g. in \cite{BuDeGe11_01}.
\begin{proposition}
Let $\RE(m) \geq -1$. We define
\begin{equation}
(\mathcal F_m^{\pre} f)(x) = \int_0^{\infty} J_m(xy) \sqrt{xy} f(y) \D y, \qquad f \in C_c^{\infty}(\R_+), 
\end{equation}
where $J_m$ is the Bessel function. Then $\mathcal F_m^{\pre}$ extends
to a bounded operator $\mathcal F_m$ on $L^2(\R_+)$, known as the {\em Hankel transformation}. $\mathcal F_m$ is a self-transposed involution, unitary if $m$ is real.
\end{proposition}

Recall from Subsection \ref{Remarks about notation}
that the operator $X$ is defined by \[ (Xf)(x)= x f(x), \qquad \Dom(X) = \{ f \in L^2(\R_+) \, | \, x f(x) \in L^2(\R_+) \}. \]

\begin{proposition}
If $\RE(\alpha) > - \frac12$, one has
\begin{equation}
\mathcal F_{\alpha + \frac12} A_{\alpha}^{\max} \mathcal F_{\alpha - \frac12} = -X, \qquad \mathcal F_{\alpha - \frac12} A_{- \alpha}^{\min} \mathcal F_{\alpha + \frac12} =  X.
\label{eq:AF_id}
\end{equation}
\end{proposition}
\begin{proof}
Using the identity 
\begin{equation}
x^{-m} \frac{\D}{\D x} x^m J_m(x) =J_{m-1}(x)
\label{eq:Bessel_der}
\end{equation}
one checks that
\begin{equation}
(\mathcal F_{\alpha + \frac12}^{\pre} A_{\alpha} f)(x) = - x (\mathcal F_{\alpha- \frac12}^{\pre} f)(x)
\label{eq:AF_intertwine}
\end{equation}
for $f \in C_c^{\infty}(\R_+)$. If $|\RE(\alpha)| < \frac12$, \eqref{eq:AF_intertwine} may be checked to hold also for $f(x) = \chi(x) x^{\alpha}$. Taking closures we obtain $\mathcal F_{\alpha + \frac12} A_{\alpha}^{\max} \subset - X \mathcal F_{\alpha - \frac12}$, so $\mathcal F_{\alpha + \frac12} A_{\alpha}^{\max} \mathcal F_{\alpha - \frac12} \subset - X$. Since $\mathcal F_{\alpha + \frac12} A_{\alpha}^{\max} \mathcal F_{\alpha - \frac12}$ is a~closed operator and $C_c^{\infty}(\R_+)$ is a dense subspace of $\Dom(X)$ with respect to the graph topology, the opposite inclusion will be established by demonstrating that $C_c^{\infty}(\R_+) \subset \Dom(\mathcal F_{\alpha + \frac12} A_{\alpha}^{\max} \mathcal F_{\alpha - \frac12})$. Let $f \in C_c^{\infty}(\R_+)$. It is clear that $\mathcal F_{\alpha - \frac12} f$ is a smooth function. Using the identity $\left( x \frac{\D}{\D x} - y \frac{\D}{\D y} \right) \sqrt{xy} J_m(xy) =0$ we find
\begin{equation}
\frac{\D}{\D x} (\mathcal F_{\alpha - \frac12} f)(x) = \frac{1}{x} \int_0^{\infty} \sqrt{xy} J_m(xy) \frac{\D}{\D y} \frac{f(y)}{y} \D y.
\end{equation} 
Since $\frac{\D}{\D y} \frac{f(y)}{y}$ is in $L^2(\R_+)$, we get that $\frac{\D}{\D x} (\mathcal F_{\alpha - \frac12} f)(x)$ is square-integrable over $[1 , \infty[$. Next we use the series expansion of $J_m$ to find that for small $x$
\begin{equation}
(\mathcal F_{\alpha - \frac12} f)(x) = \frac{\int_0^{\infty} y^{\alpha} f(y) \D y}{2^{\alpha - \frac12} \Gamma(\alpha + \frac12)}   x^{\alpha} + O(x^{\alpha+1}). 
\end{equation}
Hence $\mathcal F_{\alpha - \frac12} f \in \Dom(A_{\alpha}^{\max})$, so $f \in \Dom (A_{\alpha}^{\max} \mathcal F_{\alpha - \frac12}) = \Dom (\mathcal F_{\alpha + \frac12} A_{\alpha}^{\max} \mathcal F_{\alpha - \frac12})$. We proved the first equality in \eqref{eq:AF_id}. The other one may be obtained by taking the transpose.
\end{proof}

Following \cite{DeRi17_01} (see also \cite{BuDeGe11_01}), we consider the formal differential operator
\begin{equation}
L_{m^2} = - \partial_x^2 + \frac{m^2 - \frac{1}{4}}{x^2}.
\end{equation}
We let $L_{m^2}^{\min}$ be the closure of its restriction to
$C_\mathrm{c}^{\infty}(\R_+)$ and $L_{m^2}^{\max}$ be the restriction to
$\Dom(L_{m^2}^{\max}) = \{ f \in L^2(\R_+) \, | \, L_{m^2} f \in L^2(\R_+)
\}$. If $\RE(m) > -1$, operator $H_m$ is defined as the restriction of
$L_{m^2}$ to $\Dom(L_{m^2}^{\min}) + \C \chi x^{m+ \frac{1}{2}}$,
where $\chi$ is a smooth function equal to one in a~neighborhood of
zero. We remark that $H_{\frac12}$ and $H_{-\frac12}$ are the
Dirichlet Laplacian and the Neuman Laplacian,
respectively. Furthermore, $H_m$ can be diagonalized as follows:
\begin{equation}
H_m = \mathcal F_m X^2 \mathcal F_m.
\end{equation}

\subsection{Dirac-Coulomb Hamiltonians with \texorpdfstring{$\lambda=0$}{lambda=0}
}\label{app.super}

Dirac-Coulomb Hamiltonians with $\lambda=0$
can be  analyzed  without Whittaker functions, just with Bessel functions.

Let us set  $U := \frac{1}{\sqrt{2}} \begin{bmatrix} 1 & -1 \\ 1 & 1 \end{bmatrix}$.
Using   \eqref{eq:susy_H} we obtain for all $\omega$
\begin{align}
D_{\omega,0}^{\min} =  U \begin{bmatrix} 0 & - A_{\omega}^{\min} \\  A_{- \omega}^{\min} & 0 \end{bmatrix} U^{-1},\qquad D_{\omega,0}^{\max} =  U \begin{bmatrix} 0 & - A_{\omega}^{\max} \\  A_{- \omega}^{\max} & 0 \end{bmatrix} U^{-1}.
\end{align}

Using Proposition \ref{A_dom} we rewrite the operators $D'^{\pm}_{\omega} :=
D_{\omega, 0 , \pm \omega, [ \mp 1 : 1 ]}$
as
\begin{subequations} \label{were}
\begin{align} 
D'^{+}_{\omega} &=  U \begin{bmatrix} 0 & - A_{\omega}^{\max}  \\ 
  A_{- \omega}^{\min} & 0 \end{bmatrix} U^{-1},\quad \RE(\omega)>-\frac12; \\
D'^{-}_{\omega} &= U \begin{bmatrix} 0 & - A_{\omega}^{\min} \\  A_{- \omega}^{\max} & 0 \end{bmatrix} U^{-1},\quad-\RE(\omega)>-\frac12.
\end{align}
\end{subequations}

\begin{proposition}
  Introduce
  \begin{equation}
    \cW_\omega'^{\pm}:=\frac12
\begin{bmatrix} \mathcal F_{\pm\omega \pm \frac12}{+}\mathcal F_{\pm\omega \mp \frac12}
   &  \mathcal F_{\pm\omega \pm \frac12}{-}\mathcal F_{\pm\omega \mp \frac12} \\
   \mathcal F_{\pm\omega \pm \frac12}{-}\mathcal F_{\pm\omega \mp \frac12} &
   \mathcal F_{\pm\omega \pm \frac12}{+}\mathcal F_{\pm\omega \mp
     \frac12} \end{bmatrix},
\quad \pm\RE(\omega)>-\frac12.\end{equation} Then 
  $    \cW_\omega'^{\pm}$ are involutions and
we have the following diagonalizations
\begin{equation}
  D_{\omega}'^\pm=\cW_\omega'^{\pm}\begin{bmatrix}\mp X & 0 \\ 0&
   \pm X \end{bmatrix}\cW_\omega'^{\pm}.
\end{equation}
\end{proposition}

\begin{proof} We insert \eqref{eq:AF_id} into \eqref{were}:
\begin{subequations}
\begin{gather}
D_{\omega}'^+ = U \begin{bmatrix} \mathcal F_{\omega + \frac12} & 0 \\ 0 & \mathcal F_{\omega - \frac12} \end{bmatrix} \begin{bmatrix} 0 & X \\ X & 0 \end{bmatrix} \begin{bmatrix} \mathcal F_{\omega + \frac12} & 0 \\ 0 & \mathcal F_{\omega - \frac12} \end{bmatrix} U^{-1}, \\
D_{\omega}'^- = U \begin{bmatrix} \mathcal F_{-\omega - \frac12} & 0 \\ 0 & \mathcal F_{-\omega + \frac12} \end{bmatrix} \begin{bmatrix} 0 & -X \\ -X & 0 \end{bmatrix} \begin{bmatrix} \mathcal F_{-\omega - \frac12} & 0 \\ 0 & \mathcal F_{-\omega + \frac12} \end{bmatrix} U^{-1}.
\end{gather}
\end{subequations}
Then we use
\begin{align}
  \begin{bmatrix}0&\pm X\\\pm X&0\end{bmatrix}
=U^{-1}
\begin{bmatrix}\mp X&0\\0&\pm X\end{bmatrix}U.
\end{align}
\end{proof}

\begin{corollary}
We have
\begin{subequations}
\begin{gather}
(D'^+_{\omega})^2 = U \begin{bmatrix} H_{\omega + \frac12} & 0 \\ 0 & H_{\omega - \frac12} \end{bmatrix} U^{-1}, \\
(D'^{-}_{\omega})^2 = U \begin{bmatrix} H_{- \omega - \frac12} & 0 \\ 0 & H_{- \omega + \frac12} \end{bmatrix} U^{-1}.
\end{gather}
\end{subequations}
\end{corollary}

\begin{remark}
At least formally, operators $D_{\omega}'^{\pm}$, $X \sigma_2 $
(declared to be odd) and $(D_{\omega}'^{\pm})^2$, $X^2$, $A$ (declared
to be even) furnish a representation of the Lie superalgebra
$\mathfrak{osp}(1|2)$. We leave a detailed description of this
representation for a future study.
\end{remark}

\section{Dirac Hamiltonian in \texorpdfstring{$d$}{d} dimensions} \label{ddimensions}

Separation of variables of a spherically symmetric Dirac Hamiltonian
in dimension 3
is described in many texts and belongs to the standard curriculum of
relativistic quantum mechanics \cite[p. 267]{Dirac}.
Of course, it is even more
straightforward to solve a rotationally symmetric Dirac Hamiltonian in
dimension 2. However, to our knowledge, the first treatment
in any dimension is due to Gu, Ma and Dong \cite{XIAOetal}.

In this appendix we show that a spherically symmetric Dirac Hamiltonian
in an arbitrary  dimension can be reduced to  1 dimension.
Unlike in  \cite{XIAOetal},
we arrive at the radial Dirac equation by relatively simple
algebraic computations which do not involve a detailed analysis of
representations of the Lie algebra
$\mathfrak{so}(d)$.

The main role in this separation is played by a certain operator
$\kappa$ that commutes with the Dirac operator. This operator in dimension 3
goes back to Dirac himself.
It seems that for the 
first time it has been generalized to  other dimensions in
\cite{XIAOetal}. We analyze this operator in detail.

Recall that operators  belonging to the center of the envelopping
algebra of 
$\mathfrak{so}(d)$ are called {\em Casimir operators of $\mathfrak{so}(d)$.}
One of them, built in a standard way as a bilinear form in the
generators, will be called the {\em square of angular momentum} or
simply the {\em quadratic Casimir} (even though it is not the only Casimir bilinear in generators: these form a vector space generically of dimension $1$, and of dimension $2$ if $d=4$).
$\kappa$ does not coincide with the quadratic Casimir.
One can ask whether $\kappa$ is also a Casimir operator.
We will analyse this question in detail. It~turns out that the answer
is positive in even, and negative in odd dimensions.

\subsection{Laplacian in \texorpdfstring{$d$}{d} dimensions}
\label{app-Dir1}

Spherical coordinates can be interpreted as a map
\begin{subequations}
\begin{align}\R^d\backslash\{0\}\ni x&\mapsto (r,\hat x)\in
\R_+\times  \mathbb{S}^{d-1},\\
\hat x=\frac{x}{|x|},&\quad r=|x|.
\end{align}
\end{subequations}
It induces a unitary map
\begin{equation}L^2(\R^d)\to  L^2(\R_+, 
r^{d-1})\otimes L^2(  \mathbb{S}^{d-1}).\label{simpli0}
\end{equation}
We also have the obvious map
\begin{equation}L^2(\R_+, 
r^{d-1})\ni f\mapsto r^{\frac{d-1}{2}}f\in L^2(\R_+).
\label{simpli}\end{equation}
The product of \eqref{simpli} and \eqref{simpli0}
will be denoted
\begin{equation}U:L^2(\R^d)\to  L^2(\R_+)\otimes L^2(  \mathbb{S}^{d-1}).\label{simpli1}
\end{equation}

The momentum is defined as
\[p_i:=-\i\partial_i.\]
We also introduce the radial momentum
\begin{align}
  R &:=\frac{x}{2|x|}p+p\frac{x}{2|x|}
      =\frac{x}{|x|}p-\i\frac{d-1}{2|x|}.
  \end{align}
  Here is the radial momentum and its square in spherical coordinates:
  \begin{subequations}
\begin{align}
R&      =-\i\partial_r-\i\frac{d-1}{2r},\\
  R ^2&=-\partial_r^2-\frac{d-1}{r}\partial_r 
 + \left(\frac14-\Big(\frac{d-2}{2}\Big)^2 
\right)\frac{1}{r^2}. 
\end{align}
\end{subequations}
After applying $U$ we obtain
\[URU^{-1}=-\i\partial_r.\]

In the standard way we introduce the angular momentum  and its square:
\begin{subequations}
\begin{align}\label{angul}
  L_{ij}&:=x_ip_j-x_jp_i,\\
  L^2&:=\sum_{i<j}L_{ij}^2.
\end{align}
\end{subequations}
They furnish the standard representation of the Lie algebra $\mathfrak{so}(d)$ 
 on $\mathbb{S}^{d-1}$:
 \begin{subequations}
\begin{align}
[
  L_{ij},x_k]&=-\i\delta_{jk}x_i+\i\delta_{ik}x_j,\\
[
L_{ij},p_k]&=-\i\delta_{jk}p_i+\i\delta_{ik}p_j,\\
  \Big[
    L_{ij},L_{kl}\Big]&=-\i\delta_{jk}L_{il}
  -\i\delta_{il}L_{jk}
  +\i\delta_{ik}L_{jl}
  +\i\delta_{jl}L_{ik}.
\end{align}
\end{subequations}
The angular momentum squared $L^2$ is the quadratic Casimir operator of $\mathfrak{so}(d)$.

The representation \eqref{angul} is decomposed into subspaces of
spherical harmonics of the order~$\ell$. The 
representation 
of $\mathfrak{so}(d)$ of this type will be called {\em spherical of degree
  $\ell$}. On this representation we have
\begin{align}L^2=&\ell(\ell+d-2)
 =\Big(\ell+\frac{d-2}{2}\Big)^2-
  \Bigl(\frac{d-2}{2}\Bigr)^2. \label{sphero} \end{align}

The Laplacian on $\R^d$ in the spherical coordinates is
\begin{subequations}
\begin{align}-\Delta&=-\partial_r^2-\frac{(d-1)}{r}\partial_r+\frac{L_d^2}{r^2},\\
&=R ^2+\Big(-\frac14+\Big(\ell+\frac{d-2}{2}\Big)^2\Big)\frac1{r^2}.
  \end{align}
  \end{subequations}
  Sandwiching it with $U$ we obtain
  \begin{equation}
    U(-\Delta)U^{-1}=-\partial_r^2+\Big(-\frac14+\Big(\ell+\frac{d-2}{2}\Big)^2\Big)\frac1{r^2}.
    \end{equation}
 
\begin{remark}
Discussion above is valid even for $d=1$, with $\mathbb S^0 := \{ \pm 1 \}$. This case is peculiar in that the only allowed values of $\ell$ are $0$ and $1$, corresponding to even and odd functions. $d=2$ is also special: $\ell$ takes arbitrary integer values, while for $d \geq 3$ one has $\ell \geq 0$. 
\end{remark} 

\subsection{Dirac operator in \texorpdfstring{$d$}{d} dimensions}
\label{app-Dir2}

Let $\alpha_i$, $i=1, \ldots ,d$ and $\beta$ be the Clifford matrices
acting irreducibly in a finite dimensional space $\cK$. They satisfy the Clifford relations
\begin{equation}\label{alpha}
[\alpha_i,\alpha_j]_+=2\delta_{ij},\quad [\alpha_i,\beta]_+=0,\quad \beta^2=1.
\end{equation}
We recall that $\dim(\cK) = 2^{\lfloor \frac{d+1}{2} \rfloor}$ and
that for even $d$ one has $\beta = \pm \i^{\frac{d}{2}}\prod \limits_{i=1}^d \alpha_i$. The two sign choices give non-isomorphic representations of the Clifford algebra. By averaging arguments, $\cK$ admits a
positive definite hermitian form such that $\beta$ and $\alpha_i$ are unitary and hence hermitian. This form is unique up to positive scalars; we fix one once and for all.

Using the Einstein summation convention unless there is a summation 
sign, we introduce the following operators on $L^2(\R^d)\otimes\cK$:
\begin{subequations}
\begin{align}
  D&
     :=\alpha_i p_i
     ,\\
  T&:=-\i\sum_{i<j}\alpha_i\alpha_j L_{ij}+\frac{d-1}{2}, \label{eq:Tdef} \\
  S&:=\frac{\alpha_ix_i}{|x|},\\
  \kappa&:=\beta T=T\beta
  .\end{align}
  \end{subequations}

\begin{proposition}
  We have
  \begin{subequations}
  \begin{align}
SD&=R+\frac{\i}{|x|}T,&DS&=R-\frac{\i}{|x|}T,\\
    SR&=RS,& ST&=-TS,\\
    S\beta&=-\beta S&    \beta D&=-D \beta,\\
    DT&=-TD,& D\kappa &=\kappa D.\label{idid}
    \end{align}
    \end{subequations}
  \end{proposition}
    \proof
    Let us prove the first identity of
    \eqref{idid}. We have
    \begin{align}
     & [\alpha_i\alpha_jL_{ij},\alpha_kp_k]_+\\
      =&[\alpha_i\alpha_j,\alpha_k]L_{ij}p_k+
         \alpha_k \alpha_i\alpha_j[L_{ij},p_k]_+. \nonumber
         \end{align}   
    Using
    \begin{subequations}
    \begin{align}
      \alpha_j\alpha_k\alpha_i&=\alpha_k\alpha_i\alpha_j+2\delta_{jk}\alpha_i-2\delta_{ij}\alpha_k,\\
      \alpha_i\alpha_j\alpha_k&=\alpha_k\alpha_i\alpha_j+2\delta_{jk}\alpha_i-2\delta_{ik}\alpha_j
   , \end{align}
   \end{subequations}
 we obtain
 \begin{align}
&3\alpha_k\alpha_i\alpha_j[L_{ij},p_k]_+ \nonumber \\
  = &
                \alpha_k\alpha_i\alpha_j[L_{ij},p_k]_+ +\alpha_j\alpha_k\alpha_i[L_{ki},p_j]_++\alpha_i\alpha_j\alpha_k[L_{jk},p_i]_+ \nonumber \\
   =&        \alpha_k\alpha_i\alpha_j\big([L_{ij},p_k]_++[L_{ki},p_j]_++[L_{jk},p_i]_+\big) \nonumber \\
              + & 2\alpha_i[L_{ji},p_j]_+-2\alpha_i[L_{ij},p_j]_+-2\alpha_i[L_{ij},p_i]_+ \nonumber \\
   =&6\alpha_i[L_{ji},p_j]_+=-6  \i \alpha_jp_j (d-1)-12\alpha_i L_{ij} p_j.\label{wew1}
 \end{align}
 Moreover,
 \begin{align}
   [\alpha_i\alpha_j,\alpha_k]L_{ij}p_k&=2(\alpha_i\delta_{jk}-\delta_{ik}\alpha_j)L_{ij}p_k \nonumber \\
   &=4\alpha_iL_{ij}p_j.\label{wew2}
 \end{align}
 Now the sum of $\frac{\i}{6}$\eqref{wew1} and $\frac{\i}{2}$\eqref{wew2} is $(d-1)\alpha_ip_i$.
 \qed

 \subsection{Decomposition into incoming and outgoing Dirac waves} \label{app:Dirac_decomp}

Let 
 \[\Pi_\pm:=\frac12(1\pm S)\]
 be the spectral projections of $S$ onto $\pm1$. Define 
 \begin{equation}
 \cH_\pm:=\Pi_{\pm}(L^2(\mathbb{R}^d)\otimes\cK),\ \text{so that }\ 
   L^2(\mathbb{R}^d)\otimes\cK=\cH_+\oplus\cH_-.
   \end{equation}
 For an operator $B$ on  $L^2(\mathbb{R}^d)\otimes\cK$ let us write
\begin{equation}
B_{\pm\pm}=\Pi_\pm B\Pi_\pm,\qquad B_{\pm\mp}=\Pi_\pm B\Pi_\mp.
\end{equation}
Clearly, 
\begin{subequations}
\begin{align}
  S=\begin{bmatrix}1&0\\0&-1
  \end{bmatrix},&\qquad
                    R=\begin{bmatrix}R_{++}&0\\0&R_{--}
                    \end{bmatrix},\\
    \beta=\begin{bmatrix}0&\beta_{+-}\\\beta_{-+}&0
    \end{bmatrix},&\qquad T
                      =\begin{bmatrix}0&T_{+-}\\T_{-+}&0
                      \end{bmatrix},\\
  D   =\begin{bmatrix}R_{++}&\frac{\i}{|x|}T_{+-}\\-\frac{\i}{|x|}T_{-+}&-R_{--}
  \end{bmatrix},&\qquad
                  \kappa=\begin{bmatrix}\beta_{+-}T_{-+}&0\\0&\beta_{-+}T_{+-} 
  \end{bmatrix}=                                      \begin{bmatrix}T_{+-}\beta_{-+}&0\\0&T_{-+}\beta_{+-} 
  \end{bmatrix}.                                      
  \end{align}
  \end{subequations}
$D$ commutes with the self-adjoint operator $\kappa$. We can
therefore reduce ourselves to the eigenspace of $\kappa$ with eigenvalue
 $\omega\in\R$, denoted $\cH_\omega$ (see subsection \ref{app:spheres} for a description of these eigenspaces). We can write
\begin{equation}
  D   =\begin{bmatrix}R_{++}&\frac{\i \omega}{|x|}\beta_{+-}\\-\frac{\i \omega}{|x|}\beta_{-+}&-R_{--}
  \end{bmatrix}.
  \end{equation}
Using spherical coordinates, we can identify
 $L^2(\R^d)\otimes\cK$ with
 $ L^2(\R_+,r^{d-1})\otimes L^2(\mathbb{S}^{d-1})\otimes\cK$. Applying \eqref{simpli} and treating $\beta_{\pm\mp}$ as
identifications, we can rewrite the above equation as
\begin{equation}
  \label{1reduc}
  D   =\begin{bmatrix}-\i\partial_r&\frac{\i\omega}{r}\\-\frac{\i\omega}{r}&\i\partial_r 
  \end{bmatrix}.
  \end{equation}

The $d$-dimensional  Dirac Hamiltonian can be reduced to
  1 dimension (with $2 \times 2$ matrix structure) if it
  is perturbed by four kinds of radial terms: the electric
potential
$V(r)$, the mass~$m(r)$ (called also the {\em Lorentz scalar}), the radial vector potential $A(r)$ and the anomalous (Pauli) coupling to the electric field $E(r)$.
The reduction \eqref{1reduc} leads to
\begin{align} \notag D_{V,m,A,E}\,:= \,& D+V(r)+m(r)\beta+\frac{A(r)\alpha_i
  x_i}{r}+\frac{\i E(r)\beta\alpha_i x_i}{r}\\[2ex]
  =&
  \begin{bmatrix}-\i\partial_r+V(r)+A(r)&\frac{\i\omega}{r}+m(r)-\i E(r)\\-\frac{\i\omega}{r}+m(r)+\i
    E(r)&\i\partial_r +V(r)-A(r)
  \end{bmatrix}.\label{2reduc}\end{align}
We prefer another form, related by a similarity transformation:
\begin{align} \label{3reduc}
\frac1{\sqrt2}\begin{bmatrix}
  1&-\i\\\i&-1\end{bmatrix}
             D_{V,m,A,E}
                                                                                                    \frac1{\sqrt2}\begin{bmatrix}
  1&-\i\\\i&-1\end{bmatrix}
=\begin{bmatrix}
-\frac{\omega}{r}+E(r)+V(r) & - \partial_r-\i A(r) -m(r)\\
\partial_r +\i A(r)-m(r)& \frac{\omega}{r}-E(r) +V(r)
\end{bmatrix}.
\end{align}
For $m=A=E=0$ and $V=-\frac{\lambda}{r}$, this  is the 1-dimensional Dirac operator studied in our paper. 

We remark that the radial electromagnetic potential $A(r)$ is necessarily pure gauge. Indeed, it enters the Dirac operator only in the combination $\partial_r + \i A(r)$, which may be written as $\e^{- \i \phi(r)} \partial_r \e^{i \phi(r)}$ for a function $\phi(r)$ such that $\phi'(r) = A(r)$. Coupling $E(r)$ arises if the Dirac Lagrangian is extended by the Pauli term, proportional to $\overline \psi \frac{\i}{2} \gamma^\mu \gamma^\nu F_{\mu \nu} \psi$ with a purely electric and radial field strength tensor $F$.

\subsection{Composite angular momentum}

Introduce the spin operators
\begin{align}\label{cliff}
  \sigma_{ij}&:=-\frac\i2[\alpha_i,\alpha_j]  .\end{align}
$\frac12\sigma_{ij}$ yield a representation of $\mathfrak{so}(d)$ on the spin space $\cK$:
\begin{subequations}
\begin{align}
  \Big[
 \frac12\sigma_{ij},\alpha_k\Big]&=-\i\delta_{jk}\alpha_i+\i\delta_{ik}\alpha_j,\\
  \Big[
    \frac12\sigma_{ij},\frac12\sigma_{kl}\Big]&=-\i\delta_{jk}\frac12\sigma_{il}
  -\i\delta_{il}\frac12\sigma_{jk}
  +\i\delta_{ik}\frac12\sigma_{jl}
                                                +\i\delta_{jl}\frac12\sigma_{ik}. \label{eq:sigma_commutator}
                                                      \end{align}
                                                      \end{subequations}
Irreducible representations of $\mathfrak{so}(d)$ 
contained in $\cK$ will be called {\em spinor representations}.  
Their quadratic Casimir is given by
  \begin{align}\label{eigo2}
    \frac{\sigma^2}4=
   \frac14 \sum_{i<j}\sigma_{ij}^2=&\frac{d(d-1)}{8}.
  \end{align}

If $d$ is even, then there are two inequivalent spinor
representations. They correspond to the  eigenspaces of $\beta$ with
eigenvalues $\pm 1$.

If $d$ is odd, then $\cK$ is also a direct sum of two spinor
representations, however they are equivalent
to one another. The decomposition of $\cK$ into irreducible components
exists but 
is clearly non-unique. One possible choice
corresponds to the eigenvalues $\pm1$
of $\beta$.

We also have the composite representation of $\mathfrak{so}(d)$ given by  
\begin{align}
J_{ij}&:=L_{ij}+\frac12\sigma_{ij}.
\end{align}
      Clearly,
      \begin{subequations}
\begin{align}
[J_{ij},J_{kl}]&=-\i\delta_{jk}J_{il}
  -\i\delta_{il}J_{jk}
  +\i\delta_{ik}J_{jl}
                 +\i\delta_{jl}J_{ik},\\
      [J,x\cdot\alpha]&=
  [J,p\cdot\alpha]=[J,p^2] =[J,x^2]=0 .            \end{align}
  \end{subequations}
  The quadratic Casimir of this representation, also called the {\em square of the 
  total angular momentum}, is 
  \begin{subequations}
  \begin{align}\label{eigo1}J^2=&\sum_{i<j}J_{ij}^2=L^2+L\sigma+\frac{\sigma^2}{4},\\
\text{where } \qquad L\sigma&:=\sum_{i<j}L_{ij}\sigma_{ij}.
  \end{align}
\end{subequations}

\begin{proposition}
We have the following relation:
\begin{equation}
\kappa^2=J^2+\frac{(d-1)(d-2)}{8}.
\label{eq:Om2_J2_rel}
\end{equation}
\end{proposition}

\begin{proof}
Directly from the definition we have
\begin{equation}
\kappa^2 = \frac{(d-1)^2}{4} + (d-1) L \sigma + \sum_{\substack{i < j \\ k < l}} L_{ij} L_{kl} \sigma_{ij} \sigma_{kl}.
\label{eq:Omega2_step1}
\end{equation}    
To simplify the last term we write
\begin{equation}
\sigma_{ij} \sigma_{kl} = \frac{1}{2} [\sigma_{ij}, \sigma_{kl}] + \frac{1}{2} [\sigma_{ij}, \sigma_{kl}]_+.
\label{eq:ss_com_anticom}
\end{equation}
A simple expression for the first term is given by \eqref{eq:sigma_commutator}. The second one is
\begin{equation}
\frac{1}{2} [\sigma_{ij}, \sigma_{kl}]_+ = -  \alpha_{[ i} \alpha_j \alpha_k \alpha_{l]} + 2 \delta^{[i}_{[k} \delta^{j]}_{l]}
\label{eq:ss_anticom}
\end{equation}
in which $[ \cdots ]$ denotes skew-symmetrization of the enclosed indices. In order to prove this formula, first note that both sides are skew-symmetric with respect to the transposition of $i$ and $j$ or $k$ and $l$, so we may assume that $i \neq j$ and $k \neq l$. We have three cases. If sets $A = \{ i , j \}$ and $B=\{ k , l \}$ are disjoint, then all $\alpha$ matrices involved anticommute and hence both sides are equal to $- \alpha_i \alpha_j \alpha_k \alpha_l$. If $A \cap B$ has one element, one checks that both sides vanish. Finally, if~$A=B$ then both sides are equal to $\pm 1$, the sign depending on the order of indices.

Now plug \eqref{eq:ss_com_anticom} and \eqref{eq:ss_com_anticom} into \eqref{eq:Omega2_step1}. The term with $\alpha_{[ i} \alpha_j \alpha_k \alpha_{l]}$ drops out after summing over indices because $L_{[ij} L_{kl]}=0$. In the term with \eqref{eq:sigma_commutator} we can replace $L_{ij} L_{kl}$ by $\frac{1}{2} [L_{ij},L_{kl}]$, by skew-symmetry with respect to $ij \leftrightarrow kl$. After simplifications with \eqref{eigo2} and \eqref{eigo1} we obtain the claim.
\end{proof}

Recall that on $\cH_\omega$ the operator $\kappa$ acts as multiplication by $\omega$. We will now characterize $\cH_\omega$ more closely.

\begin{proposition} \label{eq:Hom_inc}
Let $\omega$ be such that $\cH_\omega \neq \{ 0 \}$. Then there exist $\ell$ and subspaces $\cW_\ell,\cW_{\ell-1}\subset L^2(\R^d)$ spherical of degree 
$\ell$ resp. $\ell-1$ such that
\begin{subequations}
\begin{align}
 \cH_{\omega}&\subset  (\cW_\ell{\oplus}\cW_{\ell-1})\otimes\cK,\label{compos}\\
  |\omega| &= \ell + \frac{d-3}{2}, \label{eigo3} \\ 
  \left. J^2 \right|_{\cH_\omega} &=\ell^2+\ell(d-3)+\frac{d^2-9d+16}{8}. \label{eigo4}
\end{align}
\end{subequations}
\end{proposition}

\proof
Exceptional cases $d=1,2$ are easy to analyze separately: one has $\kappa =0$ in the former case and $\kappa = \pm L_{12} + \frac{1}{2} \beta$ (with the sign depending on the choice of Clifford matrices) in the latter. From now on we assume that $d \geq 3$. We note that \eqref{eq:Om2_J2_rel} and $J^2 \geq 0$ imply that $\omega \neq 0$.

$\kappa$ commutes with $\beta$, hence also with
$L\sigma$. Therefore we can decompose $\cH_\omega$
with respect to the eigenvalues of $L\sigma$. From \eqref{eq:Tdef} we obtain
\begin{align}    L\sigma=& \omega \beta -\frac{d-1}{2},
  \end{align}
which has on $\cH_\omega$ two distinct eigenvalues
  \begin{equation}\pm\omega-\frac{d-1}{2}.\label{eigo}\end{equation}
Both sings are realized because $D$ anticommutes with $\beta$ and preserves $\cH_{\omega}$.

Clearly $L^2=J^2-L\sigma-\frac{\sigma^2}{4}$ has on $\cH_\omega$ two distinct
eigenvalues corresponding to \eqref{eigo}.  As seen from
\eqref{sphero}, the representation of orbital angular momentum is uniquely determined by $L^2$.
Therefore, for some $\ell_+,\ell_-\in\N$, $\ell_+ > \ell_-$,
\begin{equation}
 \cH_{\omega}\subset
 (\cW_{\ell_+}{\oplus}\cW_{\ell_-})\otimes\cK.\label{compos1}
\end{equation}

Comparing the identities
\begin{align}
  J^2&=\ell_\pm(\ell_\pm+d-2)\mp |\omega|-\frac{d-1}{2}+\frac{d(d-1)}{8},\\
  J^2&=\omega^2-\frac{(d-1)(d-2)}{8},\end{align}
we obtain the equation 
\begin{equation}
|\omega|(|\omega|\pm 1)=\Big(\ell_\pm+\frac{d-1}{2}\Big) \Big(\ell_\pm+\frac{d-3}{2}\Big).
  \end{equation}
whose solutions take the form $ \ell_+ + \frac{d-3}{2} \in \{ |\omega|, -|\omega| - 1 \}$, $\ell_- + \frac{d-3}{2} \in \{ |\omega|-1, -|\omega| \}$. In both cases the second solution has to be discarded because $\ell_\pm + \frac{d-3}{2} \geq 0$. Hence \eqref{eigo3} holds and $\ell_- = \ell_+ - 1$. Then \eqref{eigo4} is obtain by feeding \eqref{eigo3} into \eqref{eq:Om2_J2_rel}.
  \qed

    We remark that the sign of $\omega$ cannot be obtained from the
    above calculation. Indeed,
    the spectrum of $\kappa$  on
   $L^2(\R^d)\otimes\cK$  is always invariant with respect to $\omega\mapsto-\omega$.
    If $d$ is odd, then $\prod \limits_{j=1}^d \alpha_j$ commutes with
    $L$ and~$\alpha_i$, but anticommutes with $\beta$ and hence with
    $\kappa$.
     If $d$ is even, then  $\kappa$ anticommutes with the parity operator
\begin{equation}\label{parity}
f(x_1,\dots,x_d) \mapsto \beta \alpha_1 f(-x_1, \dots , x_d), \qquad f \in L^2(\R^d, \cK).
\end{equation}
However, this operation does not preserve the type of angular momentum representation. Indeed, it anticommutes with $\beta$ and hence exchanges the two spinor representations.

\subsection{Analysis in various dimensions}

Let us review the lowest dimensions.

\noindent{$\bf d=1$.} There is no angular momentum and one has $\omega =0$. 

\noindent{$\bf d=2$.} Unitary irreducible representations of $\mathfrak{so}(2)$ are enumerated by
spin values $m\in \R$. The corresponding quadratic Casimir is equal to
$m^2$.
There are two types $\cK_{\pm \frac12}$ of spinor representations, corresponding
to $m=\pm\frac12$.
Spherical representations correspond to $\ell \in\Z$.

One convenient choice of Clifford representation is given by Pauli matrices: $\alpha_1 = \sigma_1$, $\alpha_2 = \sigma_2 $, $\beta=\pm \sigma_3$. Then $\kappa = \pm J$ and hence
\begin{align}
  \cH_{\omega}=  (\cW_\ell{\otimes}\cK_{-\frac12})
                \oplus
(\cW_{\ell-1}{\otimes}\cK_{\frac12}).
\end{align}
with $\omega = \pm (\ell - \frac{1}{2}) \in \Z + \frac12$. Sign in the relation between $\omega$ and total angular momentum depends on the choice of sign in $\beta$, but after fixing Clifford matrices it is one-to-one.

    \noindent{$\bf d=3$.} Unitary irreducible representations of $\mathfrak{so}(3)$ are parametrized by 
  spin $j \in \frac{1}{2}\mathbb N$ or the quadratic Casimir
  $j(j+1)$. All spinor representations have the spin $\frac12$.
The representation on $\cH_\omega$
      has spin $\ell-\frac12$. We have $\omega=\pm \ell \in \{ \pm 1 , \pm 2 , \dots \}$, i.e.\ two distinct values of $\omega$ correspond to the same total spin.

    \noindent{$\bf d=4$.}
      We have $\mathfrak{so}(4)\simeq \mathfrak{so}(3)\oplus \mathfrak{so}(3)$. More explicitly,
      \begin{align}
        J_1^\pm :=\frac12(\pm J_{12}+ J_{34}),&\quad J_2^\pm
                                                :=\frac1{2}(\pm J_{13}+ J_{42}),\quad
                        J_3^\pm :=\frac1{2}(\pm J_{14}+J_{23})
                        \label{eq:so4_decomp}
      \end{align}
      span two algebras isomorphic to $\mathfrak{so}(3)$ and commuting with one
      another. Let $(J^\pm)^2$ be the corresponding quadratic Casimirs. We
      have
      \begin{equation} J^2=2(J^+)^2+2(J^-)^2.\label{analo}\end{equation}
Thus irreducible representations  of $\mathfrak{so}(4)$ are parametrized by pairs of spins
$(j^+,j^-)\in(\frac12\N)^2$ with the quadratic Casimir
$2j^+(j^++1)+2j^-(j^-+1)$. We have also the obvious analogs of \eqref{eq:so4_decomp} and \eqref{analo} for $L_{ij}$ and
$\frac12\sigma_{ij}$.

Representations of $\mathfrak{so}(4)$ on spherical
harmonics satisfy
\begin{equation}
  L_{12}L_{34}+L_{13}L_{42}+L_{14}L_{23}=0.
  \end{equation} Therefore,
$(L^+)^2=(L^-)^2$. Hence a spherical representation of degree $\ell$
corresponds to the pair of spins $(\frac\ell2,\frac\ell2)$
with the quadratic Casimir 
$\ell(\ell+2)=2\frac{\ell}{2}\Big(\frac\ell2+1\Big)+2\frac{\ell}{2}\Big(\frac\ell2+1\Big)$. Spinor
representations  of $\mathfrak{so}(4)$ are of types $(\frac12,0)$ and $(0,\frac12)$,
distinguished by the eigenvalue of $ \alpha_1 \alpha_2 \alpha_3
\alpha_4$.  They satisfy
\begin{equation}
\alpha_1 \alpha_2 \alpha_3 \alpha_4 \, \sigma_i^\pm = \mp \sigma_i^\pm, \qquad (\sigma^\pm)^2 = \mp \frac32 \alpha_1 \alpha_2 \alpha_3 \alpha_4.
\end{equation}
Furthermore, we have $\beta = \pm \alpha_1 \alpha_2 \alpha_3 \alpha_4$, with the sign in this relation distinguishing Clifford
representation. Using these relations we derive
\begin{equation}
\kappa = \mp 2 (J^+)^2 \pm 2 (J^-)^2. 
\end{equation}

From spherical representations and spinor representation it is
possible to build total angular momentum representations of two types:
$(\frac\ell2,\frac{\ell-1}{2})$ and
$(\frac{\ell-1}{2},\frac\ell2)$. They have the same quadratic Casimir
\begin{equation}
J^2=\ell(\ell+1)-\frac12 
\end{equation}
but can be distinguished by $\omega$:
\begin{equation}
\omega = \mp \left( \ell + \frac12 \right) \qquad \text{and} \qquad \omega = \pm \left( \ell + \frac12 \right).
\end{equation}
The inclusion \eqref{compos} may now be stated more precisely:
\begin{subequations}
\begin{align}
\cH_\omega \cong \left( \frac{\ell}{2}, \frac{\ell-1}{2} \right) \subset&\Big(\frac{\ell-1}2,\frac{\ell-1}2\Big){\otimes}\Big(\frac12,0\Big)\oplus 
               \Big(\frac\ell2,\frac\ell2\Big){\otimes}\Big(0,\frac12\Big), \qquad \text{for } \pm \omega <0  \\
\cH_\omega \cong \left( \frac{\ell-1}{2},\frac\ell2 \right) \subset  &\Big(\frac{\ell-1}2,\frac{\ell-1}2\Big){\otimes}\Big(0,\frac12\Big)\oplus 
                 \Big(\frac\ell2,\frac\ell2\Big){\otimes}\Big(\frac12,0\Big), \qquad \text{for } \pm \omega >0. 
\end{align}
\end{subequations}
As in dimension $2$, the relation between the total angular momentum representation and $\omega$, taking valued in $\{ \pm \frac32, \pm \frac52, \dots \}$, is one-to-one after fixing Clifford matrices.

For general dimensions we label irreducible
representations as in \cite[Section 19]{FH}. 

\noindent{ $\bf d=2n+1, \, n \geq 2$.}
 Irreducible representations
are in $1-1$ correspondence with labels $(a_1,\dots,a_n)\in\mathbb{N}^n$. Spherical harmonics
of degree $\ell$
have type $(\ell,0,\dots)$, while spinor representations have type $(0,\dots,1)$.
Their tensor product decomposes as
\begin{align}
  (\ell,\dots,0)\otimes(0,\dots,1)&=  (\ell,\dots,1)\oplus
                                    (\ell-1,\dots,1), \qquad \ell \geq 1.
\end{align}
Thus the only possible types of $\cH_{\omega}$ are $( \ell-1, \dots , 1)$. This representation occurs as a~subrepresentation only in two tensor products:
\begin{equation}  (\ell,\dots,0){\otimes}(0,\dots,1), \qquad 
 (\ell-1,\dots,0){\otimes}(0,\dots,1).
 \end{equation}
We have $|\omega|=\ell+n -1$, thus $\omega$ takes values
$\{ \pm n ,\pm( n+1), \dots \}$, with opposite $\omega$ corresponding to the same total angular momentum. In particular it is not possible to express $\kappa$ as a polynomial in $J_{ij}$.

\noindent{ $\bf d=2n, \, n \geq 3$.} 
Types of irreducible representations
are parametrized by $(a_1,\dots,a_n)\in\mathbb{N}^n$. $\ell$th degree spherical harmonics are of type
$(\ell,0,\dots)$.  Spinor representations are of two types:
$(0,\dots,1,0)$.
and $(0,\dots,0,1)$.
We have  tensor products decompositions ($\ell \geq 1$):
\begin{subequations}
\begin{align}
  (\ell,\dots,0)\otimes(0,\dots,1,0)&=  (\ell,\dots,1,0)\oplus
                                      (\ell-1,\dots,0,1),\\
    (\ell,\dots,0)\otimes(0,\dots,0,1)&=  (\ell,\dots,0,1)\oplus 
                                    (\ell-1,\dots,1,0).
\end{align}
\end{subequations}
It follows that $\cH_{\omega}$ must be of the type $(\ell-1,
\dots,0,1)$ or $(\ell-1,\dots,1,0)$. These two
  representations have the same quadratic Casimir, however they are
  exchanged by the parity operator \eqref{parity}. Hence they can be
  distinguished by the sign of the following Casimir element, defined
  as the $n$th wedge power of the 2-form $J$:
\begin{equation}\label{propor}
 \bigwedge_{j=1}^n J := \frac{1}{2^n} \epsilon^{i_1 \dots i_{2n}} J_{i_1 i_2} \cdots J_{i_{2n-1} i_{2n}}.
\end{equation}
Here $\epsilon$ is the Levi-Civita symbol.

We will show that \eqref{propor} is actually proportional to $\kappa$.
 Using the fact that skew-symmetrization of the product of two or more $L_{ij}$ vanishes and Clifford relations we derive
\begin{equation}
\bigwedge_{j=1}^n J = \frac{n (2n-2)!}{2^{2n-2}} \left( (-\i)^n \alpha_1 \cdots \alpha_{2n} \right) \left( L \sigma + \frac{2n-1}{2} \right).
\end{equation}
A Clifford representation is determined up to isomorphism by
specifying the sign in the  relation $\beta = \pm (-\i)^n \alpha_1 \cdots \alpha_{2n}$. Then we have
\begin{equation}
\bigwedge_{j=1}^n J = \pm \frac{n(2n-2)!}{2^{2n-2}} \kappa.
\end{equation}
As in lower even dimensions, for fixed Clifford matrices angular
momentum types are in one-to-one correspondence with the values $\omega \in \{ \pm \left( n - \frac12 \right), \pm \left( n + \frac12 \right), \dots \}$.

\subsection{Dirac operators on manifolds}

The operator $\kappa$, which is central to the separation of variables of the radially symmetric Dirac equation, is closely related to the Dirac equation on the sphere.
We would like to give a~short discussion of this topic.

Before we discuss the case of a sphere, in this subsection we
give a short introduction to Dirac operators on Riemannian manifolds. We take Clifford module bundles as central objects. A popular alternative is based on the concept of a {\em spin structure}. Spinor bundles are then constructed by the associated bundle construction, see \cite[p. 7--44, 77--135]{Lawson} for an exposition. A comparison between the two approaches is presented in \cite{Trautman2}.

Given a  Euclidean vector space $E$ with the scalar
  product of $u,v\in E$ denoted $u\cdot v$, we let $\Cl(E)$ be the
corresponding {\em Clifford algebra}, that is the quotient of the
tensor algebra of $E$ by the ideal generated by elements of the form
$u \otimes u - u \cdot u$. Then $\R$ and $E$ are naturally embedded in
$\Cl(E)$ (in concrete matrix realizations of $\Cl(E)$ the latter
embedding is realized by contraction of vectors with $\alpha$ matrices
such as \eqref{alpha}). In this subsection we identify elements of $E$
with their images in $\Cl(E)$.

The automorphism $\alpha$ of $\Cl(E)$ characterized by the equation $\alpha(u) = -u$ for $u \in E$ is called the {\em main automorphism} or the {\em parity}. Elements of $\Cl(E)$ fixed (negated) by $\alpha$ are said to be even (odd). The {\em transposition} is the anti-automorphism of $\Cl(E)$ characterized by $(u_1 \dots u_n)^\T = u_n \dots u_1$ for $u_1, \dots, u_n \in E$.

The {\em spin group} $\Spin(E)$ is the group of even invertible elements $g \in \Cl(E)$ such that 
\begin{equation}
gug^{-1} \in E \text{ for every } u \in E, \qquad g^\T g =1.
\end{equation}
If $g \in \Spin(E)$, then the endomorphism $u \mapsto gug^{-1}$ of $E$ belongs to the special orthogonal group $\SO(E)$. Thus we have a homomorphism $\Spin(E) \to \SO(E)$. This homomorphism is surjective with kernel $\{ \pm 1 \}$. Since this is a central subgroup of $\Cl(E)$, the adjoint action of $\Spin(E)$ on $\Cl(E)$ descends to an action of $\SO(E)$ on $\Cl(E)$. The Lie algebra $\mathfrak{spin}(E)$ of $\Spin(E)$ is the subspace of $\Cl(E)$ spanned by elements of the form $[u_1, u_2]$ with $u_1, u_2 \in E$. We have an isomorphism $\mathfrak{spin}(E) \cong \mathfrak{so}(E)$, which takes $[u_1,u_2]$ to the endomorphism
\begin{equation}
 u_3 \mapsto [[u_1,u_2],u_3] = 4 u_1 (u_2 \cdot u_3) - 4 u_2 (u_1 \cdot u_3).
\label{eq:spin_rep}
\end{equation}
Therefore, $A\in \mathfrak{so}(E)$ is mapped to
\begin{equation}\frac{1}{8}\sum_{ij}[e_i,e_j](e_i\cdot
  Ae_j),\label{cliffo2}
\end{equation}
where $e_i$ form an orthonormal basis of $E$.

Every $\Cl(E)$-module $\cV$ is a direct sum of irreducible
modules. Let $\cV$ be an irreducible complex representation. The even
subalgebra of $\Cl(E)$ (and in particular the spin group $\Spin(E)$)
is represented faithfully on $\cV$. There exists a positive-definite
hermitian form $( \cdot | \cdot )$ on~$M$, called a {\em spinor scalar
  product}, such that $( \psi_1 | c \psi_2 ) =( c^\T \psi_1 | \psi_2
)$ for $c \in \Cl(E)$ and $\psi_1, \psi_2 \in \cV$. It is unique up to
positive scalars. Furthermore, there exists an antilinear operator
$\Theta$ on $\cV$, called a {\em spinor conjugation}, such that
\begin{equation}
\Theta c \Theta^{-1} = \begin{cases} c & \text{if } n \not \equiv 3 \text{ mod } 4, \\ \alpha(c) & \text{if } n \equiv 3 \text{ mod } 4, \end{cases} \qquad \Theta^2 = \begin{cases} 1 & \text{if } n \in \{ 0 , 1 , 2 , 7 \} \text{ mod } 8, \\ -1 & \text{if } n \in \{ 3,4,5,6 \} \text{ mod } 8. \end{cases}
\end{equation}
$\Theta$ is unique up to a phase factor.

Now let $M$ be a Riemannian manifold with tangent bundle $TM$ and the Levi-Civita connection $\nabla$. For every $x \in M$ consider the Clifford algebra $\Cl(T_x M)$. Together these Clifford algebras form a bundle $\Cl(TM)$ of Clifford algebras over $M$. If $M$ is oriented, we can locally choose positively oriented orthonormal framings $\{ e_i \}_{i=1}^d$ and put
\begin{equation}
\vol_M = e_1 \cdots e_d.
\end{equation}
The right hand side does not depend on the choice of framing, hence it defines a global section of $\Cl(TM)$.

The Levi-Civita connection extends uniquely to a connection on $\Cl(TM)$ satisfying the Leibniz rule:
\begin{equation}
\nabla (c_1 c_2) = (\nabla c_1 ) c_2 + c_1  \nabla c_2
\label{eq:Leibniz1}
\end{equation}
for sections $c_1, c_2$ of $\Cl(TM)$. This connection commutes with the main automorphism and the transposition. If defined,  $\vol_M$ is covariantly constant. 

A vector bundle $\Sigma$ whose fiber $\Sigma_x$ is a representation of $\Cl(T_x M)$ (with the module structure smoothly varying with $x$) is called a {\em Clifford module bundle}. A connection $\nabla$ on $\Sigma$ will be called {\em Clifford covariant} if it satisfies
\begin{equation}\label{cliffo}
\nabla (c \psi) = (\nabla c ) \psi + c  \nabla \psi.
\end{equation}
If in addition for every $x \in M$ and every null-homotopic loop $\gamma$ based at $x$ the holonomy endomorphism $\mathrm{hol}_{\Sigma, \gamma} \in \mathrm{GL}(\Sigma_x)$ is an element of $\Spin(T_x M)$, we call $\nabla$ a {\em locally spin connection}. If this is true for all loops, we say that $\nabla$ is a {\em spin connection}. A Clifford module bundle equipped with a~spin connection will be called a {\em spinor bundle}.

\begin{lemma}  \label{holo_lift}
If $\nabla$ is a spin connection, then
the holonomy endomorphism $\mathrm{hol}_{\Sigma, \gamma} \in \Spin(T_xM)$ lifts the holonomy $\mathrm{hol}_{TM, \gamma} \in \SO(T_xM)$ of the Levi-Civita connection. 
\end{lemma}
\proof 
  By the Clifford covariance \eqref{cliffo}, for any
$c\in\Cl(T_xM)$ we have
\begin{equation}
\hol_{\Cl(T M), \gamma}(c) \hol_{\Sigma, \gamma} \psi = \hol_{\Sigma, \gamma} (c \psi) = \hol_{\Sigma, \gamma} c  \, \hol_{\Sigma, \gamma} ^{-1} \hol_{\Sigma, \gamma} \psi.
\end{equation}
Thus
\begin{equation}\label{cliffo1}
  \hol_{\Cl(TM), \gamma}(c) 
=  \hol_{\Sigma, \gamma} c \, \hol_{\Sigma, \gamma}^{-1} . 
\end{equation}
As $c$ we can choose $u\in T_xM\subset\Cl(T_xM)$ and rewrite
\eqref{cliffo1} as
\[
  \hol_{TM, \gamma}u 
=  \hol_{\Sigma, \gamma} u \, \hol_{\Sigma, \gamma}^{-1} . \qedhere
\]


From now on we assume that $M$ is orientable. As a 
  consequence, the holonomies of the Levi-Civita connection are 
  always contained in $SO(T_xM)$. (On non-orientable manifolds  they 
  may be contained in $O(T_xM)$)

The following lemma
allows us to conveniently  check whether a given connection is spin \cite{Trautman2}.

\begin{lemma} \label{spin_con_criterion}
Let $\Sigma$ be a bundle of irreducible Clifford modules with a Clifford covariant connection $\nabla$. Then $\nabla$ is a spin connection if and only if there exist a spinor scalar product $(\cdot | \cdot )$ and a spinor conjugation $\Theta$ on $\Sigma$ such that
\begin{equation}
\nabla \Theta \psi = \Theta \nabla \psi, \qquad \D (\psi_1 | \psi_2 ) = (\nabla \psi_1 | \psi_2) + (\psi_1 | \nabla \psi_2). \label{eq:theta_product_constant}
\end{equation}
\end{lemma}
\begin{proof}
$\Rightarrow$. We focus on one connected component $M_0$ of $M$ and choose a point $x$ therein. Then we choose a spinor scalar product $(\cdot | \cdot)$ and a spinor conjugation $\Theta$ in $\Sigma_x$. By assumption, they are invariant under $\hol_{\Sigma, \gamma}$ for every loop based at $x$. Now parallel transport $( \cdot | \cdot )$ and $\Theta$ to all other fibers over $M_0$. Invariance under holonomies implies that the result is independent of the choice of paths, smooth and covariantly constant, hence satisfies \eqref{eq:theta_product_constant}.

$\Leftarrow$. Let $\gamma$ be a loop based at $x$ and let $c \in \Cl(T_xM)$. Let $g \in \Spin(T_xM)$ be a lift of $\hol_{T M , \gamma} \in \SO(T_xM)$. Arguing as in the proof of Lemma \ref{holo_lift} we see that $\hol_{\Sigma, \gamma} c \, \hol_{\Sigma, \gamma}^{-1} = g c g^{-1}$. By irreducibility of $\Sigma_x$, this implies that $\hol_{\Sigma, \gamma} = z g$ for some $z \in \C$. Since both $\hol_{\Sigma, \gamma}$ and $g$ preserve the scalar product, $|z|=1$. Since both commute with $\Theta$, $z \in \R$. Thus $\hol_{\Sigma, \gamma}$ coincides with $g$ or $-g$ and hence belongs to $\Spin(T_xM)$.  
\end{proof}

\begin{lemma} \label{spinor_bundles_split}
Every spinor bundle is a direct sum of spinor bundles whose fibers are irreducible Clifford modules.
\end{lemma}
\begin{proof}
Analogous to the proof of $\Rightarrow$ in Lemma \ref{spin_con_criterion}.
\end{proof}

Recall that for a vector bundle $\Sigma$ with a connection $\nabla$, the expression
\begin{equation}
\Omega(\mathbf U,\mathbf V):=\nabla_{\mathbf U} \nabla_{\mathbf V} - \nabla_{\mathbf V} \nabla_{\mathbf U} - \nabla_{[\mathbf U, \mathbf V]_{\mathrm{Lie}}}.
\label{eq:con_curv}
\end{equation}
defines an $\mathrm{End}(\Sigma)$-valued $2$-form, called the {\em
  curvature} of $\nabla$. Here $[ \cdot , \cdot ]_{\mathrm{Lie}}$ is
the Lie bracket of vector fields $\mathbf U, \mathbf V$. If $\Sigma =
TM$ and $\nabla$ is the Levi-Civita connection, then $\Omega$ is
denoted by $R$ and called the {\em Riemann tensor}. One checks that $\left. R(\mathbf U, \mathbf V) \right|_x$ is an element of $\mathfrak{so}(T_xM)$. 

\begin{lemma}
If $\nabla$ is a spin connection, then its curvature takes the form
\begin{equation}
\Omega(\mathbf U, \mathbf V) = \frac{1}{8} \sum_{i,j} (e_i \cdot R(\mathbf U, \mathbf V) e_j) [e_i, e_j].
\label{eq:spin_con_curv}
\end{equation}
A partial converse holds: every Clifford covariant connection with curvature given by the formula above is a locally spin connection.
\end{lemma}
\proof The curvature may be extracted from holonomies along
infitesimal parallelograms. Therefore by Lemma \ref{holo_lift}, the
curvature of $\nabla$ at $x$ is an element of $\mathfrak{spin}(T_xM)$,
coinciding with $R(\mathbf U, \mathbf V)$ taken in the
representation \eqref{cliffo2}.

Now we prove the converse. If $\gamma$ is any path from $y$ to $x$ and $\hol_{\Sigma, \gamma} \in \mathrm{Hom}(\Sigma_x, \Sigma_y)$ is the corresponding parallel transport, then by the Clifford covariance
\begin{equation}
\hol_{\Sigma, \gamma} \Omega( \mathbf U,\mathbf V) \hol_{\Sigma, \gamma}^{-1} = \frac{1}{8} \sum_{i,j} (e_i \cdot R(\mathbf U, \mathbf V) e_j) [ \hol_{\T M,\gamma}(e_i), \hol_{\T M,\gamma}(e_j)] \in \mathfrak{spin}(T_xM).
\end{equation}
Let $\gamma_s$ be a family of loops $[0,1] \to M$ based at $x$. For $t \in [0,1]$ let $\gamma_s^t := \left. \gamma_s \right|_{[0,t]}$. Then
\begin{equation}
\hol_{\Sigma, \gamma_s}^{-1} \frac{\D}{\D s} \hol_{\Sigma, \gamma_s} = \int_0^1 \hol_{\Sigma, \gamma_s^t}^{-1} \Omega \left( \frac{\partial \gamma_s(t)}{\partial s} , \frac{\partial \gamma_s(t)}{\partial t} \right) \hol_{\Sigma, \gamma_s^t} \D t.
\end{equation}
It follows that for a null-homotopic loop $\gamma$ based at $x$ we have that $\hol_{\Sigma, \gamma} \in \Spin(T_x M)$. 
\qed

Next we define the Dirac operator on sections of a spinor bundle $\Sigma$. Let us choose a locally defined orthonormal framing $\{ e_i \}$ of $TM$. Now put
\begin{equation}
 D \psi = -\i \sum_i e_i \cdot \nabla_{e_i} \psi.
\label{eq:Dirac_def}
\end{equation}
Here the multiplication by $e_i$ is the Clifford multiplication ($e_i$ being regarded as a section of $\Cl(TM)$). It is not dificult to check that $ D \psi$ does not depend on the choice of framing, so~local expressions on the right hand side of \eqref{eq:Dirac_def} can be glued to obtain a globally defined differential operator.

If $\Sigma$ is a spinor bundle over an  oriented Riemannian manifold $M$, there exist two
distinguished second order differential operators acting on sections of $\Sigma$: the square of the Dirac operator $ D^2$ and the Bochner Laplacian. 
To describe the latter, let $(\cdot|\cdot)$ be a spinor scalar product. It~yields a scalar product on $T^* M \otimes \Sigma$. Now the Bochner Laplacian,
at least formally, is (minus) the operator associated to the quadratic form
\begin{equation}
-(\psi, \Delta \psi) := \int_M \big(\nabla \psi(x)|\nabla \psi(x)\big) \D x.\label{scalar}
\end{equation}
Equivalently, the Bochner Laplacian can be defined without invoking the scalar product by
\begin{equation}
\Delta:= \sum_i \left(\nabla_{e_i} \nabla_{e_i}-\nabla_{\nabla_{e_i} e_i} \right),
\end{equation}

Note that the Bochner Laplacian uses  the {\em covariant Hessian} 
\begin{equation}
\mathrm{Hess}(\mathbf U ,\mathbf V) = \nabla_{\mathbf U} \nabla_{\mathbf V} - \nabla_{\nabla_{\mathbf U} \mathbf V},
\end{equation}
which is bilinear over
$C^\infty(M)$
and satisfies
\begin{equation}
  \mathrm{Hess}(\mathbf U ,\mathbf V)-
    \mathrm{Hess}(\mathbf V ,\mathbf U)=\Omega 
(\mathbf U ,\mathbf V).\label{torsion}\end{equation}
\eqref{torsion} follows from the torsion-freeness of the Levi-Civita connection, that is
\begin{equation}
  \nabla_\mathbf{U}\mathbf{V}-\nabla_\mathbf{V}\mathbf{U}-[\mathbf{U},\mathbf{V}]_{\mathrm{Lie}}=0.\end{equation}

In the following proposition
we recall the celebrated {\em Lichnerowicz formula}:

\begin{proposition} The square of the Dirac operator and the Bochner Laplacian are related by\begin{equation}\label{lichnerowicz}
D^2 = -\Delta + \frac{1}{4} \mathrm{Sc},
\end{equation}where $\mathrm{Sc}$ is the scalar curvature.
\end{proposition}

\proof
Let $\mathrm{Hess}^{\mathrm s}$ be the symmetric part of the Hessian.
We choose an orthonormal framing $\{e_i\}$. Then
\begin{align}
(\i D)^2 & = \sum_{i,j} e_i \nabla_{e_i} e_j \nabla_{e_j} = \sum_{i,j} \left( e_i e_j \nabla_{e_i} \nabla_{e_j} + e_i (\nabla_{e_i} e_j) \nabla_{e_j} \right) \\
& = \sum_{i,j} e_i e_j \left( \mathrm{Hess}(e_i, e_j) + \nabla_{\nabla_{e_i} e_j} \right) + \sum_{i,j} e_i (\nabla_{e_i} e_j) \nabla_{e_j} \notag \\
& = \sum_{i,j} e_i e_j \left( \mathrm{Hess}^{\mathrm s}(e_i, e_j) + \frac{1}{2} \Omega(e_i, e_j) \right) + \sum_{i,j} \left( e_i e_j \nabla_{\nabla_{e_i} e_j} + e_i (\nabla_{e_i} e_j) \nabla_{e_j} \right) \notag \\
& = \Delta + \frac{1}{32} \sum_{i,j,n,m} [e_i,e_j][e_n,e_m] \left( e_n \cdot R(e_i, e_j) e_m \right). \notag
\end{align}
Below we will show that the last two terms in the third line cancel. The last term in the fourth line may be shown to be equal to $- \frac{1}{4} \mathrm{Sc} = - \frac14 \sum_{ij}e_i\cdot R(e_i,e_j)e_j$ using Clifford relations and symmetries of the Riemann tensor. 

Connection coefficients are defined by the formula
\begin{equation}
\nabla_{e_i}e_j = \sum_k c_{ijk} e_k.
\end{equation}
$\nabla_{e_i}(e_j \cdot e_k)=0$ and metric compatibility of the connection give $c_{ijk} + c_{ikj}=0$. We have
\begin{equation}\label{torsion1}
 \sum_{j} \left( e_j \nabla_{\nabla_{e_i} e_j} +(\nabla_{e_i} e_j) \nabla_{e_j} \right) 
= \sum_{j,k} c_{ijk}\left(  e_j \nabla_{e_k} +  e_k \nabla_{e_j} \right).
\end{equation}
Now switch the roles of $j,k$ in the second term to see that
\eqref{torsion1} vanishes.
\qed

Now suppose that $N$ is an orientable submanifold of $M$ of codimension $1$.
Then there exists a smooth field of unit normal vectors $\nu$. We have the following relation between the Levi-Civita connection on $M$ and on $N$:
\begin{align}\label{levici}
\nabla^N_{\mathbf U} \mathbf V &= \nabla^M_{\mathbf U} \mathbf V - (\nabla^M_{\mathbf U} \mathbf V \cdot \nu) \nu =\nabla^M_{\mathbf U} \mathbf V + (\mathbf V \cdot \nabla^M_{\mathbf U} \nu) \nu,
\end{align}
where $\mathbf U, \mathbf V$ are tangent to $N$. That is, $\nabla^N_{\mathbf U} \mathbf V$ is the projection of $\nabla^M_{\mathbf U} \mathbf V$ onto $TN$.

\eqref{levici} can be rewritten as follows:
\begin{align}
\nabla^N_{\mathbf U} \mathbf V &= \nabla^M_{\mathbf U} \mathbf V  + \frac12 \nu \left( (\nabla^M_{\mathbf U} \nu) \mathbf V + \mathbf V (\nabla^M_{\mathbf U} \nu)  \right) = \nabla^M_{\mathbf U} \mathbf V + \frac12 [\nu \nabla^M_{\mathbf U} \nu, \mathbf V ].   \notag
\end{align}
where now $\mathbf V,$ $\nu$ and $\nabla_{\mathbf U}^M\nu$ are treated as sections  of the Clifford bundle $\Cl(TM)$.
This is immediately generalized to general Clifford fields
\begin{equation}\label{cliffi}
\nabla^N_{\mathbf U} = \nabla^M_{\mathbf U} + \frac12 [\nu \nabla^M_{\mathbf U} \nu, \cdot]
\end{equation}

Now  assume that $\Sigma^M$ is a Clifford module bundle over $M$ with a Clifford covariant connection $\nabla^M$. The restriction of $\Sigma^M$ to $N$, denoted $\Sigma^N$, is a bundle of Clifford modules. \eqref{cliffi} motivates defining the following connection on $\Sigma^N$:
\begin{equation}
\nabla^N_{\mathbf U} = \nabla^M_{\mathbf U} + \frac12 \nu \nabla^M_{\mathbf U} \nu.
\label{eq:nablaN_spin}
\end{equation}

By construction, $\nabla^N$ is Clifford covariant.

\begin{lemma} \label{spin_con_res}
If $\Sigma^M$ is a spinor bundle, so is $\Sigma^N$.
\end{lemma}
\begin{proof}
By Lemma \ref{spinor_bundles_split} we may assume that $\Sigma^M$ is a bundle of irreducible Clifford modules. If~$d := \dim(M)$ is even, then $\Sigma^N$ splits into eigenbundles of $\vol_N$, which are bundles of irreducible Clifford modules. If $d$ is odd, $\Sigma^N$ is irreducible.

Now choose a spinor scalar product $( \cdot | \cdot )$
  and a spinor conjugation $\Theta^M$ on $\Sigma^M$. Let $ \Theta^N =
  \Theta^M$ if $d \in \{ 1 , 2 \}$ mod $4$ and $ \Theta^N = \nu
  \Theta^M$ if $d \in \{ 0, 3 \}$ mod $4$,
in both cases restricted to $\Sigma^N$. The restriction of $( \cdot | \cdot )$ to
  $\Sigma^N$ and $\Theta^N$ are a spinor scalar
  product and a spinor conjugation satisfying
  \eqref{eq:theta_product_constant}.
If $d$ is even, this is still true if we further restrict to eigenbundles of $\vol_N$.
  The result follows from Lemma \ref{spin_con_criterion}.
\end{proof}

Assume now that we have a covariantly constant section $\beta$ of
$\mathrm{End}(\Sigma )$ satisfying $\beta^2=1$ and anticommuting with
$T M\subset\Cl(TM)$. Let us consider the operator $\Gamma = -\i \beta
\nu$ acting on sections of $\Sigma^N$. It satisfies $\Gamma^2=1$ and
commutes with all sections of $\Cl(T N)$. Hence its eigenbundles
$\Sigma_\pm^N$ for eigenvalues $\pm 1$ are also Clifford module
bundles over $N$.  Using \eqref{eq:nablaN_spin} one checks that $\Gamma$ commutes also with the covariant differentiation, so $\Sigma_\pm^N$ inherit the spin connection.

Operator $\vol_M$ commutes with covariant differentiation and
anticommutes with $\Gamma$, hence it takes sections of $\Sigma_\pm^N$
to sections of $\Sigma_\mp^N$. If $d$ is odd, $\vol_M$ commutes with
Clifford fields and hence defines an isomorphism of spinor bundles
$\Sigma_+^N \cong \Sigma_-^N$. If $d$ is even, $\Sigma_+^N$ and
$\Sigma_-^N$ are non-isomorphic as Clifford module bundles. 
In this case, we can take $\beta:=\pm\i^{\frac{d}2}\mathrm{vol}_M$
and $\Gamma$ coincides up to phase with the $\Cl(T N)$ section
$\vol_N$.

 If $\Sigma$ is irreducible for $\Cl(TM)$ and $\beta$, by dimensional consideration, $\Sigma_\pm^N$ are  bundles of {\em irreducible} Clifford modules.

\subsection{Dirac operators on spheres} \label{app:spheres}

Now let us consider the sphere $\S^{d-1}$ of radius $1$ (thus we put
$|x|=1$ below).
We will apply the formalism of the previous section with $M:=\R^d$ and
$N:=\S^{d-1}$.
For brevity, we will write $\S$ for $\S^{d-1}$. 
We will use the notation of Subsections
\ref{app-Dir1}
and \ref{app-Dir2}, such as $S$, $R$, $T$ and $\kappa$.

The normal vector $\nu$ is identified with $S$.
The Levi-Civita connection on $\S$ is
\begin{equation}
\nabla_{\mathbf U}^\S \mathbf V  = \partial_{\mathbf U} \mathbf V + (\mathbf U \cdot \mathbf V)x \label{clef}
\end{equation}
for vector fields $\mathbf U, \mathbf V$ tangent to the sphere. Here $\partial_{\mathbf U} = \sum \limits_{i=1}^d \mathbf U_i \partial_i$. Consider the vector space $\cK$ from previous subsections. $\S \times \cK$ is a Clifford module bundle with connection
\begin{equation}
\nabla_{\mathbf U}^\S = \partial_{\mathbf U} + \frac12 S \mathbf U .
\label{eq:spinor_connection}
\end{equation}

\begin{proposition}
The connection \eqref{eq:spinor_connection} is a spin connection with curvature
\begin{equation}
\Omega(\mathbf U,\mathbf V) = \frac14 [\mathbf U,\mathbf V].
\label{eq:curvature}
\end{equation}
If $d=2$, the holonomy of $\nabla$ along $\S^1$ is equal to $-1$.
\end{proposition}
\begin{proof}
All but the last statement follow from Lemma \ref{spin_con_res}. \eqref{eq:curvature} may also be obtained from a~simple direct computation. Now let $d =2$. We parametrize $\S^1$ as $x = ( \cos(\alpha), \sin(\alpha))$. Then \eqref{eq:spinor_connection} takes the form
\begin{equation}
\nabla_{\frac{\partial}{\partial \alpha}} \psi = \frac{\partial \psi}{\partial \alpha} + \frac{1}{2} \vol_{\R^2} \psi. 
\end{equation}
It follows that solutions of the parallel transport equation $\nabla_{\frac{\partial}{\partial \alpha}} \psi =0$ satisfy $\psi(2 \pi) = - \psi(0)$.
\end{proof}

Eigenbundles $\cK_\pm \subset \S \times \cK$ of $\Gamma = -\i \beta S$ to eigenvalues $\pm 1$ are irreducible spinor bundles, isomorphic if $d$ is odd and non-isomorphic otherwise.

Choose a local orthonormal framing $\{ e_i \}_{i=1}^{d-1}$ of $T
\S$. Denote the Dirac operator on $\S$ by
  $D_{\S}$
and on $\R^d$ by $D$. Using \eqref{eq:spinor_connection} we manipulate its definition  to the form
\begin{align}
D_{\S} &= -\i \sum_i e_i \nabla_{e_i} \psi = -\i \sum_i e_i \partial_{e_i} \psi - \frac{ \i }{2} \sum_i e_i S e_i \psi \\
&= D + \i S \partial_\nu \psi +\i \frac{d-1}{2} S \psi = D -  S R.  \notag 
\end{align}
Next let $\{ e_i \}_{i=1}^d$ be the canonical basis of $\R^d$. Multiplying the above by $S$ from the left we find
\begin{align}
S D_{\S} &= SD - R = -\i \sum_{i,j} x_i e_i e_j \partial_j - SR \\
& = -\frac{\i}{2} \sum_{i,j} \left( [e_i, e_j] + [e_i, e_j]_+ x_i \partial_j \right) - R = \i T. \notag
\end{align}
Hence we have
\begin{equation}
D_{\S} = \i  ST. 
\end{equation}
Let us note that this is exactly the off-diagonal term of $D$ with respect to decomposition into eigenspaces of $S$. Next observe that
\begin{equation}
\kappa = \Gamma D_{\S} = D_{\S} \Gamma.
\end{equation}
It follows that on sections of $\cK_\pm$, operator $\kappa$ acts as $\pm D_{\S}$.

We also remark that if $\psi$ is a $\cK$-valued polynomial homogeneous of degree
$\ell$ annihillated by $D$, then $(S \mp \i)\psi$ is an eigenvector of
$D_{\S}$ to eigenvalue $\pm (\ell + \frac{d-1}{2})$. Indeed,
\begin{equation}
D_{\S} \psi = - SR \psi = \i \left( \ell + \frac{d-1}{2} \right) S \psi,
\end{equation}
which implies
\begin{align}
D_{\S}(S\mp\i)\psi &= (-S \mp \i) D_{\S} \psi \\
&= \i \left( \ell + \frac{d-1}{2} \right) (- S \mp \i) S \psi = \pm \left( \ell + \frac{d-1}{2} \right) (S \mp \i ) \psi. \notag
\end{align}
By the relation between $D_{\S}$ and $\kappa$, this calculation reproduces the spectrum of $\kappa$ found in Proposition \ref{eq:Hom_inc}.

We claim that a complete set of eigenfunctions of $ D_\S$ is
obtained by the construction above, similarly as spherical harmonics
are obtained by restricting scalar-valued homogeneous harmonic
polynomials, e.g. \cite[p. 73--81]{Axler}. This may be seen from the
Stone-Weierstrass theorem and the following lemma. Besides this
application, the lemma elucidates the decomposition of spaces of
spinor-valued polynomials into irreducible representations of
$\Spin(\R^d)$ and relates eigenvectors of $D_{\mathbb S}$ (and hence also of $\kappa$) to harmonic polynomials.


Consistently with our notation, in the following lemma $x$ denotes the element of the Clifford algebra $x=\sum_i x_i e_i$, whereas $x_i$ are real numbers. $x^j$ is the $j$-th power of $x$.

\begin{lemma}
Let $\cK_\ell$ be the space of $\cK$-valued polynomials homogeneous of  degree $\ell$ and let $\cK_{\ell}^0$ be the kernel of $D$ acting in $\cK_\ell$. Then
\begin{equation}
\cK_{\ell} = \cK_{\ell}^0 \oplus x \cdot \cK_{\ell-1}.
\label{eq:cKell_decomp}
\end{equation}
In particular we have a vector space decomposition
\begin{equation}
\cK_{\ell} = \bigoplus_{j=0}^\ell x^j \cK_{\ell-j}^0.
\label{eq:cKell_decomp2}
\end{equation}
Moreover, $\dim(\cK_\ell^0) = \binom{d+\ell-2}{\ell} \dim(\cK)$, and (clearly)
 $\dim(\cK_\ell) = \binom{d+\ell-1}{\ell} \dim(\cK)$.

Let $\cH_\ell$ be the space of scalar-valued harmonic polynomials homogeneous of degree $\ell$. Then
\begin{equation}
D : \cH_\ell \otimes \cK \to \cK_{\ell -1}^0
\end{equation}
is a surjection with kernel $\cK^0_\ell$. In particular there is an exact sequence
\begin{equation}
\dots \rightarrow \cH_{\ell+1} \otimes \cK \xrightarrow{D} \cH_{\ell} \otimes \cK \xrightarrow{D} \cH_{\ell-1} \otimes \cK \rightarrow \dots
\end{equation}

\end{lemma}
\begin{proof}
Let $\{ e_i \}_{i=1}^d$ be an orthonormal basis of $\R^d$. A general element of $\cK_\ell$ has the form
\begin{equation}
\psi = \sum_{i_1, \dots , i_\ell=1}^d x_{i_1} \cdots x_{i_\ell} \psi_{i_1 \dots i_\ell}
\end{equation}
with coefficients $\psi_{i_1 \dots i_\ell}$ fully symmetric. Acting with $D$ we find that $D \psi =0$ if and only if
\begin{equation}
\psi_{1 i_2 \dots i_\ell} = \sum_{j \neq 1} e_1 e_j \psi_{j i_2 \dots i_\ell}.
\end{equation}
It is easy to see that this system of equation may be uniquely solved once $\psi_{i_1 \dots i_\ell}$ is fixed for all indices $i_1, \dots , i_\ell$ different than $1$. The formula for $\dim(\cK^0_\ell)$ follows.

Using the above result it is easy to check that
\begin{equation}
\dim(\cK_\ell) = \dim(\cK_\ell^0) + \dim(x \cdot \cK_{\ell-1}).
\end{equation}
Hence \eqref{eq:cKell_decomp} will follow once we establish that $\cK_\ell^0 \cap x \cdot \cK_{\ell-1} = \{ 0 \}$.

We proceed by induction. There is nothing to prove for $\ell =0$. Suppose that \eqref{eq:cKell_decomp} holds for $\ell \leq \ell'$. Then also \eqref{eq:cKell_decomp2} holds for $\ell \leq \ell'$. Let us put $\ell = \ell'+1$ and let $\psi \in \cK_\ell^0 \cap x \cdot \cK_{\ell-1}$. By the induction hypothesis we have
\begin{equation}
\psi = \sum_{j =1}^\ell x^j \psi_j
\end{equation}
with $\psi_j \in \cK_{\ell - j}$. Now let us observe that
\begin{subequations}
\begin{align}
&[\i D , x^2] = 2 x, \qquad &\text{hence } [\i D , x^{2k}] = 2k x^{2k-1}, \\
&[\i D, x]_+ = 2 \sum_i x_i \partial_i + d, \qquad &\text{hence} [\i D, x^{2k+1}]_+ = x^{2k} (2 \sum_i x_i \partial_i + d + 2k).
\end{align}
\end{subequations}
Thus since $\psi$ is annhilated by $D$ and $\sum_i x_i \partial_i \psi_j = (\ell -j) \psi_j$, we obtain
\begin{equation}
0 = \i D \psi = \sum_{j \text{ even}} j x^{j-1} \psi_j + \sum_{j \text{ odd}} (\ell + d-1) x^{j-1} \psi_j.
\end{equation}
By induction hypothesis, $x^{j-1} \psi_j$ and $x^{j'-1} \psi_{j'}$ belong to subspaces of $\cK_{\ell-1}$ with trivial intersection if $j \neq j'$. Therefore each term in the above sum has to vanish separately. Since operators $x^{j-1}$ are injective, all $\psi_j$ vanish. Thus $\psi = 0$.

For the last part, note that $\cH_\ell \otimes \cK$ is the space of
harmonic $\cK$-valued polynomials homogeneous of degree $\ell$. It is
annihilated by $D^2$, so $D$ maps it into $\cK_{\ell-1}^0$. Statement
about the kernel is obvious. Then $
\dim(\cH_\ell\otimes\cK)=\dim(\cK_\ell^0)+\dim(\cK_{\ell-1}^0)$
follows from the well-known formula
\begin{equation}
\dim(\cH_\ell) = \binom{d+\ell-1}{\ell} - \binom{d+ \ell -3}{\ell -2}.
\end{equation}
This implies the surjectivity.
\end{proof}

As for any oriented Riemannian manifold,   two natural second order operators act on sections of spinor bundles: the square of the Dirac operator $D^2_{\mathbb S}$ and the Bochner Laplacian $\Delta_{\S}$.
In the case of spheres we have  an additional natural second order operator: the square of the total angular momentum 
$J^2$.
It turns out that for $d\geq3$ the operator $J^2$ is distinct from
both $ D_{\S}^2$ and $-\Delta_{\S}$. More
precisely, we have
\begin{align}
  D^2_{\S} &= -\Delta_{\S} +
                          \frac{(d-1)(d-2)}{4},\label{lich}\\
  D^2_{\S}&= J^2 +
                          \frac{(d-1)(d-2)}{8}.\label{Om22}
\end{align} \eqref{lich} is the Lichnerowicz
formula for the spheres (indeed,
the scalar curvature of $\S$ is 
${(d-1)(d-2)}$). Moreover, $  D^2_{\S}=\kappa^2$. Hence 
\eqref{Om22} is essentially the formula \eqref{eq:Om2_J2_rel}.
Thus, as we were surprised to find out, \eqref{eq:Om2_J2_rel} is
distinct from the Lichnerowicz formula.

\section{Mellin transformation} \label{app:operators_Rplus}

For a Schwartz function $b$ on $\R$ we define its Fourier
transform by
\begin{equation} (\mathfrak F b)(k)=\int_{- \infty}^{\infty} b(t)\e^{-\i k t}\D t. \end{equation}
It is extended to the space of Schwartz distributions in the usual way. Restriction of $\frac{1}{\sqrt{2 \pi}}\mathfrak F$ to $L^2(\R)$ is a unitary operator.

An isomorphism $\iota : C_c^{\infty}(\R) \to C_c^{\infty}(\R_+)$ is
defined by $(\iota f)(x)=x^{- \frac{1}{2}} f(\ln(x))$. Dualizing,
we~extend it to an isomorphism between spaces of distributions on $\R$
and on $\R_+$. Restriction of $\iota$ to $L^2(\R)$ is a~unitary
operator onto $L^2(\R_+)$. By {\em Schwartz class functions} and {\em tempered distributions on} $\R_+$ we shall mean smooth functions (respectively distributions) on $\R_+$ which correspond through $\iota$ to Schwartz class functions (respectively tempered distributions) on $\R$.

Mellin transform is defined as the composition $\mathfrak M = \mathfrak F \iota^{-1}$. It is an isomorphism between spaces of tempered distributions on $\R_+$ and on $\R$. If $f \in C_c^{\infty}(\R_+)$, then
\begin{equation}
  (\mathfrak M f)(k) = \int_0^{\infty}  f(x) x^{- \frac{1}{2} - \i k} \D x.
\end{equation}

Recall that $A$, $J$ and $K$ are defined in Subsection 
\ref{Remarks about notation}. We note the following identities:
\begin{subequations}
\begin{gather}
\mathfrak M Jf(k) = \mathfrak M f(- k), \\
A = \mathfrak M^{-1} K \mathfrak M.
\end{gather}
\end{subequations}

The following lemma will be used in Proposition \ref{melino}.
  The Mellin transformation plays here a secondary role.

\begin{lemma} \label{Mellin_lim}
Suppose that $f_{\epsilon}$ is a family of tempered distributions on $\R_+$ with parameter $\epsilon \in ]0,1]$ satisfying the following conditions:
\begin{itemize}
\item $\mathfrak M f_{\epsilon} \in L^1_{\loc}(\R)$, 
\item there exists $g \in L^1_{\loc}(\R)$ such that $\mathfrak M f_{\epsilon} \to g$ pointwise for $\epsilon \to 0$,
\item there exist $c, N \geq 0$ independent of $\epsilon$ such that $|\mathfrak M f_{\epsilon}(k)| \leq c (1+k^2)^N$ for almost every $k$.
\end{itemize}
Then there exists a tempered distribution $f_0$ on $\R_+$ such that $f_\epsilon \to f_0$ for $\epsilon \to 0$ in the sense of tempered distributions. Moreover, $\mathfrak M f_0 = g$.
\end{lemma}

\begin{lemma} \label{fun_calc}
Let $b$ be a tempered distribution whose Mellin transform is a Borel function. Put
\begin{equation}
(B^{\pre}f)(x) = \int_0^{\infty} b(xk) f(k) \D k, \qquad f \in C_c^{\infty}(\R_+).
\end{equation}
Then $B^{\pre}f$ is in $L^2(\R_+)$ and $B^{\pre}$ is a closable operator on $L^2(\R_+)$ with closure
\begin{equation}
B = \mathfrak M b(A)  J.
\end{equation}
\end{lemma}

\section{Whittaker functions} \label{whittaker}

In this appendix we review some properties of special functions used
in this paper.
In~particular, we discuss Whittaker
functions, which play the central role in our paper.  We~follow the conventions of
\cite{DeRi18_01} and \cite{DeFaNgRi20_01}. 

\subsection{Confluent equation}

Before discussing the Whittaker equation and its solutions, let us say
a few words about the
closely related confluent equation and the hypergeometric equation.

The {\em confluent
  equation} has the form
\begin{equation}\label{confluent-equation}
\big(z\partial_z^2 + (c-z)\partial_z - a\big)v(z)=0.
\end{equation}
Let us list three of its standard solutions:
\begin{subequations}
\begin{align}
{}_1F_1(a;c;z)&\qquad \text{characterized by }\sim 1 \text{ near }0;\\
z^{1-c} {}_1F_1(a+1-c;2-c;z) &\qquad \text{characterized by }\sim 
                              z^{1-c} \text{ near }0;\\
                        z^{-a}{}_2F_0(a,a+1-c;-;-z^{-1}) &\qquad \text{characterized by }\sim 
                              z^{-a} \text{ near }+\infty.      
\end{align}
\end{subequations}
We note that the function ${}_1\mathbf F_1(a;c;z)= \frac{1}{\Gamma(c)} {}_1 F_{1}(a;c;z)$ is holomorphic in all variables. The other two solutions are defined for $z \notin ] - \infty, 0 ]$; thus ${}_2 F_0(a,b;-;z)$ is defined on $\C \backslash [0 , \infty [$.

We will also need the hypergeometric equation
\begin{equation}
  \big(z(1-z)\partial_z^2+(c-(a+b+1)z)\partial_z-ab\big)v(z)=0.
 \label{hy1}\end{equation}
  Among its 6 standard solutions, members of the famous Kummer's
  table, let us list three:
  \begin{subequations}
\begin{align}
{}_2F_1(a,b;c;z)&\qquad \text{characterized by }\sim 1 \text{ near }0;\\
z^{1-c} {}_2F_1(a+1-c,b-c+1;2-c;z) &\qquad \text{characterized by }\sim 
                              z^{1-c} \text{ near }0;\\
                        {}_2F_1(a,b;a+b+1-c;1-z) &\qquad \text{characterized by }\sim 
                              1 \text{ near }1.      \label{eq:hypergeometric_at_1}
  \end{align}
  \end{subequations}
${}_2 F_1(a,b;c;z)$ may be defined by a power series convergent for
$z$ in the unit disc. It admits analytic continuation along any path
in $\C \backslash \{ 1 \}$. To make it a single-valued function,  it is customary to restrict its domain to
$z \in \C \backslash [1, \infty[$.
   Then ${}_2\mathbf{F}_1(a,b;c;z) = \frac{1}{\Gamma(c)} {}_2 F_1(a,b;c;z)$ is holomorphic in all variables. It satisfies
\begin{equation}
{}_2\mathbf{F}_1 (a,b;0;z) = abz {}_2\mathbf{F}_1(a+1,b+1;2;z),
\label{eq:2F1c0}
\end{equation}

We have the following identities valid for $z\not\in]-\infty,0]$:
\begin{align} 
  &\frac{\sin(\pi c)}{\pi}  z^{-a}{}_2F_0(a,a+1-c;-;-z^{-1}) \label{id3} \\ \notag
  =&\frac{ {}_1\mathbf F_1(a;c;z)}{\Gamma(a+1-c)}
     -\frac{ z^{1-c}{}_1 \mathbf F_1a+1-c;2-c;z)}{\Gamma(a)},  \\[2ex]
  & \frac{\sin(\pi c)}{\pi} {}_2\mathbf F_1(a,b;a+b+1-c;1-z) \label{id4}  \\ \notag
  =&\frac{ {}_2\mathbf{F}_1(a,b;c;z)}{\Gamma(a-c+1)\Gamma(b-c+1)}
     -\frac{  z^{1-c}{}_2 \mathbf{F}_1(a-c+1,b-c+1;2-c;z)}{\Gamma(a)\Gamma(b)}.
\end{align}
\eqref{id3} can be taken as an alternative definition of ${}_2 F_0$.

\begin{lemma}
Let $\epsilon >0$, $\RE(z) < \frac12$. Then
\begin{equation}
{}_2 F _1 (a,b+ \lambda, c+ \lambda ; z) = (1-z)^{-a} + O(\lambda^{-1}), \qquad |\arg(\lambda)| \leq \pi - \epsilon, \ |\lambda| \to \infty.
\label{eq:hypergeometric_asymptotic_pre}
\end{equation}
\end{lemma}
\begin{proof}
Assumption about $\arg(\lambda)$ guarantees that for $|\lambda|$ sufficiently large $c + \lambda \notin - \mathbb N$, so the left hand side of \eqref{eq:hypergeometric_asymptotic_pre} is finite. We apply the Pfaff transformation:
\begin{equation}
{}_2 F _1 (a,b+ \lambda, c+ \lambda ; z) = (1-z)^{-a} {}_2 F_1 \left( a,c-b ; c + \lambda ;  \frac{z}{z-1} \right).
\end{equation}
The claim follows from the standard series defining ${}_2 F_1$, because $\left| \frac{z}{z-1} \right| < 1$.
\end{proof}

\begin{lemma}
The following asymptotic expansion holds for $\RE(z) < \frac12$, $z \notin ] - \infty, 0 ]$, $s \to \pm \infty$:
\begin{equation}
{}_2 \mathbf F_1(a,b-\i s;  c ; 1-z) \sim \sgn(s) \cdot \left( \frac{(1-z)^{-a} ( \i s)^{-a}}{\Gamma(c-a)} + \frac{z^{c-a-b+ \i s} (1-z)^{b-c+ \i s} (- \i s)^{a-c}}{\Gamma(a)} \right)
\label{hypergeometric_asymptotic}
\end{equation}
locally uniformly in $a,b,c,z$.
\end{lemma}
\begin{proof}
Using \eqref{id4} we get
{ \small
\begin{align}
&\frac{\sin \left( \pi(a+b-c+1 - \i s) \right)}{\pi} {}_2 \mathbf F_1(a,b-\i s, c ; 1-z) \label{eq:F_decomp} \\
= & \frac{{}_2 \mathbf F_1(a,b-\i s;a+b-c+1 - \i s; z)}{\Gamma(c-b+ \i s) \Gamma(c-a)} - z^{c-a-b+ \i s} \frac{{}_2 \mathbf F_1(c-b + \i s,c-a; c-a-b+1 + \i s ; z)}{\Gamma(a) \Gamma(b - \i s)}. \nonumber
\end{align}
}
Then \eqref{eq:hypergeometric_asymptotic_pre} gives for large $|s|$:
{ \small
\begin{align}
&\frac{\sin \left( \pi(a+b-c+1 - \i s) \right)}{\pi} {}_2 \mathbf F_1(a,b-\i s;c ; z) \\
\sim  & \frac{(1-z)^{-a}}{\Gamma(c-b+ \i s) \Gamma(c-a) \Gamma(a+b-c+1- \i s)} -  \frac{z^{c-a-b+ \i s} (1-z)^{b-c- \i s}}{\Gamma(a) \Gamma(b - \i s) \Gamma(c-a-b+1 + \i s)}. \nonumber
\end{align}
}
Using $\sin \left( \pi(a+b-c+1 - \i s) \right) \sim \frac{1}{2 \i} \e^{\pi |s| + \i  \pi \, \sgn(s) (a+b-c+1)}$ and  Stirling's formula
\begin{equation}
\Gamma(z_0 + z ) \sim \sqrt{\frac{2 \pi}{z}} \e^{-z} z^{z+z_0}, \qquad |\arg(z)| \leq \pi - \epsilon, \ |z| \to \infty
\label{eq:Stirling}
\end{equation}
yields, after algebraic manipulations, formula \eqref{hypergeometric_asymptotic}.
\end{proof}

\subsection{Hyperbolic-type Whittaker equation}

The standard form of
  the  Whittaker equation is
\begin{align}
\Big ( -\partial_{z}^2 + \Big ( m^2 - \frac14 \Big ) \frac{1}{z^2} - \frac{ \beta }{ z } +\frac14 \Big)g=0. 
\label{whi0}
\end{align} 
In this section we briefly describe solutions of the Whittaker equation, following mostly \cite{DeRi18_01,DeFaNgRi20_01}.

We will sometimes call \eqref{whi0} the {\em hyperbolic-type Whittaker 
  equation}, to distinguish it from the {\em trigonometric-type Whittaker 
  equation}, which differs by the sign in front of $\frac14$.

There are two kinds of standard solutions to the Whittaker equation.

The function $\mathcal{I}_{\beta,m}$ is defined by 
\begin{equation}
\mathcal{I}_{\beta,m}(z) = z^{\frac12 + m} {\e^{ \mp \frac{z}{2}} }{}_1 \mathbf F_1 \Big ( \frac12 + m \mp \beta ; 1 + 2m ; \pm z \Big ) .\label{eq:def_I}
\end{equation}
The standard domain of $\cI_{\beta, m}$ is $\C
\backslash ]- \infty , 0 ]$. We have
\begin{subequations}
\begin{gather}
\cI_{-\beta,m}(-z)=\e^{-\i\pi(\frac12+m) \sgn(\IM(z))}  \cI_{\beta,m}(z) \qquad \text{for } z \in \C \backslash \R, \label{eq_miracle} \\
\overline{\cI_{\overline \beta, \overline{\vphantom{\beta}m}}(\overline z)}= \cI_{\beta, m}(z). \label{eq:I_real}
\end{gather}
\end{subequations}
The case  $2m\in\mathbb{Z}$ is called degenerate, and then
\begin{equation}\cI_{\beta,-m}(z)= \Big(-\beta - m+\frac12\Big)_{2m}\, \cI_{\beta,m}(z).
\label{eq:I_sym}
\end{equation}

The function $\mathcal{K}_{\beta,m}$ is defined by
\begin{align}
\cK_{\beta,m}(z) &:=z^\beta\e^{-\frac{z}{2}} {}_2F_0\Big(\frac12+m-\beta,\frac12-m-\beta;-;-z^{-1}\Big) \label{eq:def_K} \\
&=\frac{\pi}{\sin(2\pi m)}\Big(
   -\frac{\cI_{\beta,m}(z)}{\Gamma(\frac12-m-\beta)}+
\frac{\cI_{\beta,-m}(z)}{\Gamma(\frac12+m-\beta)}\Big). \nonumber
 \end{align}
It satisfies
\begin{subequations}
\begin{gather}
\mathcal{K}_{\beta,-m}(z)=\mathcal{K}_{\beta,m}(z), \label{eq:K_sym} \\
\overline{\cK_{\overline \beta, \overline{\vphantom{\beta}m}}(\overline z)} = \cK_{\beta, m}(z). \label{eq:K_real}
\end{gather}
\end{subequations}

The Wronskian of $\cI_{\beta,m}(\cdot)$ and $\cK_{\beta,m}(\cdot)$ takes the form
\begin{equation}
\Wr(\cI_{\beta,m},\cK_{\beta,m})  = - \frac{ 1}{ \Gamma \Big ( \frac12 + m - \beta \Big ) }. \label{eq:def_gamma.}
\end{equation}
If $\frac12 + m - \beta \in - \mathbb N$, then the Wronskian vanishes and we have
\begin{align}
\cK_{\beta,m}(z) =& \e^{\i \pi (\frac{1}{2} + m - \beta)} \Gamma \left( \frac12 + m + \beta \right) \cI_{\beta,m}(z). 
\label{eq:Laguerre}
\end{align}
In this case functions in \eqref{eq:Laguerre} essentially coincide with Laguerre polynomials
\begin{equation}
\cI_{\pm (\frac{1}{2}+m+n),m}(z) = \frac{n!\;\!z^{\frac{1}{2}+m}\e^{\mp\frac{z}{2}}}{\Gamma(1+2m+n)}L_n^{(2m)}(\pm z).
\end{equation}

Asymptotics of $\mathcal{I}_{\beta,m}$ for small arguments are of the form
\begin{equation}
\mathcal{I}_{\beta,m}(z) =
\frac{ z^{\frac12+m} }{ \Gamma( 1 + 2m ) } \Big ( 1 - \frac{\beta}{ 1 + 2 m } z + O( z^2 ) \Big ) ,  \text{ if } m \neq -\frac12,-1,-\frac32,\dots,  \label{eq:equivIbm0.} \quad z \to 0 ,
\end{equation}

The function $\mathcal{K}_{\beta,m}$ satisfies, for $z \to 0$,
\begin{small}
\begin{equation}
\begin{array}{rll}
\mathcal{K}_{\beta,m}(z) =&
  z^{\frac12} \Big ( \frac{ \Gamma( -2m ) }{ \Gamma( \frac12 - m - \beta ) } z^m  + \frac{ \Gamma( 2 m ) }{ \Gamma( \frac12 + m - \beta ) } z^{-m} \Big ( 1 - \frac{ \beta }{ 1 - 2 m } z \Big ) \Big )&\vspace{0,1cm} \\&
\qquad\qquad  \qquad + O( |z|^{\frac32 + \mathrm{Re}(m)} )+O(|z|^{\frac52-\RE(m)})  &\text{ for } \mathrm{Re}(m)\in[0,1[ , \, m\neq 0,\frac12 \vspace{0,1cm} \\
\mathcal{K}_{\beta,0}(z) =&- \frac{ z^{\frac12} }{ \Gamma( \frac12 - \beta ) } \big ( \mathrm{ln}(z) + \psi( \frac12 - \beta ) + 2\lambda \big ) + O( |z|^{\frac32}\ln(z) ) , &\text{ for } m = 0 , \, \beta \notin \frac12 + \mathbb{N} \vspace{0,1cm} \\
\mathcal{K}_{\beta,0}(z) =&(\beta-\frac12)!(-1)^{\beta-\frac12} z^{\frac12} + O( |z|^{\frac32} ) , &\text{ for } m = 0 , \, \beta \in \frac12 + \mathbb{N} \vspace{0,1cm} \\
\mathcal{K}_{\beta,\frac12}(z) =&\frac{ 1 }{ \Gamma(  - \beta ) } \big ( -\frac1\beta +z\ln(z)+  z(\psi(1-\beta)+2\lambda-1+\frac{1}{2\beta})\big )&\vspace{0,1cm} \\&
\qquad\qquad\qquad+ O( z^2\ln(z)  ) , &\text{ for } m = \frac12 , \beta \notin \mathbb{N} \vspace{0,1cm} \\
\mathcal{K}_{0,\frac12}(z) =&1-\frac{z}{2}+O(z^2),&\text{ for }m=\frac12,\beta=0\vspace{0,1cm} \\
\mathcal{K}_{\beta,\frac12}(z) =&\beta! (-1)^{\beta - 1} z  + O( z^2  ) , &\text{ for } m = \frac12 , \beta \in \mathbb{N}^\times \vspace{0,1cm} \\
\mathcal{K}_{\beta,m}(z) =&  \frac{ \Gamma( 2m ) }{ \Gamma( \frac12 + m - \beta ) } z^{\frac12 - m} + O( |z|^{\frac32 - \mathrm{Re}(m) } ) , &\text{ for } \mathrm{Re}(m) \ge 1 .
 \end{array}
 \label{eq:K_small_arg}
\end{equation}
\end{small}
Here $\gamma$ denotes Euler's constant and $\psi$ is the digamma function.

Asymptotics for $|z| \to \infty$ are given by ($\epsilon > 0$):
\begin{subequations}
\begin{align}\label{eq:equivKbm_infty}
  \mathcal{K}_{\beta,m}(z)& = z^{\beta}\e^{-\frac{z}{2}} \big ( 1 + O( z^{-1} ) \big ) , \quad|\mathrm{arg}(z)|\leq\frac32\pi-\epsilon, \\
  \cI_{\beta, m }(z) & = \frac{z^{- \beta} \e^{\frac{z}{2}}}{\Gamma( \frac{1}{2} + m - \beta )} \left( 1 + O (z^{-1}) \right), \quad |\mathrm{arg}(z)| \leq \frac{\pi}{2} - \epsilon.
\end{align}
\end{subequations}

The analytic continuations of $\cK$, and more precisely the functions
        $z \mapsto \cK_{-\beta,m}(\e^{\pm\i\pi}z)$ are also solutions of 
        the Whittaker Equation \eqref{whi0}. One can define $\cK_{- \beta, m}(\e^{\pm i \pi} z )$ as the unique holomorphic function of $z \in \C \backslash ]- \infty, 0]$ which coincides with $\cK_{-\beta,m}(-z)$ on $\C_\mp$. Then
        \begin{align}
\cK_{-\beta,m}(\e^{\pm\i\pi}z)&=\frac{\pi}{\sin(2\pi m)}\Big( -\frac{\e^{\pm\i \pi(\frac12+m)}\cI_{\beta,m}(z)}{\Gamma(\frac12-m+\beta)}+
\frac{\e^{\pm\i\pi(\frac12-m)}\cI_{\beta,-m}(z)}{\Gamma(\frac12+m+\beta)}\Big).
\label{eq:K_cont}
\end{align}
$\big( \cK_{\beta,m}(z), \cK_{- \beta,m}(\e^{\pm \i \pi} z) \big)$ are linearly independent pairs of functions. In particular, $  \cI_{\beta,m}(z) $ can be expressed in terms of these functions:
\begin{align*}
  \cI_{\beta,m}(z) =\e^{\pm\i\pi\beta}\Big(\frac{\e^{\mp\i\pi(m-\frac12)}\cK_{\beta,m}(z)}{\Gamma(\frac12+m+\beta)} +\frac{\cK_{-\beta,m}(\e^{\pm\i\pi}z)}{\Gamma(\frac12+m-\beta)}
  \Big).
  \end{align*}

\subsection{Trigonometric-type Whittaker functions}

It is convenient, in parallel to \eqref{whi0},  to consider
the {\em trigonometric-type Whittaker equation}
\begin{align}
\Big ( -\partial_{z}^2 + \Big ( m^2 - \frac14 \Big ) \frac{1}{z^2} - \frac{ \beta }{ z } - \frac14\Big) g(z)=0. 
\label{Whittaker-trig}
\end{align}
It can be easily reduced to the hyperbolic-type Whittaker equation.

The function $\cJ_{\beta,m}$ is defined by the formula
\begin{equation}\label{Jbm-definition}
\cJ_{\beta,m}(z)= \e^{\mp \frac{\i \pi}{2} (m+ \frac12)} \cI_{\mp \i \beta,m} (\e^{\pm \frac{\i \pi}{2}}z).
\end{equation}
It may also be described without invoking analytic continuations beyond the principal branch:
\begin{equation}
\cJ_{\beta,m}(z) = 
\begin{cases}
\e^{\frac{\i \pi}{2} \left( \frac12+m \right)} \cI_{\i \beta, m}(- \i z ), & - \frac{\pi}{2} < \mathrm{arg}(z) < \pi, \\
\e^{-\frac{\i \pi}{2} \left( \frac12+m \right)} \cI_{-\i \beta, m}(\i z ), & - \pi < \mathrm{arg}(z) < \frac{\pi}{2}.
\end{cases}
\end{equation}
The two expression agree for $|\mathrm{arg}(z)|< \frac{\pi}{2}$, by \eqref{eq_miracle}. Combined with \eqref{eq:I_real} this implies
\begin{equation}
\overline{\cJ_{\overline \beta, \overline{\vphantom{\beta}m}}}(\overline z) = \cJ_{\beta, m}(z).
\end{equation}

We also have a pair of functions $\cH_{\beta,m}^\pm$ defined by
\begin{equation}\label{Hpm-definition}
\cH_{\beta,m}^\pm(z) = \e^{\mp \i\frac{\pi}{2}\naw{\frac{1}{2}+m}}\cK_{\pm\i\beta,m}( \mp \i z)
\end{equation}
initially for $\RE(z) >0$ and extended analytically to $z \in \C \backslash ] - \infty, 0 ]$. By \eqref{eq:K_real}, they satisfy
\begin{equation}
\overline{\cH^{\pm}_{\overline \beta, \overline{\vphantom{\beta}m}}(\overline z)} = \cH^{\mp}_{\beta, m}(z). \label{eq:H_conj}
\end{equation}

The following connection formula holds:
\begin{equation}
 \cJ_{\beta,m}(z)= \e^{-\pi\beta}\Big(\frac{\cH_{\beta,m}^+(z)}{\Gamma\naw{\frac{1}{2}+m+\i\beta}} + \frac{\cH_{\beta,m}^-(z)}{\Gamma\naw{\frac{1}{2}+m-\i\beta}}\Big). 
 \label{eq_missing}
\end{equation}

For the behaviour around $\infty$, we have for $x>0$, $x\to\infty$
\begin{equation}\label{Hpm-around-infinity}
\cH_{\beta,m}^\pm(x) \sim \e^{\mp\i\frac{\pi}{2}(\frac{1}{2}+m)}\e^{\frac{\pi\beta}{2}}x^{\pm\i\beta}\;\!\e^{\pm\i\frac{x}{2}}\big(1 + O(x^{-1})\big).
\end{equation}

\subsection{Recurrence relations} \label{recur}

Whittaker functions satisfy several recurrence relations.
There are 6 basic ones, which we quote after Appendix A5 of
\cite{DeFaNgRi20_01}.
\begin{subequations}\label{rei}\begin{align}
\Big(\sqrt{z}\partial_z+\frac{-\frac12-m}{\sqrt{z}}-\frac {\sqrt{z}}{2}\Big)\cI_{\beta,m}(z)
&= \Big(-\frac12-m-\beta\Big)\cI_{\beta+\frac12,m+\frac12}(z),\label{rei1}\\
\Big(\sqrt{z}\partial_z+\frac{-\frac12+m}{\sqrt{z}}+\frac {\sqrt{z}}{2}\Big)\cI_{\beta,m}(z)
&=\cI_{\beta-\frac12,m-\frac12}(z), \label{rei2}\\
\Big(\sqrt{z}\partial_z+\frac{-\frac12+m}{\sqrt{z}}-\frac {\sqrt{z}}{2}\Big)\cI_{\beta,m}(z)
                      &=\cI_{\beta+\frac12,m-\frac12}(z), \label{rei3}
                      \\
\Big(\sqrt{z}\partial_z+\frac{-\frac12-m}{\sqrt{z}}+\frac {\sqrt{z}}{2}\Big)\cI_{\beta,m}(z)
&=\Big(\frac12+m-\beta\Big)\cI_{\beta-\frac12,m+\frac12}(z), \label{rei4}\\
\Big(z\partial_z+\beta-\frac{z}{2}\Big)\cI_{\beta,m}(z)
&=\Big(\frac12+m+\beta\Big)\cI_{\beta+1,m}(z),\\
\Big(z\partial_z-\beta+\frac{z}{2}\Big)\cI_{\beta,m}(z)
&=\Big(\frac12+m-\beta\Big)\cI_{\beta-1,m}(z);
\end{align}
\end{subequations}
\begin{subequations}\label{rek}
\begin{align}
\Big(\sqrt{z}\partial_z+\frac{-\frac12-m}{\sqrt{z}}-\frac {\sqrt{z}}{2}\Big)
\cK_{\beta,m}(z)
&= -\cK_{\beta+\frac12,m+\frac12}(z),  \label{rek1}\\
\Big(\sqrt{z}\partial_z+\frac{-\frac12+m}{\sqrt{z}}+\frac {\sqrt{z}}{2}\Big)
\cK_{\beta,m}(z)
&= \Big(-\frac12+m+\beta\Big)    \cK_{\beta-\frac12,m-\frac12}(z),  \label{rek2}\\
\Big(\sqrt{z}\partial_z+\frac{-\frac12+m}{\sqrt{z}}-\frac {\sqrt{z}}{2}\Big)
\cK_{\beta,m}(z)
&= -\cK_{\beta+\frac12,m-\frac12}(z),  \label{rek3}\\
\Big(\sqrt{z}\partial_z+\frac{-\frac12-m}{\sqrt{z}}+\frac {\sqrt{z}}{2}\Big)
\cK_{\beta,m}(z)
&= \Big(-\frac12-m+\beta\Big)\cK_{\beta-\frac12,m+\frac12}(z),  \label{rek4}\\
\Big(z\partial_z+\beta-\frac{z}{2}\Big)\cK_{\beta,m}(z)
&= -\cK_{\beta+1,m}(z),\\
\Big(z\partial_z-\beta+\frac{z}{2}\Big)\cK_{\beta,m}(z)
&= \Big(\frac12+m-\beta\Big)    \Big(\frac12-m-\beta\Big)\cK_{\beta-1,m}(z).
\end{align}\end{subequations}
They have an interesting algebraic interpretation---they correspond to
the roots of the Lie algebra of generalized symmetries of the heat
equation in 2 dimensions, see \cite{Der} (where they are presented
using confluent functions, which as we know are equivalent to
Whittaker functions).

These recurrence relations involve 1st order differentiation, so it is
tempting to expect that they are closely related  to the
Dirac-Coulomb Hamiltonian. It turns out, however, that the relationship
is not direct. By easy algebraic manipulations involving
\eqref{rei1}-\eqref{rei4} and \eqref{rek1}-\eqref{rek4},
we  derive an additional pair of
recurrence relations described in the following proposition.
In some sense,  \eqref{rei} and \eqref{rek}  are ``lower order'' than 
\eqref{reri} and  \eqref{rerk}. 
In fact, in \eqref{rei} and \eqref{rek} the parameters $\mu,\beta$
appear only in 0th order terms, whereas in \eqref{reri} and  \eqref{rerk} 
the differential operator is multiplied by $\mu$.
\begin{proposition}
\begin{subequations}\label{reri}\begin{align}
  \Big(2\mu\partial_x+\frac{2\mu^2}{x}-\beta\Big)\cI_{\beta,\mu+\frac12}(x)&=\cI_{\beta,\mu-\frac12}(x),\\
                      \Big(2\mu\partial_x-\frac{2\mu^2}{x}+\beta\Big)\cI_{\beta,\mu-\frac12}(x)&=(\mu^2-\beta^2)\cI_{\beta,\mu+\frac12}(x);\end{align}
\end{subequations}\begin{subequations}\label{rerk}\begin{align}
  \Big(2\mu\partial_x+\frac{2\mu^2}{x}-\beta\Big)\cK_{\beta,\mu+\frac12}(x)&=-(\mu+\beta)\cK_{\beta,\mu-\frac12}(x),\\
  \Big(2\mu\partial_x-\frac{2\mu^2}{x}+\beta\Big)\cK_{\beta,\mu-\frac12}(x)&=(-\mu+\beta)\cK_{\beta,\mu+\frac12}(x).
\end{align}\end{subequations}
\end{proposition}

\proof
We first rewrite \eqref{rei1}-\eqref{rei4} and
\eqref{rek1}-\eqref{rek4} as follows:
\begin{subequations}\label{rei.}\begin{align}
\Big(\sqrt{z}\partial_z-\frac{\mu}{\sqrt{z}}-\frac {\sqrt{z}}{2}\Big)\cI_{\beta,\mu-\frac12}(z)
&= \Big(-\mu-\beta\Big)\cI_{\beta+\frac12,\mu}(z),\label{rei1.}\\
\Big(\sqrt{z}\partial_z+\frac{\mu}{\sqrt{z}}+\frac {\sqrt{z}}{2}\Big)\cI_{\beta,\mu+\frac12}(z)
&=\cI_{\beta-\frac12,\mu}(z), \label{rei2.}\\
\Big(\sqrt{z}\partial_z+\frac{\mu}{\sqrt{z}}-\frac {\sqrt{z}}{2}\Big)\cI_{\beta,\mu+\frac12}(z)
                      &=\cI_{\beta+\frac12,\mu}(z), \label{rei3.}
                      \\
\Big(\sqrt{z}\partial_z-\frac{\mu}{\sqrt{z}}+\frac {\sqrt{z}}{2}\Big)\cI_{\beta,\mu-\frac12}(z)
                                  &=\Big(\mu-\beta\Big)\cI_{\beta-\frac12,\mu}(z); \label{rei4.}
                                  \end{align}\end{subequations}\begin{subequations}\label{rekrek}
\begin{align}
\Big(\sqrt{z}\partial_z-\frac{\mu}{\sqrt{z}}-\frac {\sqrt{z}}{2}\Big)
\cK_{\beta,\mu-\frac12}(z)
&= -\cK_{\beta+\frac12,\mu}(z),  \label{rek1.}\\
\Big(\sqrt{z}\partial_z+\frac{\mu}{\sqrt{z}}+\frac {\sqrt{z}}{2}\Big)
\cK_{\beta,\mu+\frac12}(z)
&= \Big(\mu+\beta\Big)    \cK_{\beta-\frac12,\mu}(z),  \label{rek2.}\\
\Big(\sqrt{z}\partial_z+\frac{\mu}{\sqrt{z}}-\frac {\sqrt{z}}{2}\Big)
\cK_{\beta,\mu+\frac12}(z)
&= -\cK_{\beta+\frac12,\mu}(z),  \label{rek3.}\\
\Big(\sqrt{z}\partial_z-\frac{\mu}{\sqrt{z}}+\frac {\sqrt{z}}{2}\Big)
\cK_{\beta,\mu-\frac12}(z)
&= \Big(-\mu+\beta\Big)\cK_{\beta-\frac12,\mu}(z).  \label{rek4.}
\end{align}\end{subequations}

Then we compute
\begin{align}
  &\frac{1}{\sqrt{z}}\big(-\eqref{rei1.}+(-\mu+\beta)\eqref{rei2.}-(\mu+\beta)\eqref{rei3.}+\eqref{rei4.}\big),\\
  &\frac{1}{\sqrt{z}}\big((\mu-\beta)\eqref{rei1.}+(-\mu^2+\beta^2)\eqref{rei2.}+(\mu^2-\beta^2)\eqref{rei3.}+(\mu+\beta)\eqref{rei4.}\big),\\[2ex]
  &\frac{1}{\sqrt{z}}\big(-(\mu+
    \beta)\eqref{rek1.}+(\mu-\beta)\eqref{rek2.}+(\mu+\beta)\eqref{rek3.}+(\mu+\beta)\eqref{rek4.}\big),\\
    &\frac{1}{\sqrt{z}}\big((\mu-
    \beta)\eqref{rek1.}+(\mu-\beta)\eqref{rek2.}-(\mu-\beta)\eqref{rek3.}+(\mu+\beta)\eqref{rek4.}\big).
  \end{align}
\qed

\eqref{reri} and \eqref{rerk} are closely related to the
Dirac-Coulomb Hamiltonian. To see this relation let
us introduce $\omega$ satisfying $\omega^2=\mu^2-\beta^2$. Then
\eqref{reri} and  \eqref{rerk} can be rewritten in the following form:
\begin{align}
  0=&2\mu\partial_x\big(\cI_{\beta,\mu-\frac12}(x)+\i\omega\cI_{\beta,\mu+\frac12}(x)\big)\\\notag&+
\Big(-\frac{2\mu^2}{x}+\beta-\i\omega\Big) 
  \big(\cI_{\beta,\mu-\frac12}(x)-\i\omega\cI_{\beta,\mu+\frac12}(x)\big),\\
0=&    2\mu\partial_x\big(\cI_{\beta,\mu-\frac12}(x)-\i\omega\cI_{\beta,\mu+\frac12}(x)\big)\\\notag&+
\Big(-\frac{2\mu^2}{x}+\beta+\i\omega\Big) 
  \big(\cI_{\beta,\mu-\frac12}(x)+\i\omega\cI_{\beta,\mu+\frac12}(x)\big);\\[3ex]
  0=&2\mu\partial_x\big((\mu+\beta)\cK_{\beta,\mu-\frac12}(x)-\i\omega\cK_{\beta,\mu+\frac12}(x)\big)\\\notag&+
\Big(-\frac{2\mu^2}{x}+\beta-\i\omega\Big) 
  \big((\mu+\beta)\cK_{\beta,\mu-\frac12}(x)+\i\omega\cK_{\beta,\mu+\frac12}(x)\big),\\
0=&    2\mu\partial_x\big((\mu+\beta)\cK_{\beta,\mu-\frac12}(x)+\i\omega\cK_{\beta,\mu+\frac12}(x)\big)\\\notag&+
\Big(-\frac{2\mu^2}{x}+\beta+\i\omega\Big) 
  \big((\mu+\beta)\cK_{\beta,\mu-\frac12}(x)-\i\omega\cK_{\beta,\mu+\frac12}(x)\big).
                                                                          \end{align}
The eigenvalue equations for $\xi_p^\pm$ and $\zeta_p^\pm$ follow
directly from these identities.

\subsection{Integral transforms}

Let us  compute a useful integral transform of the confluent function:

\begin{lemma}
Assuming $\RE(b) >0$ and $|\RE(w)| < \RE(z)$, one has
  \begin{align}\label{id1}
    &    \int_0^\infty x^{b-1}\e^{-zx} {}_1 \mathbf F_1(a;c;wx)\D x=
                                               z^{-b}\Gamma(b) {}_2 \mathbf F_1(a,b;c;z^{-1}w). 
\end{align}
Furthermore, if $\RE(b) >0$, $\RE(b+1-c)>0$, $\RE(z)>0$, $w, z^{-1}w \notin ] - \infty, 0]$, then
\begin{align} \label{id2}
    &\int_0^\infty
    x^{b-1}\e^{-zx}{}_2F_0\big(a,a+1-c;-;-(wx)^{-1}\big)(wx)^{-a}\D x \\=&
                                                                     z^{-b} \Gamma(b)\Gamma(1+b-c) {}_2 \mathbf F_1(a,b;a+b+1-c;1-z^{-1}w).
 \notag \end{align}
\end{lemma}

\proof
We expand the confluent function in a power series and integrate term
by term:
{ \small
\begin{align}
  \int_0^\infty\sum_{n=0}^\infty \frac{(a)_n}{\Gamma(c+n)}\e^{-zx}x^{b+n-1}w^n\D
  x&=\sum_{n=0}^\infty \frac{\Gamma(b+n)(a)_n}{\Gamma(c+n)}\frac{w^n}{z^{n+b}} 
   =z^{-b}\Gamma(b)  \sum_{n=0}^\infty \frac{(a)_n(b)_n}{\Gamma(c+n)} \left( \frac{w}{z} \right)^n.
     \end{align}
     }
This proves the first identity under additional assumption $|w|<\RE(z)$. The integrand can be
majorized by an integrable function for $|\RE(w)|<\RE(z)$. Therefore, we
can extend the identity by analytic continuation to this domain,
yielding \eqref{id1}.

\eqref{id1}, \eqref{id3} and \eqref{id4} and analytic
continuation imply \eqref{id2}.
\qed

The following identity, valid for $v >0$, $\RE(\epsilon) > 0$, $\RE(m+1 - \i s)>0$, follows from \eqref{id1}:
\begin{align}
  & \int_0^{\infty} \e^{-\epsilon x} x^{- \frac{1}{2} - \i s} \cJ_{\beta, m} (vx) \D x  \label{eq:Mellin_J2} \\
=&  v^{m + \frac{1}{2}} \Big( \epsilon\pm\i\frac{v}{2}\Big)^{-m -1 + \i s} \Gamma(m+1 - \i s)  {}_2 \mathbf{F}_1 \left( m+ \frac{1}{2}  \pm \i \beta, m+1 - \i s ; 2m+1 ; \frac{v}{\frac{v}{2}\mp \i \epsilon} \right). \notag 
\end{align}

\begin{proposition}\label{melino}
Let $v>0$, $\RE(m) > -1$. Then $x\mapsto\e^{-0x}\cJ_{\beta,m}(vx):= \lim\limits_{\epsilon\downarrow 0}\e^{-\epsilon 
  x}\cJ_{\beta,m}(vx)$ is a~tempered 
distribution on $\R_+$ with the Mellin transform
\begin{align}
 & \int_0^{\infty} \e^{-0 x} x^{- \frac{1}{2} - \i s}
 \cJ_{\beta, m} (vx)\D x
     \label{eq:Mellin_J2a} \\
=&  v^{- \frac{1}{2}+\i s}2^{m+1-\i s} (\pm\i)^{-m -1 + \i s} \Gamma(m+1 - \i s)  {}_2 \mathbf{F} _1 \left( m+ \frac{1}{2}  \pm \i \beta, m+1 - \i s ; 2m+1 ; 2\pm\i0\right), \notag 
\end{align}
which is bounded by $c_{\beta,m} \left( \e^{- \frac{\pi}{2} (|s|+s)} |s|^{\IM(\beta)} + \e^{- \frac{\pi}{2} (|s|-s)} |s|^{- \IM(\beta)} \right)$, with $c_{\beta,m}$ a locally bounded function of $\beta,m$.

Similarly, for any $\mu$ (including $\mu=0$), $x\mapsto\e^{-0x}\frac{1}{\mu}\big(\cJ_{\beta,\mu-\frac12}(vx)+ \beta \cJ_{\beta,\mu+\frac12}(vx)\big)
  $ is a~tempered 
distribution on $\R_+$, whose Mellin transform  can be computed
from \eqref{eq:Mellin_J2a} and is bounded by $c_{\beta, \mu} \left( \e^{- \frac{\pi}{2} (|s|+s)} |s|^{\IM(\beta)} + \e^{- \frac{\pi}{2} (|s|-s)} |s|^{- \IM(\beta)} \right)$, with $c_{\beta,\mu}$ a locally bounded function of $\beta,\mu$.
\end{proposition}

\proof
We use the criterion from Lemma \ref{Mellin_lim}. Let
$f_{\epsilon}(x)=\e^{-\epsilon x}\cJ_{\beta,m}(vx) $. Then $t \mapsto
\e^{\frac{t}{2}} f_{\epsilon}(\e^t)$ is smooth and vanishes exponentially for $t \to - \infty$ and superexponentially for $t \to \infty$. In particular $f_{\epsilon}$ is a tempered distribution on $\R_+$. Its Mellin transform is given by the absolutely convergent integral \eqref{eq:Mellin_J2}. It is a smooth function with smooth pointwise limit $\epsilon \to 0$. Required bounds follow from \eqref{hypergeometric_asymptotic}. This completes the proof of the first part.

Next, we compute
\begin{footnotesize}
\begin{align}
 & \int_0^{\infty} \e^{-\epsilon x} x^{- \frac{1}{2} - \i s}
 \frac{1}{\mu} \left( \cJ_{\beta, \mu - \frac12} (vx) + \beta \cJ_{\beta, \mu + \frac12}(vx) \right) \D x \\
= & v^\mu \left( \epsilon \pm \i \frac{v}{2} \right)^{- \mu - \frac12 + \i s} \Gamma \left( \frac{1}{2} + \mu - \i s \right)  \notag \\ 
\times & \frac{1}{\mu} \Bigg( {}_2 \mathbf F_1 \left( \mu + \i \beta, \frac{1}{2}+ \mu - \i s; 2 \mu ; \frac{v}{\frac{v}{2}- \i \epsilon} \right) - \i \beta \left( \frac{1}{2} + \mu - \i s \right) \frac{v}{\frac{v}{2}- \i \epsilon} {}_2 \mathbf F_1 \left( 1+ \mu + \i \beta, \frac{3}{2}+ \mu - \i s; 2 \mu+2 ; \frac{v}{\frac{v}{2}- \i \epsilon} \right) \Bigg). \notag
\end{align}
\end{footnotesize}
Expression in the last line is nonsingular for $\mu \to 0$ on the
account of \eqref{eq:2F1c0}. Bounds on the growth at infinity are derived as in the first case.
\qed

{\bf Acknowledgement.}
The work of J.D. 
 was supported by National Science Center (Poland) under the
    grant UMO-2019/35/B/ST1/01651.

\end{document}